\documentclass[a4paper,11pt]{article}

\usepackage{lmodern} %change document font
\usepackage{multicol}
\usepackage{blindtext}
\usepackage{float}
\usepackage[T1]{fontenc} %256 bit font encoding
\usepackage[utf8]{inputenc}
\usepackage[backend=biber,doi=true,isbn=false,url=false,maxbibnames=5,sorting=none,style=numeric-comp]{biblatex}
\AtBeginBibliography{\small} 
\defbibfilter{appendixOnlyFilter}{
	segment=1 % Segment 1 will be chosen to be the one in appendix
	and not segment=0 % Default segment is 0
}

\usepackage{mathtools} %for, e.g., \DeclareMathOperator
\usepackage{amsfonts}
\usepackage{amsmath} %for spacing in equations and aligns
\usepackage{amsthm}
\usepackage{dsfont}
\usepackage{physics}
\usepackage[linesnumbered,ruled]{algorithm2e}
\usepackage{enumerate}
\usepackage{enumitem}
\usepackage[font=small,labelfont=bf]{caption}
\usepackage{tikz}
\usetikzlibrary{hobby} %to create irregular areas
\usepackage{subcaption}
\usepackage{graphicx}
\graphicspath{{plots/}} %Setting the graphicspath
\usepackage{placeins} %to flush floats (e.g. before new section) with \FloatBarrier 
\usepackage[normalem]{ulem} % to strikeout comments, normalem needed to bring emph back to normal
\usepackage{array} % to justify text in a tabular
\newcolumntype{C}[1]{>{\centering\let\newline\\\arraybackslash\hspace{0pt}}m{#1}}
\usepackage{color}
\newcolumntype{L}[1]{>{\raggedright\let\newline\\\arraybackslash\hspace{0pt}}m{#1}}
\usepackage[top=2.75cm,bottom=2.75cm,left=2.5cm, right=2.5cm]{geometry}
\usepackage{changepage} %to change margins for part of the text with \begin{adjustwidth}
\usepackage{multirow}%for multirow in tables
\allowdisplaybreaks
\usepackage[colorlinks=true,linkcolor=blue,citecolor=citegreen]{hyperref}
\usepackage[affil-it]{authblk}
\usepackage[nameinlink,capitalise]{cleveref}
\usepackage[capitalise]{cleveref}
\Crefname{lemma}{Lemma}{Lemmas}
\Crefname{proposition}{Proposition}{Propositions}
\Crefname{definition}{Definition}{Definitions}
\Crefname{theorem}{Theorem}{Theorems}
\Crefname{conjecture}{Conjecture}{Conjectures}
\Crefname{corollary}{Corollary}{Corollaries}
\Crefname{example}{Example}{Examples}
\Crefname{section}{Section}{Sections}
\Crefname{subsection}{subsection}{subsections}
\Crefname{appendix}{Appendix}{Appendices}
\Crefname{figure}{Fig.}{Figs.}
\Crefname{equation}{eq.}{eqs.}
\Crefname{table}{Table}{Tables}
\Crefname{item}{Property}{Properties}
\Crefname{remark}{Remark}{Remarks}

\newtheorem{theorem}{Theorem}[section]
\newtheorem{lemma}[theorem]{Lemma}

\renewcommand{\ket}[1]{| #1 \rangle}

\newcommand{\1}{\mathds{1}}

\newcommand{\RR}{\mathds{R}}

\newcommand{\identity}{\mathds{1}}

\newcommand{\Eideal}{E^{\mathrm{ideal}}}

\definecolor{citegreen}{RGB}{0,165,0}%citation color

\renewcommand{\paragraph}[1]{\textbf{\textit{#1}}}

\usepackage{authblk}

\title{Robust Lindbladian Estimation for Quantum Dynamics}
\author{Yinchen Liu}
\author{James R. Seddon}
\author{Tamara Kohler}
\author{Emilio Onorati}
\author{Toby Cubitt}

\affil{Phasecraft Ltd., London, UK}

\begin{document}

\date{10th July 2025}
\maketitle
\abstract{
We revisit the problem of fitting Lindbladian models to the outputs of quantum process tomography. A sequence of prior theoretical works approached the problem by considering whether there exists a Lindbladian generator close to a matrix logarithm of the tomographically estimated transfer matrix. This technique must take into account the non-uniqueness of the matrix logarithm, so that in general multiple branches of the logarithm must be checked. In contrast, all practical demonstrations of Lindbladian fitting on real experimental data have to our knowledge eschewed logarithm search, instead adopting direct numerical optimisation or ad-hoc approaches tailored to a particular experimental realisation. In our work, we introduce algorithmic improvements to logarithm search, demonstrating that it can be applied in practice to settings relevant for current quantum computing hardware. We additionally augment the task of Lindbladian fitting with techniques from gate set tomography to improve robustness against state preparation and measurement (SPAM) errors, which can otherwise obfuscate estimates of the model underlying the process of interest. We benchmark our techniques extensively using simulated tomographic data employing a range of realistic error models, before demonstrating their application to tomographic data collected from real superconducting-qubit hardware.}

\tableofcontents

\section{Introduction}
Quantum computers promise potential applications in the simulation of materials \cite{Feynman:1982aa,Lloyd_1996}, quantum chemistry ~\cite{QChemMcArdle_2020}, optimisation \cite{Moll_2018}, and fundamental physics \cite{Preskill_2018}. While proof-of-concept experiments \cite{Arute:2019aa,advantageZhong2020,IBMPEA} have increasingly provided evidence that real quantum hardware can yield results that are challenging to reproduce classically, to date there has been no demonstration of a real-world practical application of quantum computers that competes in accuracy and precision with the best classical methods, at a scale large enough that it could not be brute-force simulated by a conventional computer. The major roadblock to such large-scale demonstrations is the sensitivity of quantum devices to noise. This issue has been studied for a long time and is in principle resolved by the theory of quantum error correction~\cite{Shor1995EC,GottesmanThesis1997,kitaev1997quantum} and fault tolerance \cite{Shor1996FT,knill1998resilient,Gottesman1998FT,Preskill_1998FT,aharonov_2008,Campbell2017FT}. However, fault-tolerant quantum computing is known to require very large overheads for meaningful applications, both in terms of the number of qubits and quantum operations required~\cite{Kivlichan2020improvedfault,Lee2021FeMoCo}, and the huge throughput of classical data needed for the online decoding of syndrome measurements~\cite{decodingSkoric2023-ww}. The upshot is that fully scalable fault tolerance requires technical capabilities far beyond that presently available on existing quantum hardware.

In the meantime, the question of whether quantum advantage can be gleaned from near-term noisy devices remains open. Since noise in quantum devices is an inevitability, any such demonstration must necessarily involve the successful deployment of techniques to suppress or mitigate errors~\cite{dynamicaldecoupling2014,Temme_mitigation_2017,Ying_ZNE_2017,endo_practical_2018,mcclean_decoding_2020,bravyi_mitigating_2020,sun_mitigating_2020,
maciejewski_mitigation_2020,Czarnik_CDR_2021,cai_multi-exponential_2021,Zhang_2022coherentError,Cai2023errormitigation,IBMPEA}. In the absence of generic error correction, it is likely that a detailed understanding of noise processes on real hardware will be invaluable in tailoring error mitigation and suppression techniques so as to maximise the performance of the device at hand. Indeed, a number of leading error mitigation methods depend on a noise characterisation step as an integral part of the protocol ~\cite{bravyi_mitigating_2020,Czarnik_CDR_2021,IBMPEC,IBMPEA}. A large collection of noise characterisation methods exist, ranging from fast benchmarking protocols \cite{Knill_2008,Magesan_2011,Wallman_2018,helsen2020general} that provide figures of merit for the current operating status of a quantum device, to more costly tomographic methods that deliver fine-grained qualitative and quantitative information about the noise processes afflicting the device at the level of individual qubits and gates~\cite{QPTChuang_1997,Mohseni_2008,Blume:2013,GreenbaumGST,Nielsen2021gatesettomography}. Quantum process tomography (QPT)~\cite{QPTChuang_1997,Poyatos_QPT_1997} is a standard technique that in principle yields a full description of the completely positive trace-preserving (CPTP) map representing a noisy quantum channel. The quantum and classical computing cost of process tomography grows exponentially with system size, so it is tractable only for characterising few-qubit processes. It can nevertheless be a useful diagnostic tool on real devices, since typically the elementary operations available on present devices themselves act on at most two qubits at a time. One obstacle to the usefulness of QPT is that it is not robust against state preparation and measurement (SPAM) errors, since it implicitly assumes that the preparations and measurements available to the user are ideal. Gate set tomography (GST)~\cite{Blume:2013,GreenbaumGST,Nielsen2021gatesettomography,pyGSTi,Vinas2025} is an extension of QPT designed to deal with this problem. Rather than analysing individual gates in isolation, it attempts to self-consistently estimate SPAM errors and gate errors at the same time. Ultimately, both QPT and GST output a set of numerically represented CPTP map estimates. 

Since these raw tomographic snapshots can be difficult to interpret and act on, it is often more useful to build a model of the underlying physical process that generates the CPTP map. 
One approach tries to understand the channel as a solution to a quantum master equation~\cite{BreuerPettrucione}. In particular, it is often assumed that noisy processes on quantum devices are well approximated by memoryless, or Markovian, dynamics. Extracting the generator of a Markovian evolution -- which we refer to as the \emph{Lindbladian} --  delivers qualitative information about the noise that is otherwise obscured when considering the process in the channel picture.

A number of methods have been proposed and demonstrated for fitting Lindbladian models to full or partial tomographic data \cite{NMRestimationChilds2001, Quantum_optics_D_Ariano_2002, boulant2003robust, Markdynamics08,CEW09, Postmarkovian_Zhang_2022, Samach_2022, IBMPEC,dobrynin2024}. 
One of these techniques starts by taking the logarithm of the transfer matrix representation of the channel , first studied formally for non-degenerate transfer matrices in~\cite{Markdynamics08,CEW09} and extended to the degenerate case in \cite{Onorati2023fittingquantumnoise,onoratitimedependentmarkovian}. These methods must contend with the fact that the matrix logarithm is non-unique, and is sensitive to statistical errors in estimating the tomographic data. The initial theoretical works~\cite{Markdynamics08,CEW09} do not address SPAM or statistical errors at all. The direct precursor to the present work~\cite{Onorati2023fittingquantumnoise} proposed an implementable algorithm for fitting Lindbladian models to noisy tomographic data. However, it did not actively address the problem of SPAM, and the run-time of the algorithm as stated was extremely long for the two-qubit case, making it impractical for use on real devices. 

In this paper, we build on the method of \cite{Onorati2023fittingquantumnoise}, introducing algorithmic improvements to reduce the run-time and strengthen the robustness to statistical perturbation in practical cases. 
Additionally, we use ideas from gate set tomography \cite{Blume:2013,Nielsen2021gatesettomography,pyGSTi} to show that Lindbladian fitting techniques can also be made robust against realistic SPAM errors. We demonstrate the application of this approach to tomographic data collected from real noisy quantum computing hardware accessed via the cloud~\cite{IBMQ}. Our methods represent effective strategies to learn noise in real-world situations, requiring relatively few resources, both in terms of quantum experiments and classical post-processing time. They complement other Lindbladian estimation approaches~\cite{Samach_2022,StilckfrancaMarkovich,Olsacher_2025,pyGSTi} and can be used to cross-validate AI-based heuristic algorithms to study the interplay between Markovian and non-Markovian dynamics~\cite{Cemin24,datadrivenlatentdynamics}.

\subsection{Related Works}\label{sec:related_works}

Previous work on the task of Lindbladian estimation can broadly be divided into three strands: (i)~logarithm search methods, (ii)~direct numerical optimisation methods, and (iii)~ad-hoc methods tailored to particular experimental scenarios. 
Logarithm search methods \cite{Markdynamics08,CEW09,Onorati2023fittingquantumnoise} attempt to certify or estimate a Lindbladian model by taking the matrix logarithm of a transfer matrix representation of the target channel and attempting to fit a Lindbladian. These techniques must account for the fact that the complex matrix logarithm is not unique. Prior work in this strand has been theoretical, benchmarking on synthetic data with no testing on real hardware. In particular, the effects of statistical and SPAM errors on this technique have not been adequately explored or taken into account.

Prior practical demonstrations of Lindbladian estimation from noisy experimental data have tended to fall under the heading of direct numerical optimisation. The general approach is to define the problem in terms of an objective function that compares a parametrised Markovian channel with the tomographic data and pass this to an off-the-shelf numerical optimisation solver. 
A key observation is that the set of all $d^2\times d^2$ Lindbladian matrices can be exactly parametrised by $d^2(d^2-1)$ real variables. Furthermore, gradients of the objective function can often be obtained mechanically via automatic differentiation, and in such cases, gradient-based solvers can be used. Finally, a good initial guess required by the solvers is often available based on the experimenter's prior knowledge about the input channel. Prior works in this direction include \cite{boulant2003robust,Samach_2022,dobrynin2024,Pastori2022floquetLiouvillian,pyGSTi}. 

Approaches to Lindbladian estimation that fall outside these general methods tend to rely on some assumption about the system being investigated, so that the Lindbladian model is restricted in some way, and the number of parameters to estimate is reduced. Such techniques have been applied in NMR~\cite{NMRestimationChilds2001}, quantum optics~\cite{Quantum_optics_D_Ariano_2002}, and superconducting-qubit~\cite{Postmarkovian_Zhang_2022} settings for small systems. A method for studying larger systems under the assumption that the Markovian noise process is generated solely by local Pauli operators has also been studied ~\cite{IBMPEC,IBMPEA}. This method is scalable, but the assumption of a Pauli-Lindblad model is rather stringent. Finally, we note that the gate set tomography software package \emph{pyGSTi} \cite{pyGSTi} includes functionality to fit custom Lindblad models in the form of channels acting in series with the target gate set, using numerical optimisation methods different to those considered in the present work.

A more detailed discussion of prior work can be found in \Cref{sec:prior work}.

\subsection{Main Results}\label{sec:mainresults}
In the following sections, we study and give practical algorithms for Lindbladian fitting in two settings: (i) characterisation of an individual noisy process, neglecting SPAM errors, and  (ii)  simultaneous and self-consistent characterisation of a set of processes, allowing for arbitrary SPAM errors. 
We do not try to tackle these problems in full generality, but rather focus on input instances relevant to near-term quantum computing devices. More specifically, our goal is to fit Lindbladian noise models to noisy $2$-qubit quantum logic gates, where the noise is not arbitrarily strong. We build on the logarithmic search methods developed in \cite{Markdynamics08,CEW09,Onorati2023fittingquantumnoise} in three main respects.

\begin{enumerate}
\item We clarify the conditions under which good solutions can already be found by running convex optimisation over a few logarithm branches of the transfer matrix and show that this includes classes of practically relevant quantum channels larger than previously thought. In particular, this includes cases whose transfer matrix has degenerate spectrum. Instead, our new conditions merely require that the transfer matrix does not have eigenvalues close to being real negative. We call this simple algorithm the \emph{Convex Solve} method.

\item To tackle the case where some continuous degrees of freedom cannot be avoided, we develop a novel algorithm for fitting a Lindbladian to a noisy tomographic snapshot. The method leverages the fact that in practice, we can start with an initial best guess for the Markovian channel -- namely, the ideal gate that we program to be implemented by the quantum computer. At an intuitive level, we use the eigendecomposition of the best-guess Lindbladian to guide the construction of the matrix logarithm, and iteratively improve this best guess by projecting back to the set of Lindbladians. This leads to run-time reduced by orders of magnitude compared to~\cite{Onorati2023fittingquantumnoise} and makes the problem tractable for cases that were not for the algorithm designed therein. We call this the \emph{Alternating Projections} method.

\item We propose a protocol that synthesises Lindbladian fitting and gate set tomography~\cite{Nielsen2021gatesettomography} techniques by self-consistently alternating between a SPAM estimation step and a Lindbladian fitting step. For the Lindbladian fitting step, we use the Convex Solve and Alternative Projections methods from point~1 and~2 as subroutines in our implementation, but we note that the template for the protocol is generic and any method for fitting a Lindbladian to a tomographic channel estimate could be substituted. We call this algorithm \emph{Gate Set Flip-Flop}.
\end{enumerate}

\noindent The remainder of the paper is structured as follows. We provide important technical background in~\Cref{sec:background}, covering Markovian channels and Lindbladians, quantum process tomography, and gate set tomography. 
We describe our novel algorithms in detail in~\Cref{sec:algorithms}: we present in \Cref{subsec:TTE} our Convex Solve and Alternating Projections Lindbladian extraction algorithms and in \Cref{subsec:flip-flop} an augmented version incorporating gate set tomography. In \Cref{sec:Lindblad_synthetic,sec:flipflop} we benchmark the performance of the new methods extensively using synthetic data with realistic levels of statistical noise and a wide range of physically plausible noise types and strengths. Finally, in \Cref{sec:exp} we demonstrate the application of the algorithms to data obtained from real superconducting-qubit hardware. In particular we focus on characterising the noise associated with two-qubit entangling gates and idling processes, and assess whether and how the effective noise processes change when other gates are activated in parallel elsewhere on the device. 

\section{Background and Notation}\label{sec:background}
In this section, we give an overview of the necessary background and establish common notation used throughout this work. Throughout the rest of this manuscript, all norms $\lVert \cdot\rVert$ correspond to the Frobenius norm, defined by $\norm{M} = \sqrt{\sum_{j,k} \abs{m_{jk}}^2}$.

\subsection{Markovian Channels and Lindblad Generators} \label{subsec:Lindblad}

\newcommand{\Ket}[1]{| #1 \rangle\!\rangle}
\newcommand{\Bra}[1]{\langle\!\langle #1 |}
\newcommand{\BraKet}[2]{\langle\!\langle #1 | #2 \rangle\!\rangle}
\newcommand{\KetBra}[2]{\Ket{#1}\!\Bra{#2}}
The noisy dynamics that we wish to understand can be described by a quantum channel $\mathcal{E}$, i.e., a completely positive trace-preserving (CPTP) map, acting on $d\times d$ density matrices: $\rho \rightarrow \mathcal{E}(\rho)$. We can represent density matrices $\rho$ as vectors using stacked density matrix notation,
\begin{equation}
\Ket{\rho}_{j,k} = \BraKet{e_j,e_k}{\rho} \coloneq \bra{e_j}\rho\ket{e_k},
\label{eq:vec_rho}
\end{equation}
where $\ket{e_j} = (0,\cdots,0,1,0,\cdots,0)^T$ with 1 in the $j$-th position. A quantum channel $\mathcal{E}$ can then be represented by a $d^2 \times d^2$ matrix as follows
\begin{equation}\label{eq:transfer_matrix}
E_{(j,k),(l,m)} = \Bra{ e_j,e_k }E\Ket{e_l,e_m} \coloneq \Tr [\ketbra{e_k}{e_j} \mathcal{E}(\ketbra{e_l}{e_m})].
\end{equation}
The matrix $E$ is referred to as the \emph{transfer matrix} of $\mathcal{E}$ or the elementary basis representation of $\mathcal{E}$. For example, for $\mathcal{E}$ representing a noisy $2$-qubit quantum logic gate, $d=4$, so $E$ is a $16\times 16$ matrix. It is straightforward to verify using \Cref{eq:vec_rho,eq:transfer_matrix} that $E\Ket{\rho}=\Ket{\mathcal{E}(\rho)}$. Thus, in this representation, the action of the channel on the density matrix is given by matrix-vector multiplication, and the composition of channels corresponds to matrix-matrix multiplication. In practice, whenever $\mathcal{E}$ is expressed in the form
$$\mathcal{E}(\rho)=\sum_{k} A_k\rho B_k$$ for some complex matrices $\{A_k\}$ and $\{B_k\}$, we have
\begin{equation}
E=\sum_k A_k\otimes B_k^T.
\label{eq:Lindblad_matrix}
\end{equation}
In this picture, we can conveniently express the expectation value of hermitian observable $O$ on the state obtained by applying the channel $\mathcal{E}$ to initial state $\rho$, as
\begin{equation}
\Bra{O} E \Ket{\rho} = \Tr[O \mathcal{E}(\rho)],
\end{equation}
where the operator $O$ is vectorised in the same way as in~\Cref{eq:vec_rho}.

Throughout this work, we are interested in Markovian channels.
These are a subset of channels which are generated by Lindbladians, i.e. by generators of the form
\begin{equation}
\mathcal{L}(\rho) \coloneq i[\rho, H] + \sum_{\alpha}\gamma_{\alpha}\Big[J_\alpha \rho J^\dagger_\alpha - \frac{1}{2} \left(J_\alpha^\dagger J_\alpha \rho + \rho J_\alpha^\dagger J_\alpha \right) \Big]
\label{eq:Lindbladian_form}
\end{equation}
where $H$ is a Hamiltonian, each $J_\alpha$ is called a jump operator, and each $\gamma_\alpha$ is a real positive scalar \cite{lindblad1976}. In this canonical form, $H$ is traceless, all jump operators are traceless and normalised to Frobenius norm $1$, and $\tr(J_\alpha^\dagger J_{\alpha'})=0$ whenever $\alpha\neq\alpha'$. The first term in the Lindblad form describes the unitary part of the evolution, while the second term represents the dissipative part of the process. In this work, we restrict our analysis to time-independent Lindbladian generators for which the transfer matrix of the corresponding channel at time $t$ is given by $E(t) = e^{Lt}$ where $L$ is the elementary basis representation of $\mathcal{L}$.

The set of all Lindbladians in the elementary basis representation forms a closed convex cone with a simple known characterisation \cite{Markdynamics08,CEW09} through the $\Gamma$-involution, which is defined by linearly extending
\begin{equation}
(\KetBra{e_j,e_k}{e_l,e_m})^\Gamma \coloneq \KetBra{e_j,e_l}{e_k,e_m}.
\label{eq:gamma_involution}
\end{equation}
A matrix $L$ is a Lindbladian if and only if
\begin{enumerate}
\item $L^\Gamma$ is hermitian, i.e. the Choi matrix of $L$ is hermitian;

\item $L^\Gamma$ is \emph{conditionally completely positive}, i.e. $\omega_\perp L^\Gamma \omega_\perp \succeq 0$ where $\omega_\perp = (\identity-\KetBra{\omega}{\omega})$,
 and  $\Ket{\omega} \coloneq \frac{1}{\sqrt{d}}\sum_{j=1}^d \Ket{e_j,e_j}$;

\item $\Bra{\omega}L = 0$, corresponding to the trace-preserving property.
\end{enumerate}
For a small Hilbert space dimension like the $d=4$ case considered in this work, the problem of projecting a matrix to the set of all Lindbladians can be solved efficiently using convex optimisation. The projection problem is minimising $\lVert X-L\rVert_F$ subject to $L$ being a Lindbladian, with $X$ being the input and $L$ being the optimisation variable. In software, performing the projection amounts to setting up the problem using CVXPY \cite{diamond2016cvxpy} and then calling an off-the-shelf convex optimiser such as the SCS solver \cite{scs}. 

Given a Lindbladian $L$ in the elementary basis representation, it is possible to solve for a canonical decomposition comprised of a Hamiltonian $H$, jump operators $\{J_\alpha\}$, and real positive scalars $\{\gamma_\alpha\}$ such that the elementary basis representation of the Lindbladian form (see \Cref{eq:Lindbladian_form}) defined by $H$, $\{J_\alpha\}$, and $\{\gamma_\alpha\}$ is $L$. For completeness, we include a derivation of this fact in \Cref{app:lindblad_decomp}.

For $E=e^L$, $L$ is a matrix logarithm of $E$. The set of logarithms of $E$ is defined as $\log(E)=\{A:e^A=E\}$, and $\log(E)$ has infinite cardinality even for $1\times 1$ matrices since for every $\lambda\in\mathbb{C}$, $\lambda+2\pi k i$ is a logarithm of $e^{\lambda}$ for every integer $k$. Assume $E$ is diagonalisable with not necessarily unique complex eigenvalues $\mu_1,\ldots,\mu_{d^2}$. For every $j\in\{1,\ldots,d^2\}$, we can fix a logarithm $\lambda_j$ of $\mu_j$ so that $e^{\lambda_j}=\mu_j$ and $|\Im(\lambda_j)|\leq\pi$. If $\mu_j$ is not real negative, then the choice of $\lambda_j$ is unique. Then every $A\in\log(E)$ can be written in the form
\begin{equation}
A=\sum_{j=1}^{d^2}(\lambda_j+2m_j\pi i) \Ket{r_j}\Bra{l_j}
\label{eq:log_branch}
\end{equation}
where $m=(m_1,\ldots,m_{d^2})\in\mathbb{Z}^{d^2}$ is an integer vector and $\Ket{r_j}$ ($\Ket{l_j}$) is a unit right (left) eigenvector of $E$ corresponding to the eigenvalue $\mu_j$. Choosing an $m\in\mathbb{Z}^{d^2}$ corresponds to choosing a branch of the matrix logarithm. We call a branch satisfying $|\Im(\lambda_j)+2m_j\pi|\leq \pi$ for every $j\in\{1,\ldots,d^2\}$ a principal branch, and it is uniquely specified by $m=(0,\ldots,0)$ if $E$ has no real negative eigenvalues. When $\mu_1,\ldots,\mu_{d^2}$ are all distinct, $E$ has $d^2$ unique eigenspace projectors $\Ket{r_1}\Bra{l_1},\ldots,\Ket{r_{d^2}}\Bra{l_{d^2}}$, so $\log(E)$ is a countably infinite set whose discrete degrees of freedom are fully parametrised by the choice of $m$. 

The case where $E$ has degenerate eigenvalues is more intricate but not uncommon: this always occurs if $E$ corresponds to a unitary gate. A degenerate eigenspace of $E$ could split into different eigenspaces of $A$, and we call this phenomenon \emph{eigenspace splitting}. For instance, if $\mu_j=\mu_k$ for some $j\neq k$, then in specifying an $A\in\log(E)$, $\Ket{r_j}$ and $\Ket{r_k}$ could be any of uncountably many unit eigenvectors of $E$ corresponding to the eigenvalue $\mu_j=\mu_k$ while $\lambda_j+2m_j\pi i\neq \lambda_k+2m_k\pi i$ is possible. Therefore, there are continuous degrees of freedom in choosing an $A\in\log(E)$ if $E$ has degenerate eigenvalues. It is not hard to show that for a Lindbladian $L$, the condition $L^\Gamma$ is hermitian (see \Cref{eq:gamma_involution}) implies that all the complex eigenvalues of $L$ come in complex conjugate pairs \cite{Markdynamics08,CEW09,Onorati2023fittingquantumnoise}. As a result, the quintessential example of this eigenspace splitting phenomenon occurs for us when $-1$ is a $2$-fold degenerate eigenvalue of $E$, and this $2$-dimensional eigenspace gets split into two $1$-dimensional eigenspaces of $A$ corresponding to the eigenvalues $\pi i$ and $-\pi i$.

\subsection{Linear Inversion Quantum Process Tomography}
In quantum process tomography~\cite{QPTChuang_1997,Poyatos_QPT_1997}, a quantum channel is estimated by implementing the same process on a quantum device a large number of times on different input states and measuring each resulting output state in different bases. More concretely, let the $d^2\times d^2$ transfer matrix for the process of interest be denoted $E^{*}$, and let $\{\Ket{\rho_j}:j=1,\ldots,d^2\}$ and $\{\Bra{F_i}:i=1,\ldots,d^2\}$ be the chosen tomographically complete set of preparation states and measurement operators, respectively. Typically, $E^*$ is assumed to be the noisy implementation of some ideal process $E^{\mathrm{ideal}}$. Without loss of generality, we assume that each $F_i$ represents one outcome of a two-outcome POVM $\{F_i,\identity-F_i\}$. The quantum circuit that prepares state $\rho_j$, implements the target process and then carries out the measurement corresponding to $F_i$ is executed many times to obtain sufficient statistics to estimate the value of $\Bra{F_i}E^*\Ket{\rho_j} = \Tr[F_i \mathcal{E}^*(\rho_j)]$. This must be repeated for each pairing of measurement setting $F_i$ with preparation $\rho_j$. We can organise the chosen state preparations and measurements into matrices. The measurement matrix is given by a matrix in which each row is the effect of a possible measurement outcome
\begin{equation}
A^\text{ideal}=\begin{pmatrix}
\Bra{F_1} \\
\Bra{F_2} \\
\vdots \\
\Bra{F_i} \\
\vdots \\
\Bra{F_{d^2}}
\end{pmatrix}
\end{equation}
while the state preparation matrix is given by a matrix in which each column is one of the possible input states
\begin{equation}
B^{\text{ideal}}=\begin{pmatrix}
\Ket{\rho_1} & \Ket{\rho_2} & \cdots & \Ket{\rho_j} & \cdots & \Ket{\rho_{d^2}}\end{pmatrix}.
\end{equation}
In the case of an ideal process $E^{\mathrm{ideal}}$, the probabilities of obtaining outcome $F_i$ after preparing in state $\rho_j$ can be collected into a matrix
\begin{equation}
P^{\text{ideal}}=A^{\text{ideal}}E^\text{ideal}B^{\text{ideal}}.
\end{equation}
Note that $A^{\text{ideal}}$, $E^\text{ideal}$, and $B^{\text{ideal}}$ are all matrices chosen by the experimenter, so $P^\text{ideal}$ is known. However, in reality, the circuit elements used to implement $A^\text{ideal}$, $E^\text{ideal}$, and $B^\text{ideal}$ will be noisy. Thus, the true probabilities being estimated in quantum process tomography will be
\begin{equation}
P^*=A^*E^*B^*
\end{equation}
for some unknown $P^*$, $A^*$, $E^*$, and $B^*$ different from $P^\text{ideal}$, $A^{\text{ideal}}$, $E^{\text{ideal}}$, and $B^{\text{ideal}}$ respectively. In quantum process tomography, it is assumed that the state preparation and measurement settings are implemented exactly, whereas the true process $E^*$ is not known. Namely, $A^*=A^\text{ideal}$ and $B^*=B^\text{ideal}$. Then, the task is to estimate $E^*$ from $P^*$. There are a number of ways of doing this. In linear inversion tomography, $A^{\text{ideal}}$ and $B^{\text{ideal}}$ are chosen to be square and invertible, so under the assumption that $A^*=A^\text{ideal}$ and $B^*=B^\text{ideal}$, $E^*$ can be estimated via the identity
\begin{equation}
E^*=(A^{\text{ideal}})^{-1}P^*(B^{\text{ideal}})^{-1}.
\end{equation}
There are a number of more refined methods to estimate $E^*$ in standard quantum process tomography -- a shared feature of conventional approaches is the assumption that $A^*=A^{\text{ideal}}$ and $B^*=B^{\text{ideal}}$ \cite{Knee_PGDB2018, forestbenchmarking}.

In this work, we focus on $E^\text{ideal}$ being a $2$-qubit quantum logic gate, so $A^{\text{ideal}}$ and $B^{\text{ideal}}$ will each contain $16$ measurement and state preparation settings respectively. Measuring each setting independently therefore requires a suite of $16^2=256$ different circuits to estimate $P^*$. In practice, the circuit overhead can be reduced by parallelising measurement of compatible observables. For example, by measuring both qubits in the $Z$ basis, we can infer mean values for $Z\otimes \1$, $\1 \otimes Z$, and $Z \otimes Z$ using the same circuit. Using this scheme and measuring in the Pauli basis, the number of circuits can be reduced to 144. In a real experiment, since each matrix element $P^*_{ij}=\Bra{F_i}E^*\Ket{\rho_j}$ of $P^*$ is estimated by collecting statistics from a finite number of circuit executions (called tomographic shots), the experimental estimate for $P^*$ (and in turn the estimate for $E^*$) will inevitably differ from its true value, even when preparations and measurements are ideal. The amount of statistical error can be reduced by increasing the number of tomographic shots; however, doing so to a significant extent may be prohibitive on current quantum computing hardware. The practicality of taking a large number of shots is limited not just by the overall resource cost but also by the instability of noise processes on the device. In other words, the duration of the experiment increases with the number of shots taken, increasing the risk that $E^*$ can drift while data is being taken~\cite{Proctor_2020}. The number of shots can be tuned between low-precision tomography that captures a snapshot of the noise within a particular time window or higher-precision tomography that characterises the long-term average behaviour of the device. We note that since the set of Markovian channels is non-convex, the latter approach may yield non-Markovian channel estimates, even if the noise process at some particular time can be well-modelled by a time-independent Lindbladian. Statistical error can also be reduced by projecting the output of process tomography to CPTP, and there are known algorithms for doing so \cite{Mohseni_2008, Knee_PGDB2018, forestbenchmarking}.

\subsection{Gate Set Tomography}
In this section we give a brief overview of the simplest gate set tomography (GST) protocol, linear inversion GST \cite{Blume:2013}. A more comprehensive review including more sophisticated methods can be found in \cite{Nielsen2021gatesettomography}. In quantum process tomography (whether linear inversion tomography or some more sophisticated method), assuming $A^*=A^\text{ideal}$ and $B^*=B^\text{ideal}$ is a prerequisite to computing an estimate of $E^*$. Operationally, this corresponds to the assumption that there are no SPAM errors in the tomography experiment, and the experimenter knows the $A^{\text{ideal}}$ and $B^{\text{ideal}}$ matrices because they are chosen by the experimenter. This is a strong assumption that will never be satisfied exactly and may not even be approximately satisfied in near-term devices.

In linear inversion GST this assumption is replaced by two weaker assumptions:
\begin{itemize}
\item when we try to implement the (known) informationally complete state preparations $B^\text{ideal}$ and measurements $A^\text{ideal}$, the state preparations $B^*$ and measurements $A^*$ that are actually implemented (due to SPAM errors) are still informationally complete;
\item the SPAM errors are independent of which quantum process is being implemented, and remain fixed for a given preparation or measurement setting.
\end{itemize}
With these two assumptions, the GST method is as follows.
Choose a set of $k$ distinct quantum gates $E_1^{\text{ideal}},\ldots ,E_k^{\text{ideal}}$ (here each $E_i^\text{ideal}$ is the transfer matrix of a unitary gate) and together with the (instantaneous) identity gate, form a gate set
\begin{equation}
\mathcal{G}=\{\identity,E_1^{\text{ideal}},\ldots ,E_k^{\text{ideal}}\}.
\end{equation}
Then carry out quantum process tomography for every gate in $\mathcal{G}$. For each gate $i\in\{1,\ldots,k\}$, let $E_i^*$ denote the actual noisy gate implemented by the quantum device. This enables the construction of the matrices
\begin{equation}
P_i^* = A^*E_i^*B^*.
\end{equation}
Since $A^*$ and $B^*$ are unknown, it is not possible to reconstruct $E_i^*$ from $P_i^*$. However, by implementing the identity gate instantaneously, we can construct what is known as the \emph{Gram matrix}
\begin{equation}
g^*=A^*B^*.
\end{equation}
We can therefore construct an estimate 
\begin{equation}\label{eq:sim}
E_i=B(g^*)^{-1}P_i^*B^{-1}=B(B^*)^{-1}E_i^*B^*B^{-1}
\end{equation}
for $E_i^*$ up to some \emph{similarity transform} $B$. 

This may seem unfinished -- we do not know $B$, so how can we construct an estimate of $E_i^*$? One of the key insights of GST is that it is only ever possible to reconstruct each $E_i^*$ up to some similarity transform since these similarity transforms are unobservable. Thus, \Cref{eq:sim} is the best we can hope to achieve for each gate. To see this, notice that for every invertible $B$ and $A=g^*B^{-1}$,
$$A^*B^*=g^*=AB$$
and for every gate $i$,
$$A^*E_i^*B^*=A(B(B^*)^{-1}E_i^*B^*B^{-1})B=AE_iB.$$
Therefore, assuming $A$, $B$, and every $E_i$ are all physical, it is impossible to distinguish whether the true quantum process tomography experiments have occurred with measurement, preparation, and noisy gates being $A^*$, $B^*$, and $E_i^*$ respectively or $A$, $B$, and $E_i$ respectively since they both lead to identical Gram and probability matrices. This implies that the similarity transform $B$ is in fact just picking out a particular gauge. The different choices of gauge give equivalent gate sets, so when considering the entire gate set, any invertible (and physical) choice of matrix is as good as any other for $B$. However, the choice of gauge clearly has a significant effect when we are interested in characterising the individual gates within the gate set. If, for example, we take a perfect set of gates (with no errors) but choose a gauge that rotates between the different gates in the set, then this will imply that the gate fidelities are all very low and that there are large coherent errors in the device. However, this is merely an artefact of the chosen gauge. 

In gate set tomography the gauge is chosen to be the one that minimises the distance to the ideal gate set (in some metric - see \cite{Nielsen2021gatesettomography} for a discussion of different choices of metric).
More precisely, in GST the gauge is chosen to be a minimiser of the optimisation problem in some given norm $\norm{ \cdot }_{\ast}$
\begin{equation}
\min_B\left(\sum_{i=1}^k a_i\|B(g^*)^{-1}P_i^*B^{-1} - E_i^{\text{ideal}} \|_{\ast} +\sum_{j=1}^{d^2} b_j\|B_j - \Ket{\rho_j} \|_{\ast} +\sum_{i=1}^{d^2} c_i\|A_i  - \Bra{F_i} \|_{\ast}\right)
\end{equation}
where $A=gB^{-1}$, $B_j$ denotes the $j$-th column of $B$, $A_i$ denotes the $i$-th row of $A$, and $\{a_i\}$, $\{b_j\}$, and $\{c_i\}$ are non-negative weights. Intuitively, the weights should be chosen to be proportional to the believed fidelity of a gate, preparation, or measurement setting. For example, if $\Ket{\rho_1}$ corresponds to preparing a known high-fidelity $\ket{0}$ state and the other $\Ket{\rho_j}$'s are prepared from $\ket{0}$ by applying a short-depth circuit, it may be sensible to choose $b_1=1$ and $b_j=0$ for all $j\geq 2$. The simplest choices for the weights are $a_i=1$ for all $i\in\{1,\ldots,k\}$ and $b_j=c_j=0$ for all $j\in\{1,\ldots,d^2\}$, hence simplifying the gauge optimisation problem to
\begin{equation}
\min_B\sum_{i=1}^k\|B(g^*)^{-1}P_i^*B^{-1} - E_i^{\text{ideal}}\|_{\ast}.
\label{eq:gst_objective}
\end{equation}
In summary, in its simplest form, the input to GST is a set of matrices
\begin{equation}
\{g^*,P_1^*,\ldots,P_k^*\}=\{A^*B^*,A^*E_1^*B^*,\ldots,A^*E_k^*B^*\}
\end{equation}
obtained from $k+1$ runs of quantum process tomography, and the desired output is a matrix $B$ that minimises \Cref{eq:gst_objective}. Note that $B^*$ may not be an optimal solution for \Cref{eq:gst_objective} if $E_i^{\text{ideal}}\neq E_i^*$ for at least one gate $i$.

\section{Our Work}\label{sec:algorithms}

\subsection{The Lindbladian Fitting Problem}\label{subsec:TTE}

The input to the Lindbladian fitting problem is the transfer matrix $E$ of a CPTP map (recall \Cref{eq:transfer_matrix}), and the goal is to fit a Lindbladian model to $E$. More concretely, we wish to find a Lindbladian $L$ that minimises $\lVert e^L-E\rVert$.

In this work, we restrict ourselves to special cases of the Lindbladian fitting problem that satisfy two additional assumptions. Firstly, it is assumed that the input $E$ is close to some unknown Markovian channel $E^*=e^{L^*}$. That is, $\lVert E-E^*\rVert\leq c_1$ for some small constant $c_1$. We further assume that $E^*$ is not far from a known Markovian channel $E^{\text{ideal}}=e^{L^{\text{ideal}}}$. In fact, we assume $\lVert L^*-L^{\text{ideal}}\rVert\leq c_2$ for some small constant $c_2$ which is stronger than assuming $E^*$ is not far away from $E^{\text{ideal}}$. As an example, suppose $L$ is a Lindbladian generator for the $1$-qubit Pauli $Z$ gate, then $3L$ would be a Lindbladian generator for $Z^3$, which is identical to $Z$ as a unitary gate, but $\lVert L-3L\rVert$ is large. The algorithms we consider will search over the branches of the complex matrix logarithm $\log(E)=\{A:e^A=E\}$ in the outer most loop\footnote{Throughout this work, we do assume $E$ and $E^{\text{ideal}}$ are both diagonalisable.}. The bound $\lVert L^*-L^{\text{ideal}}\rVert\leq c_2$ implies a bound on the differences between the imaginary parts of the eigenvalues of $L^*$ and $L^{\text{ideal}}$, which in turn limits the number of physically relevant branches of $\log(E)$ to search over. The assumption $\lVert L^*-L^{\text{ideal}}\rVert\leq c_2$ is also intended to rule out cases like $L^{\text{ideal}}=0$, which generates the identity gate, while $L^*$ generates a far away unitary gate such as $X\otimes I$. 

In an experiment, 
\begin{itemize}
	\item $E$ would be the output of quantum process tomography (optionally projected to CPTP), 
	\item $E^{\text{ideal}}$ would correspond to an ideal gate, 
	\item $L^\text{ideal}$ is a known Lindbladian generator that the experimenter has chosen to execute, 
	\item and $E^*$ would be the Markovian channel closest to the true noisy gate being implemented in place of $E^{\text{ideal}}$. 
\end{itemize}

The difference between $E^*$ and $E^{\text{ideal}}$ captures the total amount of time-independent Markovian gate noise while the discrepancy between $E$ and $E^*$ accounts for possible non-Markovian or time-dependent Markovian gate noise as well as inevitable statistical error from performing process tomography. The two assumptions translate to the physical assumptions that the gate noise is approximately time-independent Markovian and not too strong (e.g. some noise turning $Z$ into $Z^3$). 
In this work, the two assumptions are stated heuristically, and we do not have formal bounds on the precise values of $c_1$ and $c_2$ required by our algorithm. Nonetheless, we believe our weak approximately Markovian gate noise assumption is satisfied when running process tomography on individual gates on today's quantum computing hardware. In \Cref{sec:Lindblad_synthetic}, we give numerical evidence that $c_2$ does not need to be too small.

\subsubsection{The Convex Solve Method} \label{sec:convexsolve}

To set the stage for the Alternating Projections method and emphasise the effect of statistical error, we first describe the simpler Convex Solve method (\Cref{alg:convexsolve}). Due to statistical error, the eigenvalues of $E$ are (almost surely) all distinct (regardless of what $E^{\text{ideal}}$ or $E^*$ are). This implies $\log(E)=\{A:e^A=E\}$ is a countably infinite set which we can enumerate over. The Convex Solve method then simply enumerates over each $A\in\log(E)$ and finds the Lindbladian $L$ closest to $A$. The latter task can be accomplished by solving a positive semidefinite program (recall \Cref{subsec:Lindblad}). Finally, the algorithm returns the $L$ ever obtained that achieves the smallest $\lVert e^L-E\rVert$. Again, by the assumption $\lVert L^*-L^{\text{ideal}}\rVert\leq c_2$, we can terminate the search over the branches of $\log(E)$ as soon as $\lVert A-L^{\text{ideal}}\rVert$ becomes too large (recall the $Z$ vs. $Z^3$ example). In practice, rather than explicitly checking $\lVert A-L^{\text{ideal}}\rVert$, we can restrict the search to a few low-lying branches of the matrix logarithm. We note that close variants of the Convex Solve method have been considered in prior works \cite{Markdynamics08,CEW09,Onorati2023fittingquantumnoise}, specifically for the case where the eigenvalues of $E^*$ are all distinct. As such, we do not claim that our main contribution here lies in proposing the Convex Solve method, but rather it is an improved understanding of the effectiveness of this strategy, which we explain next.

The simplest special case of the Convex Solve method simply returns the Lindbladian closest to a principal branch of $\log(E)$, which is unique as long as $E$ does not have real negative eigenvalues even when $E$ has degenerate eigenvalues. In \Cref{app:trivial}, we prove the following result, which applies in the absence of statistical error.
\begin{theorem}\label{thm:trivial}
Suppose there exists a Lindbladian $L^*$ such that $E^*=\exp(L^*)$, and 
\begin{equation}
\lVert L^*-L^{\text{ideal}}\rVert<\frac{\pi-\rho(L^\text{ideal})}{\kappa(V)},
\end{equation}
where $V$ is an invertible matrix that diagonalises $L^\text{ideal}$, $\kappa(V)$ is its condition number, and $\rho(L^\text{ideal})$ is the spectral radius of $L^\text{ideal}$. Then given $E=E^*$ as input, the Convex Solve method outputs an $L\in\mathcal{L}\cap\log(E)$.
\end{theorem}
This tells us that under the assumption of weak, time-independent Markovian noise, even when $E^*$ contains degenerate eigenvalues, the simplest special case of the Convex Solve method already solves the problem in the absence of statistical error (i.e. $E=E^*$), provided $|\Im(\lambda)|<\pi$ for every eigenvalue $\lambda$ of $L^{\text{ideal}}$. The last condition can be satisfied when the transfer matrix $\exp(L^{\text{ideal}})$ has no eigenvalues close to the negative real line, which holds with a meaningful gap for many important gates, including $\sqrt{X}\otimes I$, $T\otimes I$, and most notably, the identity gate $I\otimes I$. Of course, in reality, statistical error will always be present, but empirically we observe that the result is surprisingly robust to shot noise. Our synthetic testing in \Cref{sec:Lindblad_synthetic} suggests that the Convex Solve method also already works at the level of statistical noise induced by $10^4$ tomographic shots for these gates.

The Convex Solve method is not the whole story because the transfer matrices $E^{\text{ideal}}$ of an important family of gates that we would like to characterise do contain real negative eigenvalues, so under the weak noise assumption, $E^*$ and $E$ will have eigenvalues close to being real negative. This family includes CNOT, ISWAP, and the tensor product of Pauli and Hadamard gates. Under a realistic level of statistical noise induced by $10^4$ tomographic shots, we observe that for gates from this family, both $\lVert L-A\rVert$ and $\lVert e^L-E\rVert$ could be significantly larger than $\lVert E-E^*\rVert$ for every closest Lindbladian $L$ to every $A\in\log(E)$ (see \Cref{sec:Lindblad_synthetic} for a CNOT example). Hence, the Convex Solve method fails on these inputs. 

\begin{algorithm}\small
\SetKwInOut{Input}{input}
\SetKwInOut{Output}{output}
\SetKw{KwMin}{minimise}
\SetKw{KwSubjTo}{subject to}
\Input{Estimated (diagonalisable) transfer matrix $E$; max branch number $m_{\text{max}}$}
\Output{Closest Lindbladian generator $L_{\text{est}}$ to all checked branches of logarithm}
\caption{Convex Solve Method}\label{alg:convexsolve}
Compute (some) set of eigenvalues $\{\mu_j\}_j$ and right $\{\Ket{r_j}\}_j$ and left  $\{\Bra{l_j}\}_j$ eigenvectors of $E$

\For{$ j = 1$ \KwTo $d^2$}{$\lambda_j \gets$ principal log of $\mu_j$}

\For { $m \in \{-m_{\text{max}},\ldots,0,\ldots, m_{\text{max}} \}^{d^2}$}{
	$A_{m} \leftarrow \sum_{j=1}^{d^2}(\lambda_j + 2 \pi i m_j )\Ket{r_j} \Bra{l_j}$
    
    Solve the convex optimisation problem \\
	\hspace{1.13cm} \KwMin \hspace{0.5cm} $\norm{L - A_{m}}$  \\
	\hspace{1cm} \KwSubjTo \hspace{0.5cm}$L^{\Gamma}$ Hermitian \\
	 \hspace{3.6cm} $\omega_\perp L^\Gamma \omega_\perp \succeq 0$ \\
	 \hspace{3.6cm} $\Bra{\omega}L = 0$ \\
     and let $L_m$ be the returned optimal solution
     
     $\Delta_{m} \leftarrow \norm{e^{L_m} - E}$
}
$L_{\text{est}} \gets L_{m_{\text{opt}}}$ such that $m_{\text{opt}}$ is the branch that minimises $\Delta_{m}$

\KwRet $L_{\text{est}}$
\end{algorithm}

\subsubsection{The Alternating Projections Method}\label{sec:alternating_projection}
Now, we are ready to describe the Alternating Projections method. The algorithm given in \cite{Onorati2023fittingquantumnoise} operates under similar assumptions but is too slow in practice for the $2$-qubit $d=4$ case. In this work, we build upon the understanding developed in \cite{Onorati2023fittingquantumnoise} and give a new algorithm that has significantly better performance in terms of both accuracy and runtime. The insight in the new algorithm is that we don't have to solve the Lindbladian fitting problem ``blind'' -- we have knowledge about what a ``good'' initial guess is (because we know what gate we tried to implement). We can start from the Lindbladian generator of the ideal gate $L^\text{ideal}$. 

Here, we sketch the main idea of the algorithm (\Cref{alg:APinformal}) and defer a complete technical description of our implementation to \Cref{app:TTEAPtech}. The input to the algorithm is $E$ and an initial best-guess Lindbladian $L_0$. The default choice sets $L_0=L^{\text{ideal}}$, but a better estimate, if available, could be used instead. Let $\mathcal{L}$ denote the set of all Lindbladians and define $\log(E)=\{A:e^A=E\}$. The goal is to compute
\begin{equation}
\min_{L\in\mathcal{L}}\text{ }\lVert e^L-E\rVert=\min_{L\in\mathcal{L},A\in\log(E)}\lVert e^L-e^A\rVert
\label{eq:Lindblad_fitting_objective}
\end{equation}
and output an optimal Lindbladian $L$. We approximate the latter minimisation problem by the problem
\begin{equation}
\min_{L\in\mathcal{L},A\in\log(E)}\lVert L-A\rVert.
\label{eq:dist_two_sets}
\end{equation}
Note that if $E$ is Markovian, i.e. $E=E^*$, then the optimal values of \Cref{eq:Lindblad_fitting_objective} and \Cref{eq:dist_two_sets} are both zero, and the solution set is precisely $\mathcal{L}\cap\log(E)$. If both $\mathcal{L}$ and $\log(E)$ were closed convex sets, then the problem in \Cref{eq:dist_two_sets} can be solved using alternating projections which start from a given starting point and alternately project onto each of the two sets until convergence. In our case, $\mathcal{L}$ is a closed convex cone and projecting onto it can be accomplished by solving a positive semidefinite program (recall \Cref{subsec:Lindblad}). However, the set $\log(E)$ is not convex, and we do not know how to solve the problem $\min\lVert A-L_0\rVert$ subject to $A\in\log(E)$ directly. To this end, we devise a procedure to construct an $A\in\log(E)$ that approximately minimises $\lVert A-L_0\rVert$, inspired by the assumption that $\lVert E-e^{L_0}\rVert$ is small to begin with. Suppose $E$ has $n\leq d^2$ distinct eigenvalues associated with $n$ unique eigenspace projectors $\Pi_1,\ldots\Pi_n$. Let $v_1,\ldots, v_{d^2}$ be a set of eigenvectors of $L_0$. For every pair of eigenvector $v_j$ and projector $\Pi_k$, we compute the cost $\lVert v_j-\Pi_kv_j\rVert$. We then solve a minimum-cost maximum-flow problem to assign each eigenvector $v_j$ to a best-fit eigenspace $\Pi_k$. A matrix $\bar{A}\in\log(E)$ is reassembled by pairing up the log of the eigenvalues of $E$ (in some branch) and the eigenvectors $\Pi_kv_j$. The algorithm then projects $\bar{A}$ onto $\mathcal{L}$ to get a Lindbladian $\bar{L}$ and checks whether 
\begin{equation}
\lVert e^{\bar{L}}-E\rVert<\lVert e^{L_0}-E\rVert.
\label{eq:continue}
\end{equation}
If yes, then we update the current best-guess Lindbladian to be $\bar{L}$ and repeat from $\bar{L}$.
\begin{algorithm}\small
 \SetKwInOut{Input}{input}
 \SetKwInOut{Output}{output}
 \SetKw{KwMin}{minimise}
 \SetKw{KwSubjTo}{subject to}
 \SetKw{KwBreak}{break}
 \Input{Estimated transfer matrix $E$; initial guess Lindbladian $L_0$; precision $\beta\geq0$; maximum search depth $T$}
 \Output{Estimated Lindbladian $L^{\text{est}}$ such that $\exp L^{\text{est}} \approx E $}
Compute eigendecomposition of input tomographic transfer matrix, $E = \sum_j \mu_j \Ket{r_j} \Bra{l_j}$\\
Cluster tomographic eigenvalues into approximate eigenspaces $S_k$ according to precision $\beta$\\
Compute principal log eigenvalues $\lambda_j$ such that $e^{\lambda_j}=\mu_j$\\
\For{branches $m$ and random perturbations $\tilde{L}_0$}{
Compute log eigenvalues of tomographic data according to the chosen branch, $\hat{\lambda}_j = \lambda_j + 2 \pi i m_j$\\
\For{$t=1$ \KwTo $T$}{
Get eigenvalues $\sigma_j$ and eigenvectors $\Ket{\tilde{v}_j}$ of current Lindbladian model $\tilde{L}_{t-1}$\\
Match and project model eigenvectors $\Ket{\tilde{v}_j}$ onto tomographic subspaces $S_k$\\
Pair log-eigenvalues $\hat{\lambda}_{j'}$ with projected eigenvectors $\Pi_k \Ket{\tilde{v}_j}$, by matching to the associated model eigenvalues $\sigma_j$, so as  to construct approximate matrix logarithm $\bar{A}_m$\\
Run convex optimisation to project $\bar{A}_m$ back onto the set of Lindbladians, outputting an updated model $\tilde{L}_t$
}
Append $\tilde{L}_T$ to list of candidate solutions $\mathbb{L}$
}
\KwRet $L \in \mathbb{L}$ that minimises $\norm{e^L - E}$
\caption{Alternating Projections (informal)}\label{alg:APinformal}
\end{algorithm}

\subsubsection{Eigenvalue Clustering}

Since the eigenvectors of $E$ tend to differ non-trivially from the eigenvectors of $E^*$ under realistic levels of statistical error, we adopt the eigenvalue clustering preprocessing step proposed in \cite{Onorati2023fittingquantumnoise}. That is, the user needs to specify a precision parameter $\beta$ so that the eigenvalues of $E$ differing by less than $\beta$ are treated as being identical and their eigenspaces are merged into a single degenerate eigenspace. Note that this causes the matrix $\bar{A}$ constructed during the algorithm to belong to some enlarged set $\widetilde{\log}(E)=\{A:e^A\approx E\}$ as opposed to $\log(E)$ where the $\approx$ is controlled by the precision parameter $\beta$. Note that if $\beta=0$, then $\widetilde{\log}(E)=\log(E)$. Also, notice that our Alternating Projections method reduces to the Convex Solve method if we set the precision parameter $\beta=0$. Hence, the Alternating Projections method includes the Convex Solve method as a special case.

Consider a case where $L^{\text{ideal}}$ has a $k$-fold degenerate eigenspace $\Pi$ that gets split into two (possibly still degenerate) eigenspaces $\Pi_1$ and $\Pi_2$ of $L^*$ due to gate noise. When our algorithm computes a set of eigenvectors $v_1,\ldots,v_k$ for $\Pi$, there is no guarantee which vectors in $\Pi$ will the eigen-solver return. Indeed, $L^{\text{ideal}}$ is indifferent to the choices of $v_1,\ldots,v_k$ as all of them lead to the same eigenspace projector. However, it is possible that $v_1,\ldots,v_k$ are all far away from $\Pi_1$ or $\Pi_2$. For example, consider $\Pi=\mathbb{C}^2$, $v_1=\ket{0}$, $v_2=\ket{1}$, $\Pi_1=\text{span}\{\ket{+}\}$, and $\Pi_2=\text{span}\{\ket{-}\}$. This will not cause trouble if $\lVert E-E^*\rVert$ is sufficiently small to allow us to choose $\beta$ small enough so that the eigenvalues of $E$ corresponding to $\Pi_1$ and $\Pi_2$ are treated as distinct. However, in the parameter regimes feasible today, it is likely that the algorithm would need to choose a large $\beta$, which results in the merging of $\Pi_1$ and $\Pi_2$, so that the matrix $\bar{A}\in \widetilde{\log}(E)$ constructed is not too far away from $\mathcal{L}$. In other words, choosing a large $\beta$ blurs the goal of finding vectors close to $\Pi_1$ and $\Pi_2$. To address this issue, instead of initialising at precisely $L^{\text{ideal}}$, the algorithm tries to roughly guess the correct ``directions'' by starting at random perturbations $\tilde{L}^{\text{ideal}}$ of $L^{\text{ideal}}$. We empirically observe that for the $1$-qubit case, a handful of uniformly random perturbed starts would suffice. However, for the $2$-qubit case, the search space becomes too vast to cover by a small number of uniformly random perturbations. As a heuristic solution, we perturb diagonally
\begin{equation}
\tilde{L}^{\text{ideal}}=L^{\text{ideal}}+D
\end{equation}
and diagonally ``along the $X$-axis''
\begin{equation}
\tilde{L}^{\text{ideal}}=L^{\text{ideal}}+H^{\otimes 4}DH^{\otimes 4}
\end{equation}
for random diagonal matrices $D$ with small norms.

Due to the mismatch between \Cref{eq:Lindblad_fitting_objective} and \Cref{eq:dist_two_sets} and the fact that $\log(E)$ (or for that matter $\widetilde{\log}(E)$) is non-convex, standard convergence guarantees for alternating projections methods do not apply to our algorithm. For the Markovian $E=E^*$ case, the arguments given in \Cref{app:trivial} directly translate to the Alternating Projections method when choosing the precision parameter $\beta=0$. We leave proving stronger formal guarantees for our algorithm as future work.

\subsubsection{Criteria for Success}\label{sec:criteria}

We have implemented our algorithm in software and tested extensively on synthetic data obtained using simulated quantum process tomography. We focus on the case where the gate noise is Markovian, so $E^*$ corresponds to the ground truth channel while $E$ differs from $E^*$ solely due to statistical error. In this case, on the one hand, it is almost certain that
\begin{equation}
\min_{L\in\mathcal{L}}\text{ }\lVert e^L-E\rVert<\lVert E^*-E\rVert
\end{equation}
while on the other hand, it is pointless to search for a Lindbladian $L$ whose objective value $\lVert e^L-E\rVert$ is much smaller than $\lVert E^*-E\rVert$. We introduce two criteria for positive retrieval of Lindbladians in our testing:

	\begin{itemize}[leftmargin=*]
		\item \textbf{Success 1}: The output Lindbladian \( L \) satisfies 
		\[
		\lVert e^L - E \rVert \leq \lVert E - E^* \rVert
		\]
		
		\item \textbf{Success 2}: The output Lindbladian \( L \) satisfies 
		\[
		\lVert e^L - E^* \rVert \leq \lVert E - E^* \rVert
		\]
		This is a more stringent success criterion, where the output Lindbladian \( L \) closely fits the ground truth transfer matrix \( E^* \).
	\end{itemize}

\noindent We view passing Success 2 as a bonus since the algorithm is not intended to remove statistical error from $E$. Instead, we would realistically expect $\lVert e^L-E^*\rVert\approx\lVert E-E^*\rVert$ and passing Success 1 guarantees at a bare minimum that $\lVert e^L-E^*\rVert\leq\lVert e^L-E\rVert+\lVert E-E^*\rVert\leq 2\lVert E-E^*\rVert$. We tested our Alternating Projections method on $600$ test cases split across $30$ different gate and gate noise combinations, all under a realistic level of statistical error induced by $10^4$ tomographic shots. Our algorithm achieved success rates of $100\%$ and $61.4\%$ w.r.t. Success 1 and Success 2 respectively. See \Cref{tab:synthetic_aggregate} for more detail and \Cref{sec:Lindblad_synthetic} for additional test results.

\subsection{The Gate Set Lindbladian Fitting Problem}\label{subsec:flip-flop}

\subsubsection{Removing SPAM Errors Under Ideal Statistics}
We first describe the gate set Lindbladian fitting problem in the theoretical setting where there is no statistical error and all gate noises are Markovian. In this case, the input to the gate set Lindbladian fitting problem consists of a set of $k+1$ matrices
\begin{equation}
\{g^*,P_1^*,\ldots,P_k^*\}=\{A^*B^*,A^*E^*_1B^*,\ldots,A^*E^*_kB^*\}
\end{equation}
where each $E^*_i=e^{L_i^*}$ is the transfer matrix of a Markovian channel representing a chosen unitary gate $E^{\text{ideal}}_i=e^{L_i^{\text{ideal}}}$ perturbed by some small gate noise, each row of $A^*$ encodes a 2-outcome POVM, and each column of $B^*$ is a vectorised density matrix. It is assumed that $A^*$ and $B^*$ originate from effecting some SPAM errors to some chosen ideal measurement (2-outcome POVMs) settings $A^{\text{ideal}}$ and preparation (density matrices) settings $B^{\text{ideal}}$ respectively. Note that the SPAM errors are not assumed to be Markovian. The Gram matrix $g^*=A^*B^*$ is assumed to be invertible. By left multiplying $(g^*)^{-1}$ to each $P_i^*$, we obtain a set of matrices
\begin{equation}
\{(B^*)^{-1}E_1^*B^*,\ldots,(B^*)^{-1}E_k^*B^*\},
\end{equation}
one for each gate in the gate set. A solution to the gate set Lindbladian fitting problem consists of, and hence the goal is to find, a physical $B$ and Lindbladians $L_1,\ldots,L_k$ such that
\begin{equation}
B(g^*)^{-1}P_i^*B^{-1}=B(B^*)^{-1}E_i^*B^*B^{-1}=e^{L_i}
\end{equation}
for every gate $i\in\{1,\ldots,k\}$ and $A=g^*B^{-1}$ is physical. We define $A$ to be physical if every row of $A$ encodes a $2$-outcome POVM while $B$ is physical if $B$ is invertible and every column of $B$ is a vectorised density matrix. Let $(B,L_1,\ldots,L_k)$ be a solution and for every gate $i$, define $E_i=e^{L_i}$. While it is clear that $(B^*,L_1^*,\ldots,L_k^*)$ is a solution, it is important to note that $(B,L_1,\ldots,L_k)$ need not be equal to $(B^*,L_1^*,\ldots, L_k^*)$. In fact, given the input
\begin{equation}
\{A^*B^*,A^*E^*_1B^*,\ldots,A^*E^*_kB^*\},
\end{equation}
it is impossible to distinguish whether the underlying physical process is generated by $(A^*,B^*,L_1^*,\ldots,L_k^*)$ or $(A,B,L_1,\ldots,L_k)$ since $AB=A^*B^*$
and for every gate $i\in\{1,\ldots,k\}$,
\begin{equation}
AE_iB=A(B(B^*)^{-1}E_i^*B^*B^{-1})B=A^*E_i^*B^*.
\end{equation}
In other words, $(A^*,B^*,L_1^*,\ldots,L_k^*)$ and $(A,B,L_1,\ldots,L_k)$ generate the same input. Nevertheless, under the same weak noise assumptions that we made in \Cref{subsec:TTE}, $B^*$ is not far away from $B^{\text{ideal}}$ and $E_i^*$ is not far away from $E_i^{\text{ideal}}$ for every gate $i$. Thus, we would prefer solutions closer to $(B^{\text{ideal}},L_1^{\text{ideal}},\ldots,L_k^{\text{ideal}})$.

The Lindbladian fitting problem is precisely captured by the optimisation problem of minimising the objective function
\begin{equation}
f(B,L_1,\ldots,L_k)=\sum_{i=1}^k\left\lVert B(g^*)^{-1}P_i^* B^{-1}-e^{L_i}\right\rVert
\label{eq:gauge_obj_fun}
\end{equation}
subject to the constraints that $L_1,\ldots,L_k$ are Lindbladians, $B$ is physical, and $A=g^*B^{-1}$ is physical. Clearly, $(B,L_1,\ldots,L_k)$ is a solution to the gate set Lindbladian fitting problem if and only if $(B,L_1,\ldots,L_k)$ is feasible and $f(B,L_1,\ldots,L_k)=0$. Comparing to standard gauge optimisation in GST, note that the objective function \Cref{eq:gauge_obj_fun} is similar to \Cref{eq:gst_objective}, except replacing each $E_i^\text{ideal}$ with $e^{L_i}$, and that $B^*$ \textit{is} an optimal gauge for \Cref{eq:gauge_obj_fun}.

\subsubsection{Accepting Statistical Error}
Now consider the realistic scenario where the input matrices are not precisely
\begin{equation}
\{g^*,P_1^*,\ldots,P_k^*\}
\end{equation}
but are some nearby matrices
\begin{equation}
\{\tilde{g},\tilde{P_1},\ldots,\tilde{P_k}\}
\end{equation}
due to the inevitable statistical error in performing quantum process tomography. In this case, there is unlikely to be a physical $B$ that makes $B\tilde{g}^{-1}\tilde{P_i}B^{-1}$ exactly Markovian for any gate $i$. Thus, the goal of the gate set Lindbladian fitting problem needs to be relaxed to accommodate statistical error. To this end, we relax the goal to finding a near-physical $B$ such that $A=\tilde{g}B^{-1}$ is near-physical and $B\tilde{g}^{-1}\tilde{P_i}B^{-1}$ is close to being Markovian for every gate $i$. The amount of slack allowed for the physical constraints is an input parameter selected by the user and should be chosen based on the number of tomographic shots used to estimate the Gram matrix $\tilde{g}$. We model this problem by the optimisation problem minimise
\begin{equation}
f_{\max}(B,L_1,\ldots,L_k)=\max_{i=1,\ldots,k}\lVert B\tilde{g}^{-1}\tilde{P}_iB^{-1}-e^{L_i}\rVert
\label{eq:f_max_def}
\end{equation}
subject to the constraints that $L_1,\ldots,L_k$ are Lindbladians, $B$ is near-physical, and $A=\tilde{g}B^{-1}$ is near-physical. We minimise the max over the gate set in the definition of $f_{\max}$ as opposed to the sum in $f$ to enforce the requirement of finding a single gauge $B$ that simultaneously fits all gates in the gate set. Under the sum formulation, a gauge $B$ that fits some of the gates in the gate set extremely well but the remaining gates poorly may have a small objective value. Such a $B$ clearly does not explain the true physical process and hence is undesirable. Adding more gates to the gate set may alleviate this problem. However, in a real experiment, it is common for the experimenter to wanting to characterise only a particular gate, for example CNOT, but GST is performed on a gate set consisting of CNOT and other auxiliary gates chosen solely for the purpose of fitting the SPAM errors. Thus, it is also undesirable to require a large gate set. The max formulation in \Cref{eq:f_max_def} tries to address this trade-off. Furthermore, since some gates could be noisier than others in the gate set, substituting $\max$ in the objective function of standard gauge optimisation (see \Cref{eq:gst_objective}) incorrectly incentivises pulling gate errors in the poorest-fit gate into SPAM errors. However, the aforementioned problem does not appear in the formulation of \Cref{eq:f_max_def} since the gate errors are self-consistently fitted jointly with SPAM errors.

Applying gradient-based local minimisation algorithms to optimise \Cref{eq:f_max_def} is tricky since the max function is not differentiable everywhere. One common workaround is to approximate the max function by a smooth approximation, and this is the approach we take. The LogSumExp (LSE) function is defined as
\begin{equation}
\text{LSE}(x_1,\ldots,x_k)=\ln(e^{x_1}+\cdots+e^{x_k}), \label{eq:LSE}
\end{equation}
and it is easy to verify that for every $t>0$,
\begin{equation}
\max\{x_1,\ldots,x_k\}\leq\frac{1}{t}\text{LSE}(tx_1,\ldots,tx_k)\leq\max\{x_1,\ldots,x_k\}+\frac{\ln(k)}{t}.
\end{equation}
Moreover, the LSE function is monotonically increasing in all of its inputs. We fix the scaling factor $t$ to be some constant that depends on $k$ and the number of tomographic shots. Our final objective function for the gate set Lindbladian fitting problem is the following approximation of \Cref{eq:f_max_def} based on the LSE function:
\begin{equation}
h(B,L_1,\ldots,L_k)=\frac{1}{t}\text{LSE}(t\lVert B\tilde{g}^{-1}\tilde{P}_1B^{-1}-e^{L_1}\rVert,\ldots,t\lVert B\tilde{g}^{-1}\tilde{P}_kB^{-1}-e^{L_k}\rVert).
\label{eq:h_def}
\end{equation}

\subsubsection{The Gate Set Flip-Flop Algorithm}
The idea of the Gate Set Flip-Flop algorithm is to minimise \Cref{eq:h_def} via local minimisation by alternating between minimising $B$ with $L_1,\ldots,L_k$ fixed and minimising $L_1,\ldots,L_k$ with $B$ fixed; see \Cref{alg:flipflop} for pseudocode. More specifically, when $L_1,\ldots,L_k$ are fixed, we minimise the (restricted) objective function over $B$ subject to physical constraints on $A$ and $B$ using gradient-based methods. When $B$ is fixed, minimising $L_1,\ldots,L_k$ amounts to fitting a Lindbladian to each $E_i=B\tilde{g}^{-1}\tilde{P}_iB^{-1}$, which we can handle using either the Convex Solve or the Alternating Projections methods described in \Cref{subsec:TTE}. Note that we adopt $L_i$ at the current iteration of Gate Set Flip-Flop as the input best-guess Lindbladian for the next run of Alternating Projections. Since $E_i$ being Markovian for all $i\in\{1,\ldots,k\}$ correspond to a solution with objective value zero, which is essentially impossible in the presence of statistical error, the $E_i$'s will be non-Markovian throughout the algorithm.

\begin{algorithm}\small
\caption{Gate Set Flip-Flop}
{\label{alg:flipflop}}
\SetKwInOut{Input}{input}
\SetKwInOut{Output}{output}
\SetKwFunction{KwLindbladFit}{LindbladFit}
\SetKw{KwSubjTo}{subject to}
\SetKw{Kwminimise}{minimise}
\Input{Tomographic measurement data $(\tilde{g},\tilde{P}_1,\ldots,\tilde{P}_k)$ where $\tilde{g}$ is the measured Gram matrix and $\tilde{P}_i$ is the data for the $i$-th quantum process; initial-guess Lindbladians for each gate $(L_1^{\text{ideal}},\ldots,L_k^{\text{ideal}})$ following the same indexing as for tomographic data; initial guess $B^{\text{ideal}}$ for tomographic preparation settings; stopping condition $S$}
\Output{Estimated Lindbladians $(L_1,\ldots,L_k)$, preparation settings $B$, and measurement settings $A$.}
$B \gets B^{\text{ideal}}$, $L_1\gets  L_1^{\text{ideal}}$, \ldots, $L_k\gets L_k^{\text{ideal}}$

\While{Stopping condition $S$ not met}{
\tcp{Here we optimise the Lindbladians first. This order can be swapped, as discussed in the main text.}

\For{$i = 1$ \KwTo $k$}{
$E_i \gets B \tilde{g}^{-1} \tilde{P}_i B^{-1}$

$L_i \gets $ \KwLindbladFit{$E_i$, $L_i$}

\tcp{LindbladFit is a function that takes as input a tomographic channel estimate (and optionally a best-guess Lindbladian) and outputs an estimate for a generating Lindbladian.}
}
$B\gets$ \Kwminimise $h(B,L_1,\ldots,L_k)$ \KwSubjTo $B$ and $A=gB^{-1}$ physical \\
where $h$ is given by \Cref{eq:h_def}
}
\Return{$B,L_1,\ldots,L_k$}
\end{algorithm}

The initial guess used to seed the optimisation is $(B^{\text{ideal}},L_1^{\text{ideal}},\ldots,L_k^{\text{ideal}})$. This is a sensible choice because we seek solutions in the vicinity of $(B^{\text{ideal}},L_1^{\text{ideal}},\ldots,L_k^{\text{ideal}})$ in accordance with the weak noise assumption. In addition, we need to choose between starting with minimising $B$ or $L_1,\ldots,L_k$. Minimising $B$ first against the ideal gates $E_i^\text{ideal}$ is similar to gauge optimisation in standard GST, except with physical constraints on $B$ and a different objective function. For each gate $i$, minimising $L_i$ first against $B^{\text{ideal}}\tilde{g}^{-1}\tilde{P}_i(B^\text{ideal})^{-1}$ resembles linear inversion process tomography where SPAM errors are assumed to be non-existent. The choice that is more preferable depends on whether we have prior knowledge about the noise characteristics in the input. If we believe that state preparation errors are stronger than gate noises, then minimising $B$ first could perform the best. Conversely, if we have the prior knowledge that state preparation errors are considerably weaker than gate noises, then minimising the Lindbladians first could be preferable.

Recall the fact that the LSE function is monotonically increasing in each of its inputs, meaning that the objective values of the iterates generated by Gate Set Flip-Flop form a monotonically non-increasing sequence that is bounded from below. Thus, the sequence of objective values necessarily converges to a fixed point. However, by the nature of local minimisation, the Gate Set Flip-Flop algorithm is not guaranteed to find an optimal solution. 

The main difference between the gate set Lindbladian fitting problem and standard gauge optimisation in GST (see \Cref{eq:gst_objective}) lies in replacing each $E_i^{\text{ideal}}$ with $e^{L_i}$ in the objective function. Operationally speaking, during each iteration of Gate Set Flip-Flop, the gauge $B$ is minimised w.r.t. a set of more accurate estimates of the noisy gates than the ideal gates $E_i^{\text{ideal}}$. This can improve the accuracy of the estimated SPAM errors encoded in $B$ by reducing the amount of gate noise erroneously pulled into SPAM errors.

\FloatBarrier

\section{Synthetic Testing of Alternating Projections and Convex Solve} \label{sec:Lindblad_synthetic}

We first describe how we generate our test cases. We choose a $2$-qubit gate and Markovian gate noise combination, for example, CNOT with amplitude damping noise. This specifies the ideal and ground truth transfer matrices $E^\text{ideal}$ and $E^*$ respectively. We then feed $E^*$ through simulated quantum process tomography to obtain a matrix $\tilde{E}$. Finally, $\tilde{E}$ is projected to CPTP to generate an input transfer matrix $E$. Multiple runs of process tomography on the same $E^*$ will generate different $\tilde{E}$ and $E$ instances since the instantiation of statistical error in each run will be different. Since the gate noise is Markovian, the difference between $E$ and $E^*$ is solely due to statistical error. In this section, all inputs are generated using $10^4$ tomographic shots which is also the setting we use for our tests on real quantum hardware. We emphasise that during synthetic testing, both $E^{\text{ideal}}$ and $E^*$ are chosen and known.

Before presenting aggregate test results, we first highlight the difference between Convex Solve and Alternating Projections using a CNOT gate example. For CNOT gates, it is possible that at the level of statistical error induced by $10^4$ tomographic shots and under weak gate noises, the closest Lindbladian to every physically relevant $A\in\log(E)=\{A:e^A=E\}$ exponentiates to a channel very far away from $E$. \Cref{fig:stat_error_a,fig:stat_error_b,fig:stat_error_c} present a synthetic noisy CNOT gate example where (1)~the Lindbladian fitting problem can be solved exactly using Convex Solve when there is no statistical error in the input transfer matrix, (2)~the Convex Solve method estimates none of the Hamiltonian, the $\gamma$, or the jump operator accurately when the input suffers from a realistic level of statistical error, and (3)~Alternating Projections outputs a Lindbladian that estimates all of the Hamiltonian, the $\gamma$, and the jump operator accurately for the same statistically noisy input. For the Alternating Projections method to work for this example, it is necessary to choose the precision parameter $\beta>0$ to merge all the eigenvalues of $E$ close to being real negative into one degenerate eigenspace, hence projecting to an enlarged search space $\widetilde{\log}(E)$. All Lindbladian canonical decompositions are computed using the procedure described at the end of \Cref{subsec:Lindblad}.

For our synthetic benchmark, we focus on the two success criteria defined in \Cref{sec:criteria}. Recall that we say a test run passes Success 1 if $\lVert e^L-E\rVert\leq\lVert E-E^*\rVert$ holds, i.e., whether the output Lindbladian $L$ closely fits the input transfer matrix $E$. Since $\lVert e^L-E\rVert$ is the objective function the algorithm explicitly seeks to minimise, $\lVert E-E^*\rVert$ is the objective value attained by the ground truth Lindbladian $L^*$ besides also quantifying the amount of statistical error in the input. Success 2, the second more stringent success criterion, checks for $\lVert e^L-E^*\rVert\leq\lVert E-E^*\rVert$, i.e., whether the output Lindbladian $L$ closely fits the ground truth transfer matrix $E^*$. We view passing Success~2 as a bonus since the algorithm is not intended to remove statistical error from $E$ -- we would realistically expect $\lVert e^L-E^*\rVert\approx\lVert E-E^*\rVert$. Passing Success~1 guarantees at a bare minimum that $\lVert e^L-E^*\rVert\leq 2\lVert E-E^*\rVert$. See \Cref{tab:synthetic_aggregate} for a table of aggregate results and see \Cref{fig:CNOT_distribution} for a histogram of the distribution of the $\lVert e^L-E\rVert$ and $\lVert e^L-E^*\rVert$ values for the 20 instances of CNOT gate with coherent $X$ and dephasing noise tested.

By examining the canonical decompositions of Lindbladians returned by our algorithm, we found cases where our algorithm's predicted jump operators visibly differed from the ground truth jump operators even when the output passed Success~2 (see \Cref{fig:synthetic_bad_jump} for an example). Strictly speaking, this does not mean that the algorithm has failed since the algorithm's goal is to find \textit{a} Lindbladian that approximately generates the input dynamics, and the output Lindbladian $L$ does attain a small $\lVert e^L-E\rVert$ value.

\begin{figure}[t]
    \centering
    \includegraphics[width=\textwidth]{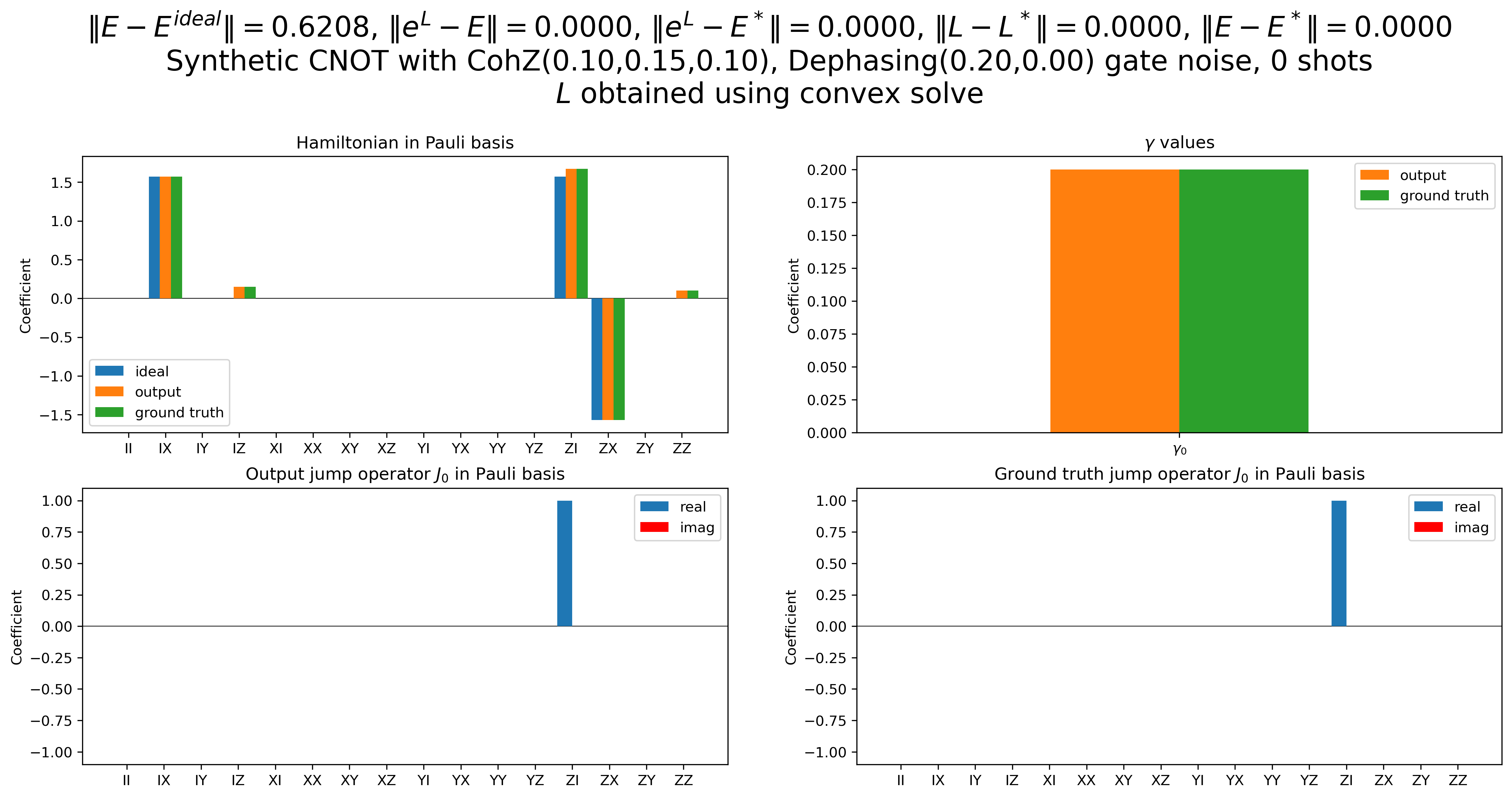}
    \caption{CNOT gate with coherent $ZI$, $IZ$, and $ZZ$ errors and dissipative $ZI$ error. $E^\text{ideal}$ is the transfer matrix of the ideal noiseless CNOT gate, and $\lVert E-E^{\text{ideal}}\rVert$ quantifies the total amount of gate and statistical noises in the input. In this example, there is no statistical error in the input $E$, so $\lVert E-E^*\rVert=0$. An exhaustive search over the branches of $\log(E)$ finds a Lindbladian $L$ satisfying $\lVert e^L-E\rVert=\lVert e^L-E^*\rVert=0$ up to numerical errors. The canonical decompositions are computed using the procedure described at the end of \Cref{subsec:Lindblad}.}
    \label{fig:stat_error_a}
\end{figure}

\begin{figure}[t]
    \centering
    \includegraphics[width=\textwidth]{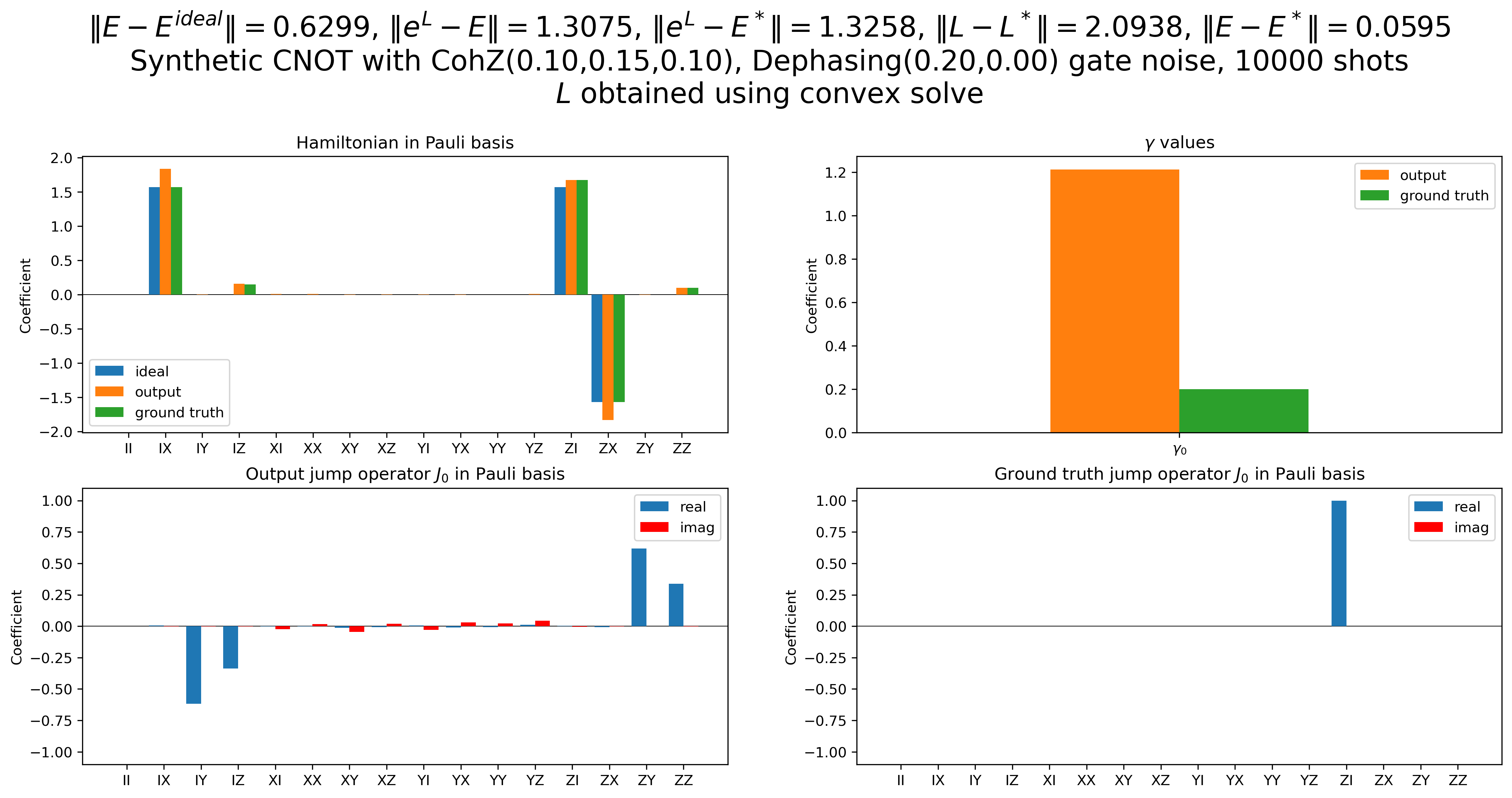}
    \caption{For the same noisy CNOT gate considered in \Cref{fig:stat_error_a}, at the level of statistical noise induced by $10^4$ shots, an exhaustive search over the branches of $\log(E)$ is no longer able to find a Lindbladian that exponentiates closely to the input. The output Lindbladian $L$ fails to capture accurately any of the Hamiltonian, the $\gamma$, or the jump operator. Notice that $e^L$ is further away from the input $E$ than the zero-knowledge estimate $E^{\text{ideal}}$ where $E^{\text{ideal}}$ is the transfer matrix of the ideal noiseless CNOT gate.}
    \label{fig:stat_error_b}
\end{figure}

\begin{figure}[t]
    \centering
    \includegraphics[width=\textwidth]{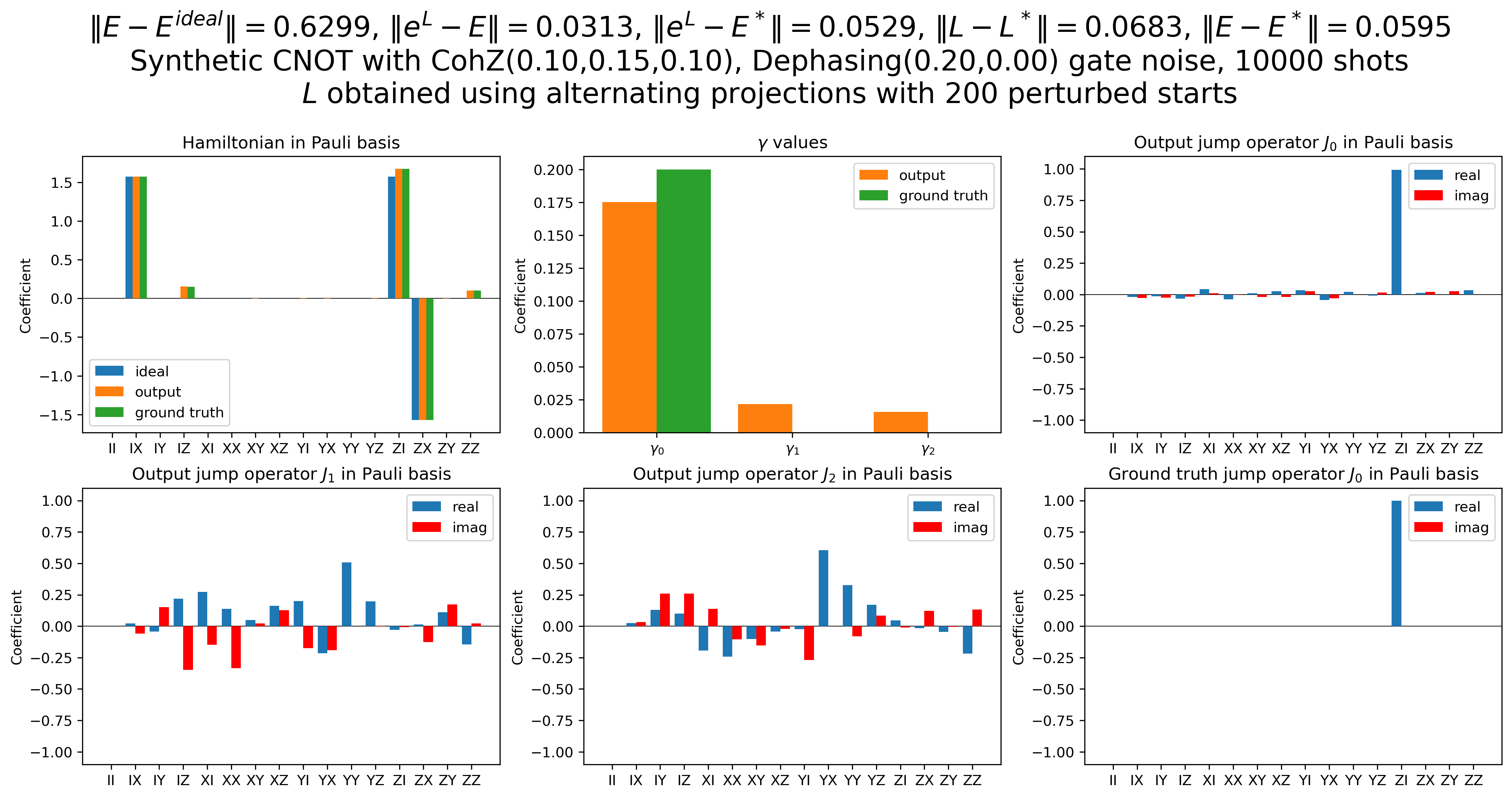}
    \caption{For the same input $E$ considered in \Cref{fig:stat_error_b}, our algorithm finds a Lindbladian that exponentiates very closely to the input. The output Lindbladian $L$ captures all of the Hamiltonian, $\gamma_0$, and the jump operator $J_0$ accurately. While the algorithm erroneously predicts the jump operators $J_1$ and $J_2$ due to statistical error, their corresponding $\gamma_1$ and $\gamma_2$ coefficients are small.}
    \label{fig:stat_error_c}
\end{figure}

\begin{table}[t]
\small
    \centering
    \begin{tabular}{c|c|c|c|c|c}
    \multicolumn{6}{c}{Synthetic testing of the Alternating Projections method}\\ \hline
         Success 1, 2 & Overrotation & Coh$X$ & Bitflip & Coh$X$ Dephasing & Coh$X$ AmpDamp Bitflip\\ \hline
        CNOT & 20, 12 & 20, 20 & 20, 20 & 20, 0 & 20, 2\\ \hline
        ISWAP & 20, 8 & 20, 16 & 20, 17 & 20, 8 & 20, 2\\ \hline
        $X\otimes H$ & 20, 13 & 20, 18 & 20, 19 & 20, 16 & 20, 20\\ \hline

        \multicolumn{6}{c}{\vspace{-0.2cm}}\\

        Success 1, 2 & AmpDamp & Coh$Z$ & Dephasing & Coh$Z$ Bitflip & Coh$Z$ AmpDamp Dephasing \\ \hline
        CNOT & 20, 6 & 20, 20 & 20, 19 & 20, 5 & 20, 8\\ \hline
        ISWAP & 20, 10 & 20, 17 & 20, 20 & 20, 9 & 20, 20\\ \hline
        $X\otimes H$& 20, 20 & 20, 17 & 20, 16 & 20, 1 & 20, 6\\ \hline
    \end{tabular}

    \vspace{0.4cm}
    
    \begin{tabular}{c|c|c|c|c|c}
    \multicolumn{6}{c}{Synthetic testing of the Convex Solve method}\\ \hline
         Success 1, 2 & Overrotation & Coh$X$ & Bitflip & Coh$X$ Dephasing & Coh$X$ AmpDamp Bitflip\\ \hline
        $\sqrt{X}\otimes I$ & 20, 16 & 20, 18 & 20, 12 & 20, 20 & 20, 20\\ \hline
        $T\otimes I$ & 20, 20 & 20, 20 & 20, 20 & 20, 20 & 20, 20\\ \hline
        $I\otimes I$ & /, / & 20, 20 & 20, 20 & 20, 20 & 20, 20\\ \hline

        \multicolumn{6}{c}{\vspace{-0.2cm}}\\

        Success 1, 2 & AmpDamp & Coh$Z$ & Dephasing & Coh$Z$ Bitflip & Coh$Z$ AmpDamp Dephasing \\ \hline
        $\sqrt{X}\otimes I$ & 20, 20 & 20, 19 & 20, 20 & 20, 19 & 20, 20\\ \hline
        $T\otimes I$ & 20, 20 & 20, 20 & 20, 20 & 20, 20 & 20, 20\\ \hline
        $I\otimes I$ & 20, 20 & 20, 20 & 20, 20 & 20, 20 & 20, 20\\ \hline
    \end{tabular}
   
    \caption{Each gate and noise model combination is tested on $20$ input instances generated using simulated quantum process tomography with $10^4$ shots. Each entry records the number of runs (out of $20$) that passed the Success 1, 2 criteria $\lVert e^L-E\rVert\leq \lVert E-E^*\rVert$, $\lVert e^L-E^*\rVert\leq \lVert E-E^*\rVert$. When using Alternating Projections, the algorithm tries 500 perturbed starts. The Alternating Projections method consistently succeeds with respect to the Success 1 criterion for all the gate and noise combinations tested. The combined success rates are $600/600=100\%$ and $385/600=61.4\%$ for Success criteria 1 and 2 respectively. The Convex Solve method also consistently succeeds for all the cases tested, in particular the idling gate $I\otimes I$. The noise strengths $\lVert L^*-L^{\text{ideal}}\rVert$ range from $0.089$ to $0.355$.}
    \label{tab:synthetic_aggregate}
\end{table}

\begin{figure}[t]
\centering
    \includegraphics[scale=0.8]{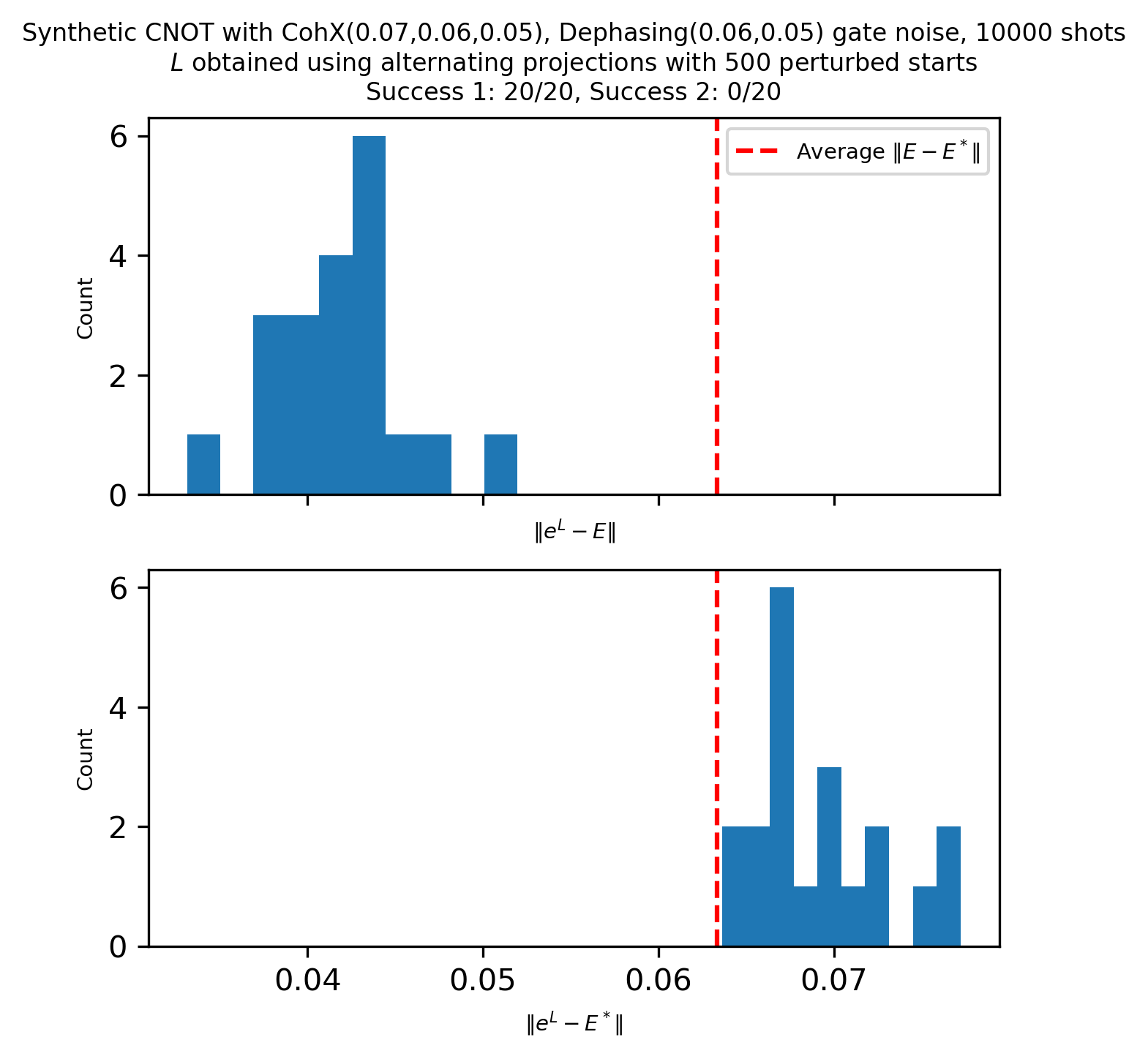}
    \caption{Distribution of the $\lVert e^L-E\rVert$ (top) and $\lVert e^L-E^*\rVert$ (bottom) values for the 20 instances of CNOT with coherent $X$ and dephasing noise tested. The algorithm is overfitting statistical error to some extent.}
    \label{fig:CNOT_distribution}
\end{figure}

\begin{figure}[t]
\centering
    \includegraphics[width=\textwidth]{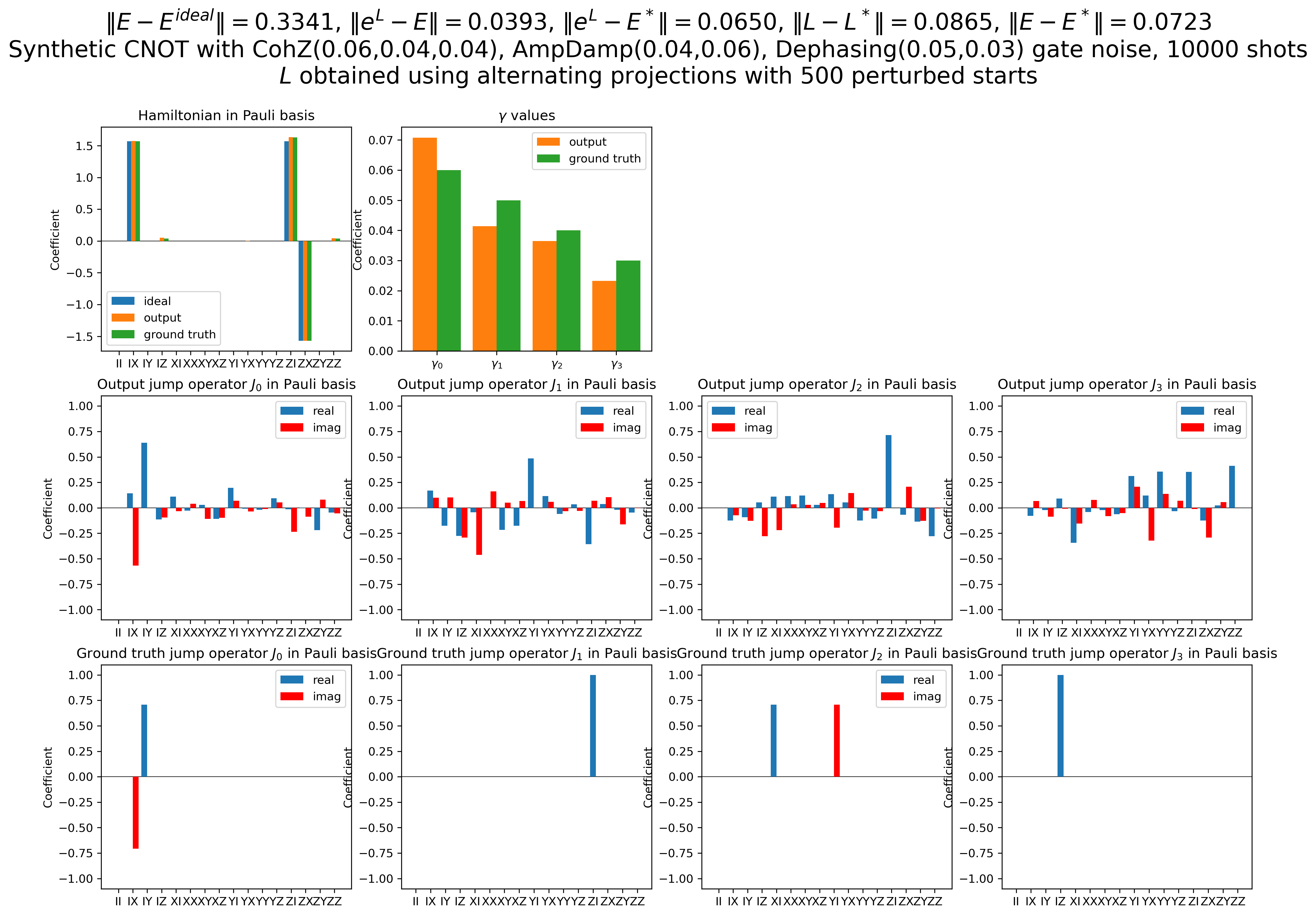}
    \caption{Canonical decomposition of a CNOT with coherent $Z$, amplitude damping, and dephasing noise instance. The algorithm succeeded on the input according to both of our success criteria. The Lindbladian found by the algorithm matches the ground truth Hamiltonian accurately and the $\gamma$ values are well-estimated. However, the predicted jump operators are difficult to interpret.}
    \label{fig:synthetic_bad_jump}
\end{figure}

\begin{figure}[t]
\centering
\subfloat[]{\includegraphics[width=.5\textwidth]{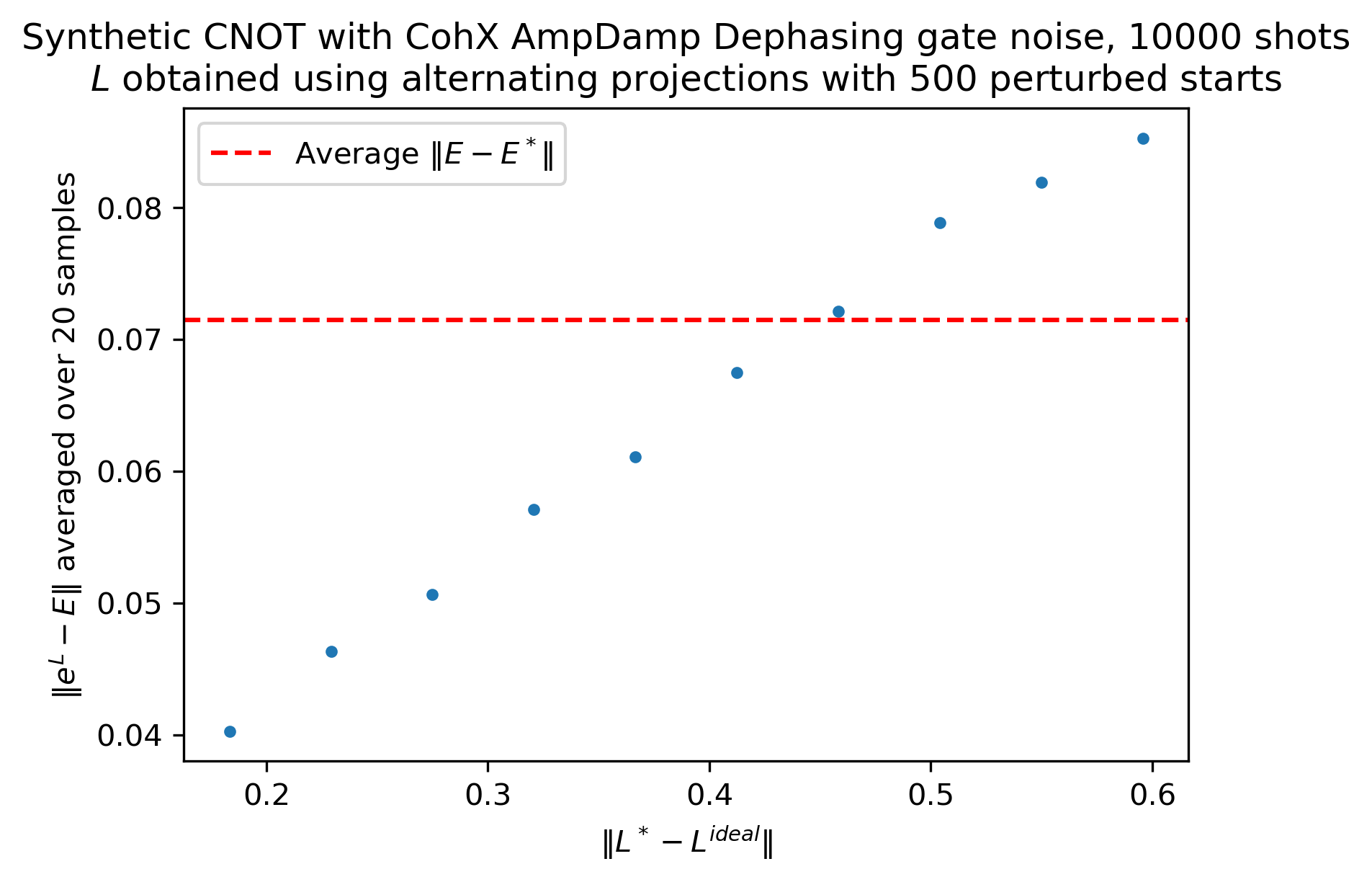}}
\subfloat[]{\includegraphics[width=.5\textwidth]{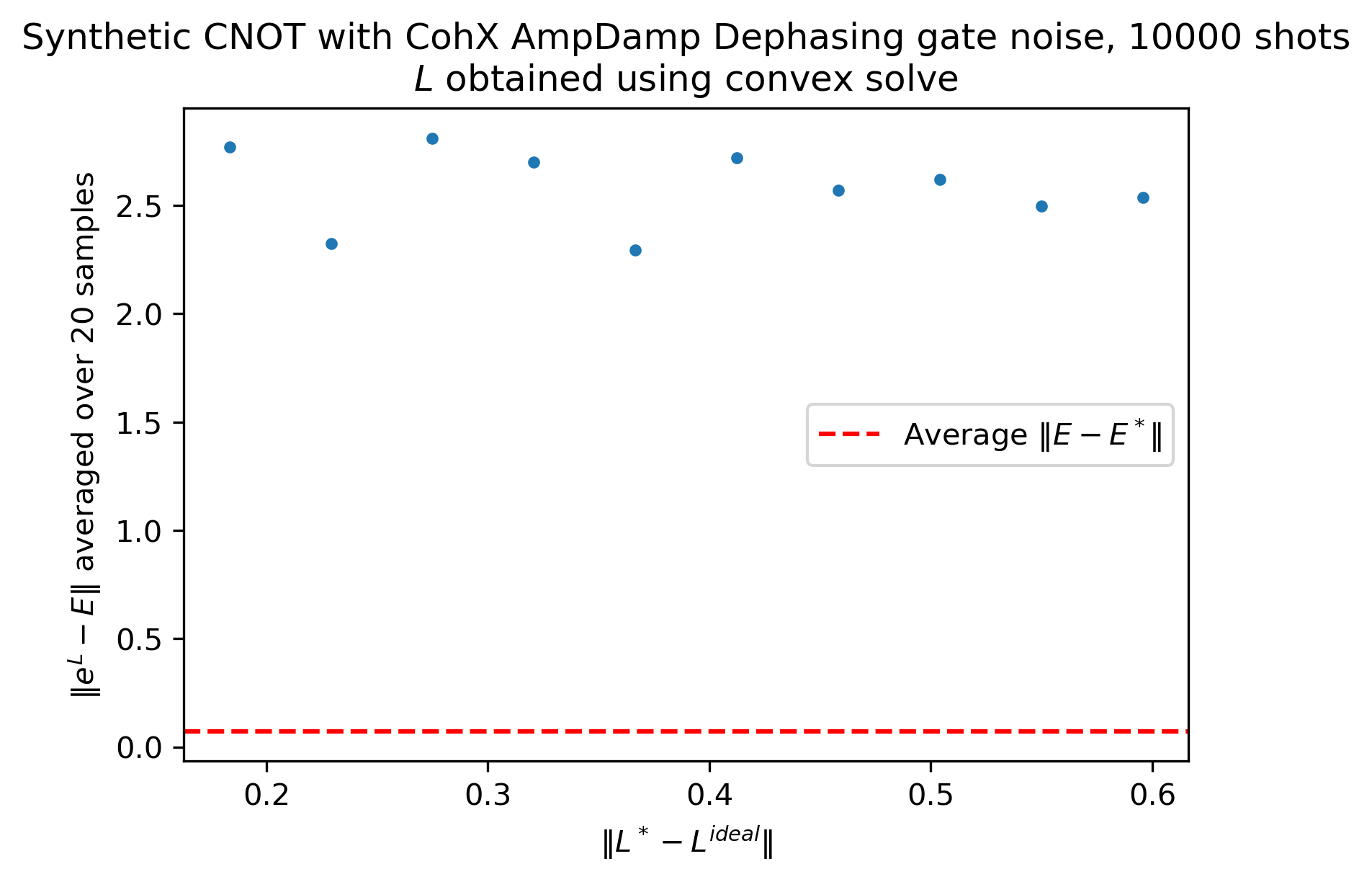}}\\
\subfloat[]{\includegraphics[width=.5\textwidth]{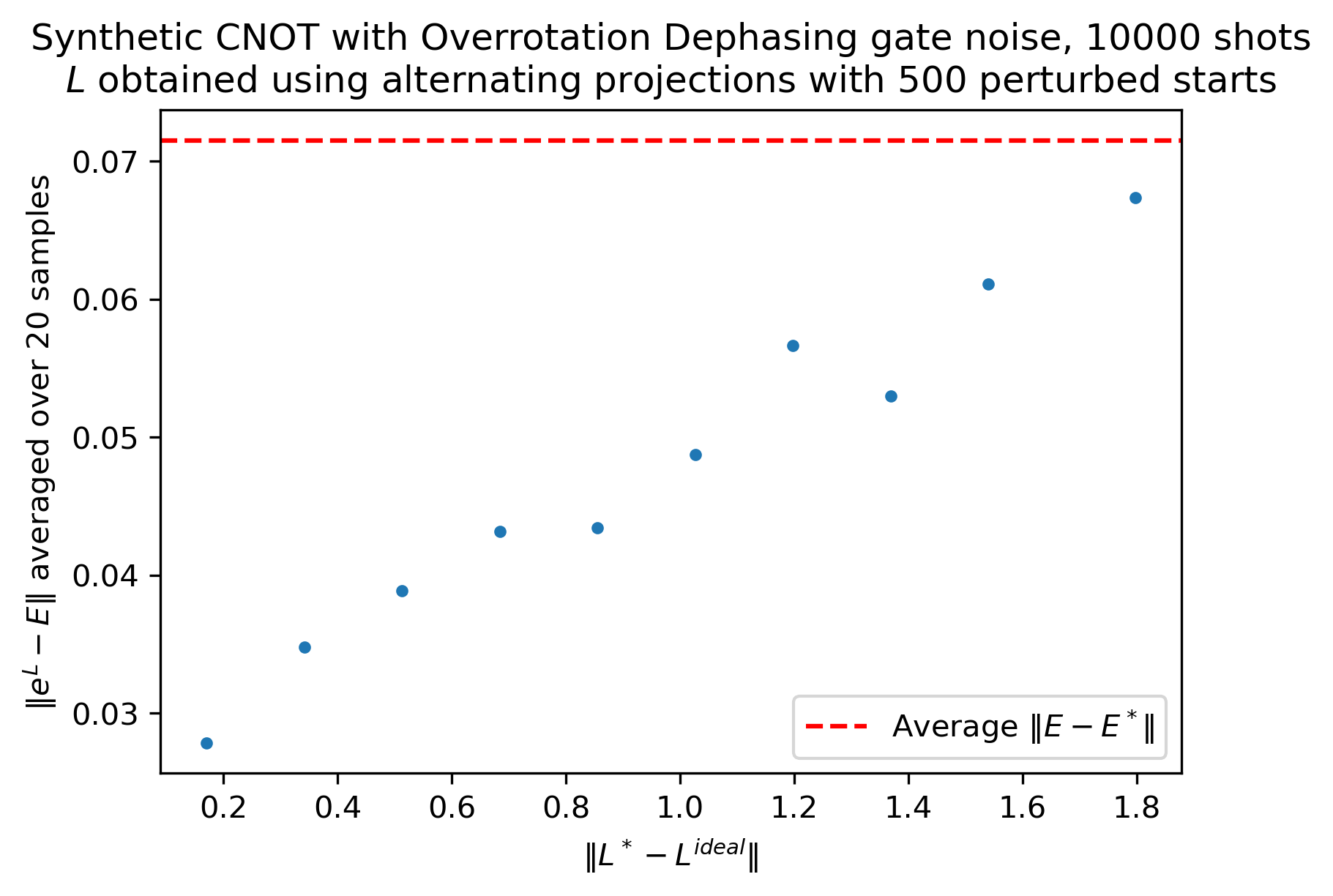}}
\subfloat[]{\includegraphics[width=.5\textwidth]{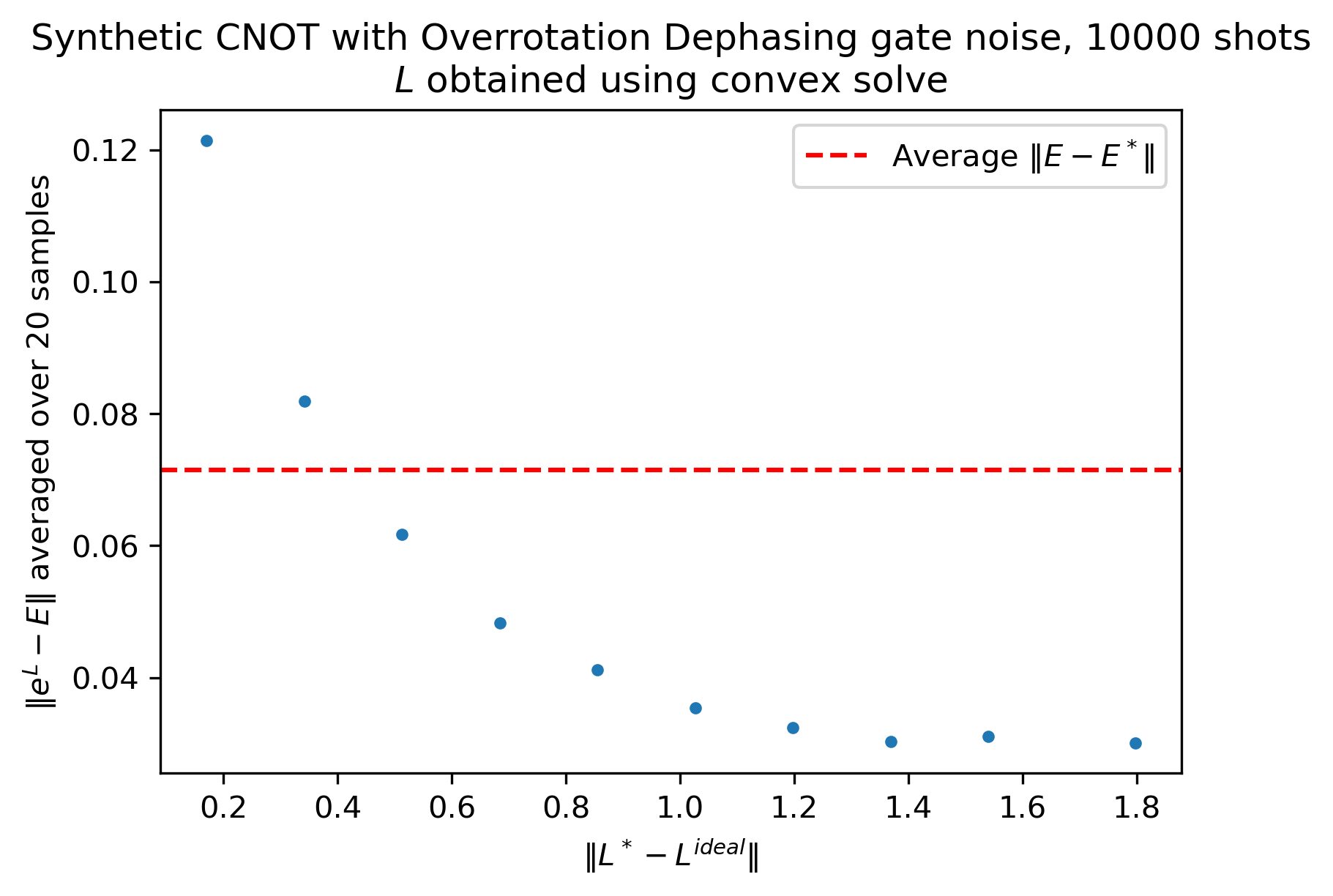}}
\caption{CNOT with increasing strengths of (a)(b) coherent $X$, amplitude damping, and dephasing gate noise and (c)(d) overrotation and dephasing gate noise. (a) The Alternating Projections method passes Success 1 for this test case until $\lVert L^*-L^{\text{ideal}}\rVert$ reaches roughly $0.45$. (b) Even as the noise becomes stronger, the inputs $E$ always have eigenvalues close to being real negative, and the Convex Solve method consistently fails. (c) For this noise model, the Alternating Projections method succeeds for the entire range of noise strengths tested. (d) As the noise becomes stronger, $E$ stops having eigenvalues close to being real negative, and the Convex Solve method starts working.}
\label{fig:CNOT_c2}
\end{figure}

Next, we investigate the weak noise assumption $\lVert L^*-L^{\text{ideal}}\rVert\leq c_2$ made near the beginning of \Cref{subsec:Lindblad}. How small does $c_2$ need to be to ensure that the Alternating Projections method succeeds? To put the $\lVert L^*-L^{\text{ideal}}\rVert$ values in perspective, we can evaluate the average gate fidelity between $\Eideal$ (which corresponds to a unitary gate) and $E^*$ using the formula
\begin{equation}
F_{\mathrm{avg}} (\Eideal,E^*) = \frac{\frac{1}{d}\Tr[(\Eideal)^\dagger E^*] + 1}{d + 1}\label{eq:avggatefidelity}
\end{equation}
adapted from~\cite{horodecki1999general,nielsen2002simple} for transfer matrices where $d=4$ is the dimension of the system. In \Cref{fig:CNOT_c2} (a), we consider running the Alternating Projections algorithm on CNOT with increasing strengths of coherent $X$, amplitude damping, and dephasing noise. The results indicate that our algorithm succeeds w.r.t. the Success 1 criterion for this specific gate and noise model combination as long as $\lVert L^*-L^{\text{ideal}}\rVert\leq 0.45$, at which point $F_{\mathrm{avg}} (\Eideal,E^*)$ is roughly $94\%$. For this gate and noise model combination, the inputs $E$ always have eigenvalues close to being real negative. Hence, we expect the Convex Solve method to fail on this test case, which is confirmed in \Cref{fig:CNOT_c2} (b). In \Cref{fig:CNOT_c2} (c) and (d), we repeat the same experiment on CNOT with overrotation and dephasing noise. The Alternating Projections method's performance is less sensitive to the noise strength in this case, and it succeeds even when $\lVert L^*-L^{\text{ideal}}\rVert=1.8$, which translates to an average gate fidelity of $83.5\%$ between $E^{\text{ideal}}$ and $E^*$. For CNOT with overrotation and dephasing noise, the $-1$ eigenvalues of $E^{\text{ideal}}$ steadily rotate away from the real negative axis as the noise strength increases, and the Convex Solve method begins working when the average gate fidelity between $E^{\text{ideal}}$ and $E^*$ drops below roughly $96\%$. 

Lastly, we examine to what extent does the Convex Solve method solve the Lindbladian fitting problem with statistical error when $E^{\text{ideal}}$ corresponds to the identity gate $I\otimes I$. In \Cref{fig:II_trivial}, we plot the results of running the Convex Solve method on noisy identity gates with statistical error with various noise models and increasing noise strengths where the average gate fidelities between $E^\text{ideal}$ and $E^*$ range from $99.9\%$ to $44.6\%$. For each noise model and noise strength combination, we run the Convex Solve method on 20 input instances generated using the procedure described at the beginning of this section. The obtained $\lVert e^L-E\rVert$ values are averaged to calculate each data point. Our results indicate that the Convex Solve method succeeds consistently in the parameter regimes tested. 

\begin{figure}[t]
\centering
    \includegraphics[width=\textwidth]{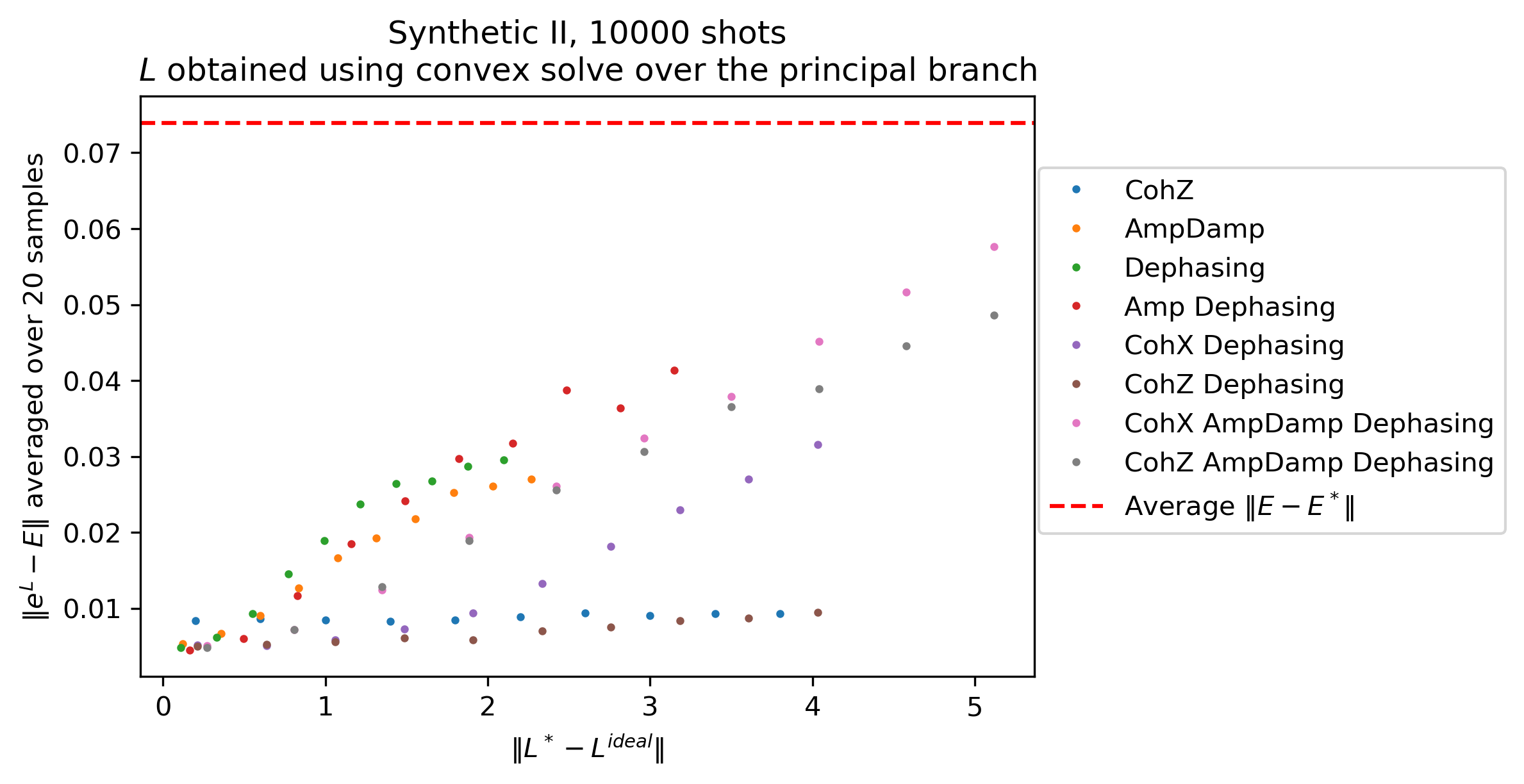}
    \caption{Each noise model is tested with 10 levels of increasing noise strengths. Each data point corresponds to the average $\lVert e^L-E\rVert$ value of 20 input instances. Across all the data points, the average gate fidelities between $E^\text{ideal}$ and $E^*$ range from $99.9\%$ to $44.6\%$.}
    \label{fig:II_trivial}
\end{figure}

Our implementation runs significantly faster than the algorithm implemented in prior work \cite{Onorati2023fittingquantumnoise}. For comparison, the same ISWAP analysis reported in \cite{Onorati2023fittingquantumnoise}, which took two weeks to run, can be performed in under two minutes using our implementation. The speedups mainly come from parallelising the search over the branches of complex logarithm, skipping over clearly unphysical branches entirely, and while in a correct branch, the Alternating Projections method requiring fewer calls to the convex optimiser to find a good fitting Lindbladian. We use the SCS package \cite{scs} to solve the convex optimisation problem of projecting a matrix to a Lindbladian. Our current implementation spends by far the most CPU time in the convex optimisation step in the inner most loop.

\section{Synthetic Testing of Gate Set Flip-Flop}\label{sec:flipflop}
In this section, we report synthetic test results for the Gate Set Flip-Flop algorithm. We first describe how we generate our test cases. The gate set consists of the six $2$-qubit gates $\{\text{CNOT},\sqrt{X}\otimes I,I\otimes\sqrt{X},T\otimes I,I\otimes T,\text{ISWAP}\}$. A SPAM error is chosen randomly from amplitude damping, bitflip, incoherent $Y$, and dephasing for each of the $16$ state preparation settings and $16$ measurement settings. For state preparation errors, the noise strength coefficient $\gamma$ for each jump operator in \Cref{eq:Lindbladian_form}  is drawn from the normal distribution with mean $0.07$ and standard deviation $0.007$. For measurement errors, the $\gamma$ value for each jump operator is sampled randomly from the normal distribution with mean $0.12$ and standard deviation $0.012$. In effect, we try to model a quantum device whose measurements are noisier than state preparations. For each gate in the gate set, a gate noise with normally distributed noise strengths is randomly chosen from overrotation with bitflip, overrotation with dephasing, coherent $Z$ with amplitude damping and dephasing, coherent $Z$ with bitflip, and coherent $X$ with dephasing. For every gate in the gate set, a simulated quantum process tomography experiment is carried out with $10^4$ shots while the Gram matrix is estimated using $10^5$ shots. We randomly generated 1000 test cases using the settings described above. 

For each test case, the Gate Set Flip-Flop algorithm is run for three iterations on the objective function \Cref{eq:h_def}. We start by minimising the Lindbladians first since the state preparations errors are weaker than the gate noises. The Lindbladian fitting problem is solved using the Alternating Projections method for the CNOT and ISWAP gates and the Convex Solve method for the other gates. Each Alternating Projections run tries 400 perturbed starts. The constrained minimisation w.r.t. $B$ is performed using the SLSQP algorithm \cite{SLSQP} implemented in NLopt \cite{NLopt}. The physical constraints on the eigenvalues and trace of each column of $B$ and the eigenvalues of each row of $A$ are enforced with a slack of $0.001$. 

Analogous to the criterion Success~1 defined in \Cref{sec:criteria}, we consider a run of the Gate Set Flip-Flop algorithm to be successful if the algorithm's output $(B,L_1,\ldots,L_k)$ satisfies $f_{\max}(B,L_1,\ldots,L_k)\leq f_{\max}(B^*,L_1^*,\ldots,L_k^*)$ -- see \Cref{eq:f_max_def} -- where $B^*$ encodes the ground truth state preparation settings and $L_1^*,\ldots,L_k^*$ are the ground truth Lindbladians for the noisy gates. In \Cref{fig:flip_flop_bulk}, we plot the distributions of $f_{\max}(B,L_1,\ldots,L_k)$ versus $f_{\max}(B^*,L_1^*,\ldots,L_k^*)$ values for the 1000 randomly generated test cases after the zeroth, first, second, and third iterations. Note that $f_{\max}(B,L_1,\ldots,L_k)$ after the zeroth iteration is just $f_{\max}(B^{\text{ideal}},L_1^\text{ideal},\ldots,L_k^\text{ideal})$. 
Each flip-flop iteration consists of one round of gauge optimisation and one round of Lindbladian fitting. We observe that our Flip-Flop algorithm succeeds consistently after three iterations. It failed on just a handful of test cases because the slack for the physical constraints was too tight for these inputs, causing the SLSQP optimiser to get stuck at $B^{\text{ideal}}$. Raising the slack to a larger value such as $0.005$ resolves the issue.

\begin{figure}[t]
\centering
    \includegraphics[width=\textwidth]{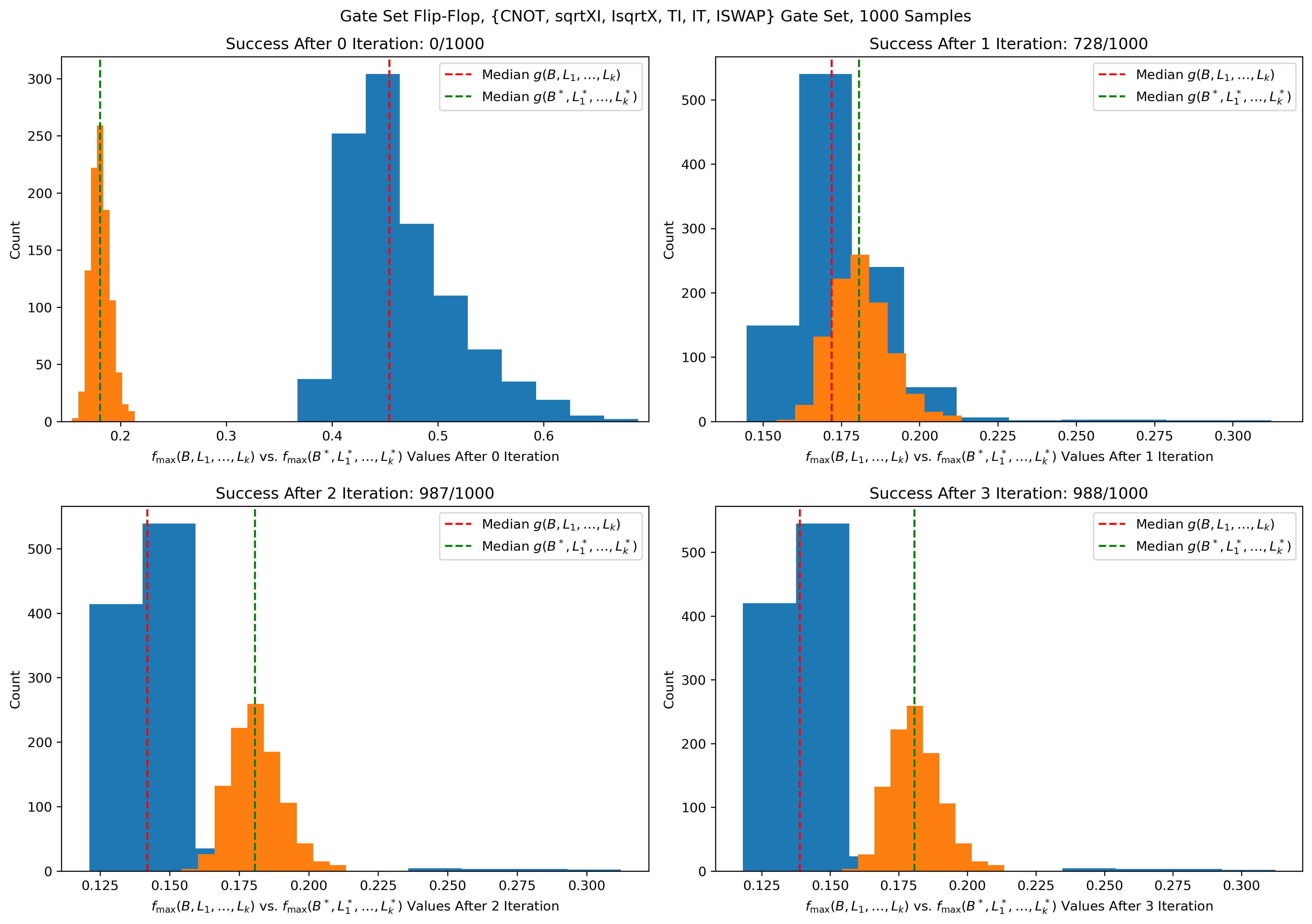}
    \caption{Synthetic test results for the Gate Set Flip-Flop algorithm on 1000 randomly generated test cases. The blue histograms plot the $f_{\max}(B,L_1,\ldots,L_k)$ values while the orange histograms plot the $f_{\max}(B^*,L_1^*,\ldots,L_k^*)$ values. A run of the algorithm is deemed successful if $f_{\max}(B,L_1,\ldots,L_k)\leq f_{\max}(B^*,L_1^*,\ldots,L_k^*)$.}
    \label{fig:flip_flop_bulk}
\end{figure}

\FloatBarrier

\newpage

\section{Analysis of Experimental Data} \label{sec:exp}
In this section we present example results produced by our algorithms when run on data obtained from real quantum computing hardware. The tomography circuits were executed on devices available via the cloud-based IBM Quantum platform~\cite{IBMQ} using the \textit{Qiskit} software development kit~\cite{Qiskit}. Before describing the demonstration in detail, we note that in general there can be several reasons why noisy gates might not be exactly generated by a time-independent Lindbladian. 
For one, the process may be more accurately modelled by a time-dependent Lindbladian, since some gates, particularly those involving two-qubit interactions, may be implemented via a composite pulse sequence \cite{QiskitPulse_Alexander_2020}. Even for simple pulses, the amplitude must be ramped on and off in some finite time, so the true Lindbladian of the process must be time-dependent. Also, the drift in the relative strength of noise components can yield to apparent non-Markovianity in tomographic data, since this can lead to a distortion in the eigenvalue structure of the process. Finally, there may of course be genuine non-Markovian noise effects, for example, if there is an unwanted coupling to a defect acting as a two-level system \cite{TLS_Burnett_2014,TLS_de_Graaf_2020}. Nevertheless, a noiseless unitary operation is always compatible with a time-independent Lindbladian, so in the case where noise is relatively weak, an effective time-independent Lindbladian may be useful as a model for quantitatively and qualitatively understanding the noise profile for a given target process. Moreover, if our algorithm does not return a Lindbladian that exponentiates close to the tomographic data, it provides evidence of plausible non-Markovian noise effects present in the device.

\subsection{Gate Set Analysis on \textit{ibm\_perth}}
We first show results obtained from the 7-qubit device known as \textit{ibm\_perth}. Gate set tomography was carried out for the set 
\begin{align}
\mathcal{G} = &\{\1\otimes \1, \text{CNOT}, R_{ZX}(0.5), R_{ZX}(0.5)^3, R_{ZX}(0.5)^5, \\
                      & \quad R_{ZZ}(0.5), R_{ZZ}(0.5)^3, R_{ZZ}(0.5)^5, T\otimes\1, \1\otimes T, \sqrt{X}\otimes\1, \1\otimes \sqrt{X}\},
\end{align}
implemented on the connected qubits labelled $3$ and $5$ using IBM's indexing. As suggested above, there are different factors explaining why we might expect the ground truth to depart from a time-independent Markovian process, and we exemplify this with this gate set. On this device, the CNOT gate was implemented as an echoed cross-resonance gate~\cite{Rigetti2010FullyMU,Chow_2011,QiskitPulse_Alexander_2020}, composed with shorter single-qubit pulses -- in effect an $R_{ZX}(\pi/2)$ gate with an internal dynamical decoupling sequence to reduce coherent error and uncontrolled couplings with other subsystems, where 
\begin{equation}
R_{ZX}(\theta) = \exp[-i \frac{\theta}{2} Z \otimes X  ]. 
\end{equation}
The CNOT was provided as a native gate on \textit{ibm\_perth}, meaning that the device has been calibrated to maximise the fidelity of this gate. However,  \textit{Qiskit} enables implementation of custom pulse sequences, and an $R_{ZX}$ gate can be implemented for arbitrary angles by extracting the pulse sequence calibrated for the $R_{ZX}(\theta)$ rotation underlying the CNOT gate and modifying the pulse duration (or amplitude) accordingly. In the gate set $\mathcal{G}$, we contrast this with $ R_{ZZ}(\theta) =  \exp[-i (\theta/2) Z \otimes Z ]$ gates which we decompose as a $Z$-rotation sandwiched by two CNOT gates. The single-qubit $\sqrt{X}$ and $T$ gates are implemented by a simple pulse and a change of frame respectively, and are mainly included in the set to assist with gauge optimisation. In each case, the noiseless process is unitary and can therefore be generated by a time-independent Lindbladian with no dissipative part. In reality, the tomographed channel could depart from this model to varying degrees depending on the gate. We may consider this a stress test of our algorithms, as it can be the case that there is no time-independent Lindbladian model that exponentiates close to the data. For our tomographically complete set of initial states, we prepare all two-qubit tensor products of four of the single-qubit Pauli eigenstates $\{\ket{0},\ket{1},\ket{+},\ket{+_i}\}$, defined in the ideal case as
\begin{equation}
\ket{0} = \begin{bmatrix}
1 \\ 
0
\end{bmatrix},
\quad 
\ket{1} = \begin{bmatrix}
0 \\ 
1
\end{bmatrix},
\quad
\ket{+} = \frac{1}{\sqrt{2}}\begin{bmatrix}
1 \\ 
1
\end{bmatrix},
\quad
\ket{+_i} = \frac{1}{\sqrt{2}}\begin{bmatrix}
1 \\ 
i
\end{bmatrix}.
\end{equation} 

In practice, on the IBM devices studied, the standard fiducial state on all qubits is $\ket{0}$, prepared via a non-unitary process. The other states in our set can then be prepared by a sequence of single-qubit gates. For our measurement settings, we measure all two-qubit Pauli observables and take our measurement operators to be the projector onto the $+1$ eigenspace of each Pauli operator. The standard measurement on IBM devices is the $Z$ measurement, so to measure other Pauli operators we use single-qubit gates to rotate into the appropriate frame. To collect the tomographic data, we parallelise Pauli measurements, so we run 144 independent circuits per gate in the set, at $10^4$ shots per circuit. All circuits are shuffled together in a random order, to average out any drift and avoid correlations with real time within the data. After initial processing, the tomographic snapshot for the $R_{ZZ}(0.5)^5$ process was found to have an eigenvalue structure incompatible with a time-independent Markovian process so was excluded from subsequent analysis. No further error mitigation is applied to the data, and the raw tomographic matrices are taken as input for the Gate Set Flip-Flop algorithm. 
In~\Cref{app:trivial} we analytically prove that in the absence of statistical error, when gate noise is sufficiently weak, and the ideal Lindbladian has no eigenvalues close to the real negative axis, a good solution can always be found by Convex Solve in the principal branch alone. This principle was also borne out in our numerical testing of synthetic data with simulated shot noise in~\Cref{sec:Lindblad_synthetic} (see~\Cref{fig:II_trivial}). To illustrate that this principle also works in practice, during the Lindbladian fitting step, we use this simple method for all gates except CNOT and $R_{ZX}(0.5)^5 = R_{ZX}(2.5)$. The latter two gates respectively have $-1$ eigenvalues and eigenvalues close to the real negative axis, so for these we use the Alternating Projections method. 

In~\Cref{fig:gatesetfittingperth}, we show figures of merit for the full gate set (excluding $R_{ZZ}(0.5)^5$) over iterations of the flip-flop algorithm. In this case, Gate Set Flip-Flop was initialised by first running the Lindbladian fit before any gauge optimisation. The solution is seen to be well converged by 8 iterations. In~\Cref{tab:perthgatesetdistances} we compare the distances between the data $E$, as evaluated in the final gauge, and the ideal transfer matrices $\Eideal$ and the Markovian estimate $\exp(L)$ respectively. In all cases the estimated Lindbladian gives at least as good a fit as assuming the ideal, and we find the fit to be significantly improved for processes where we expect the most gate noise, namely those involving repeated application of two-qubit gates. In contrast, the single-qubit gates only show a slightly improved fit, in keeping with the belief that single-qubit gates have good fidelity on modern superconducting-qubit devices. The remaining distance $\norm{E-\exp(L)}$ between the data and the estimated process is likely at least in part due to statistical noise and residual SPAM errors that were not filtered by the algorithm.

Since the ideal process $\Eideal$ for each gate is unitary, we can use the following formula to evaluate the average gate fidelity for the estimated noisy Markovian process at each iteration, adapted from~\cite{horodecki1999general,nielsen2002simple} for transfer matrices
\begin{equation}
F_{\mathrm{avg}} (\Eideal,e^L) = \frac{\frac{1}{d}\Tr[(\Eideal)^\dagger e^L] + 1}{d + 1}\label{eq:expAvgFid}
\end{equation}
where $d=4$ is the dimension of the system.

In~\Cref{fig:RZX1,fig:RZX3,fig:RZX5} we show the estimated Lindbladian decomposition for the $R_{ZX}(0.5)^k$ process for $k\in\{1,3,5\}$, comparing initial and final estimates. Note that for an ideal gate we would have $R_{ZX}(0.5)^k = R_{ZX}(k/2)$, and in the channel picture $\exp(L)^k = \exp(k L)$ for any time-independent Lindbladian $L$. By normalising the Lindbladians with respect to the number of applications of the gate, we can test the consistency of the estimated Markovian process for each sequence. In~\Cref{tab:RZXdistances} we show the distances between the time-normalised Lindbladians for each pair of processes, as well as the distances between the exponentiated transfer matrices. We also note that the Hamiltonians in~\Cref{fig:RZX1,fig:RZX3,fig:RZX5} consistently show an overrotation about the $ZX$ axis.

\newcommand\Tstrut{\rule{0pt}{3.6ex}}       % "top" strut
\newcommand\Bstrut{\rule[-2.5ex]{0pt}{0pt}} % "bottom" strut
\newcommand{\TBstrut}{\Tstrut\Bstrut} % top&bottom struts
\begin{table}[htbp]
\centering
	\begin{tabular}{c|c|c}
         Process     & $\norm*{E - \Eideal}$ & $\norm{E-\exp(L)}$ \\
\hline
CNOT & $0.42$ & $0.33$ \\
$R_{ZX}(0.5)$ & $0.31$ & $0.25$ \\
$R_{ZX}(0.5)^3$ & $0.51$ & $0.30$ \\
$R_{ZX}(0.5)^5$ & $0.66$ & $0.36$ \\
$R_{ZZ}(0.5)$ & $0.44$ & $0.28$ \\
$R_{ZZ}(0.5)^3$ & $0.79$ & $0.34$ \\
$T\otimes \1$ & $0.25$ & $0.22$ \\
$\1\otimes T$ & $0.29$ & $0.26$ \\
$\sqrt{X}\otimes \1$ & $0.38$ & $0.31$ \\
$\1 \otimes \sqrt{X} $ & $0.34$ & $0.29$ \\
        \end{tabular}\caption{Distances between the tomographic data and corresponding ideal transfer matrices $\Eideal$ and estimated Markovian channels $\exp(L)$, respectively, from gate set analysis on \textit{ibm\_perth} qubits 3 and 5. For each process $P^*$, the tomographically estimated transfer matrix estimate $E = B (g^*)^{-1} P^* B^{-1}$ is evaluated in the final gauge $B$ obtained from the flip-flop algorithm. We find that in all cases, the Markovian channel estimate gives a better fit to the data than assuming the gate is ideal.}
\label{tab:perthgatesetdistances}
\end{table}

\begin{table}[htbp]
\centering
	\begin{tabular}{c|c|c|c}
		$j$ & $k$ & $\norm{\frac{L_j}{j} - \frac{L_k}{k}}$ & $\norm{\exp(\frac{L_j}{j}) - \exp(\frac{L_k}{k})}$ \TBstrut \\ \hline
  		$1$ & $3$ & $0.097$ & $0.094$  \\
		$1$ & $5$ & $0.12$ & $0.12$ \\
		$3$ & $5$ & $0.081$ & $0.079$
	\end{tabular}

\caption{Distances between time-normalised Lindbladians and the exponentiated transfer matrices, where $L_k$ is the estimated Lindbladian for the noisy implementation of the pulse gate sequence $R_{ZX}(0.5)^k$ on \textit{ibm\_perth}. }\label{tab:RZXdistances}
\end{table}

\begin{figure}[htbp]
\centering
\begin{subfigure}{0.92\linewidth}
    \includegraphics[width=0.6\textwidth,trim={0.2cm 0 0 0},clip]{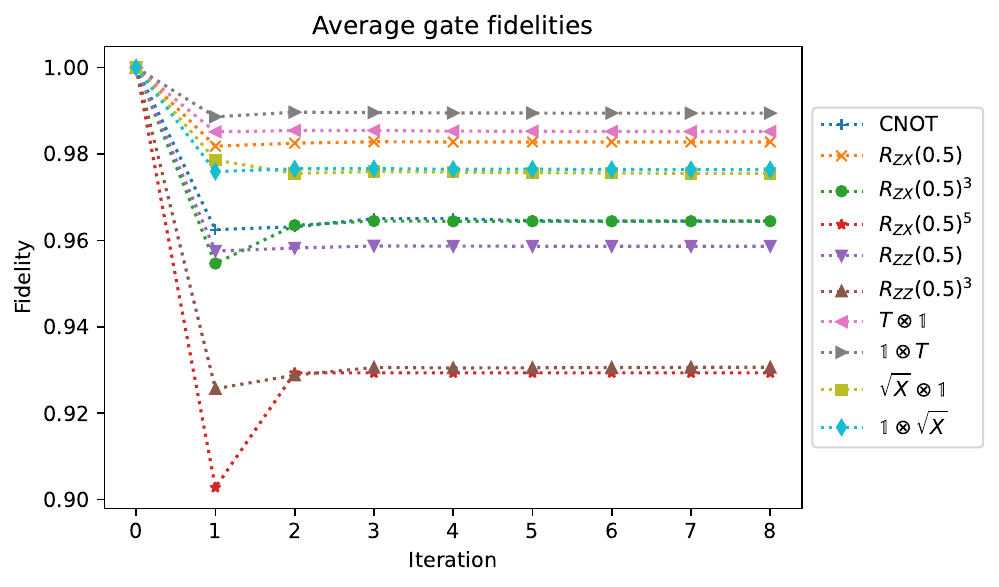}
\end{subfigure}
\begin{subfigure}{0.48\linewidth}
    \includegraphics[width=\linewidth,trim={0 0 5cm 0},clip]{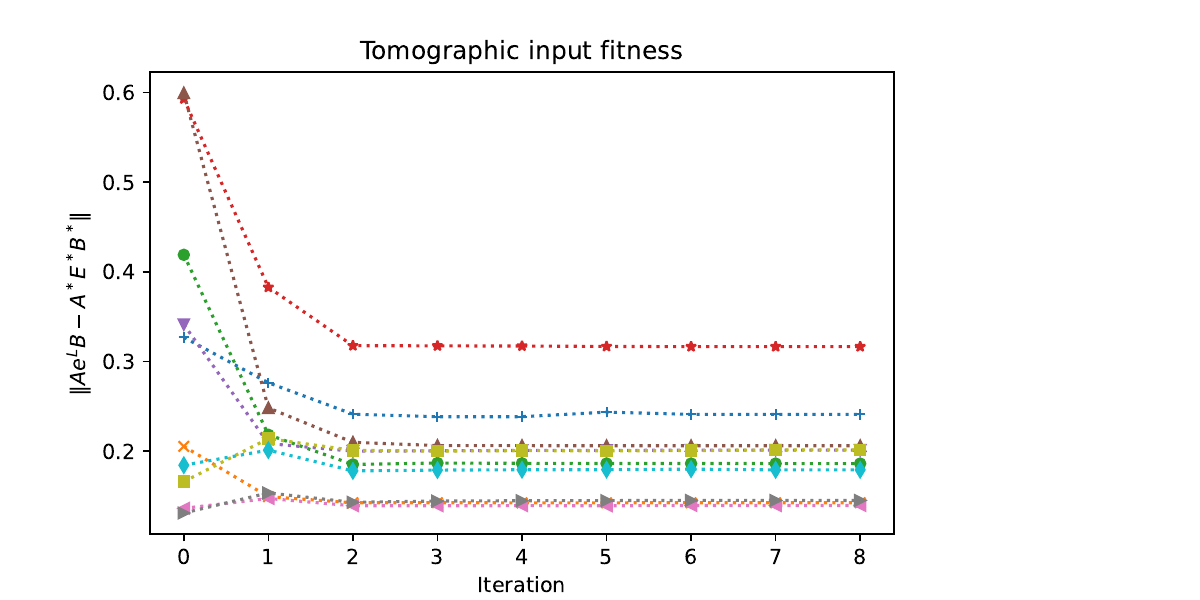}
\end{subfigure}
\begin{subfigure}{0.48\linewidth}
    \includegraphics[width=\linewidth]{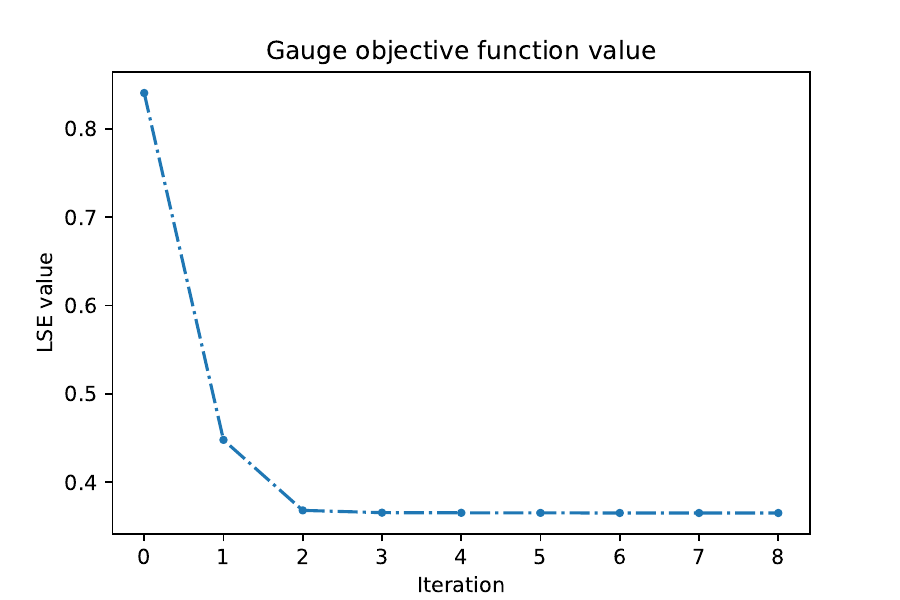}
\end{subfigure}
\caption{Gate Set Flip-Flop figures of merit for data collected from \textit{ibm\_perth}. (Top panel) The average gate fidelity between the ideal gate and the estimated Markovian process at each iteration, computed using~\Cref{eq:expAvgFid}. (Bottom left panel) Distance between raw tomographic data and the equivalent data for the estimated Markovian process assuming the current gauge. (Bottom right panel) Gauge objective function as defined in ~\Cref{eq:h_def}.}\label{fig:gatesetfittingperth}
\end{figure}

\begin{figure}
\centering
\includegraphics[width=\linewidth,trim={0 0 0 0},clip]{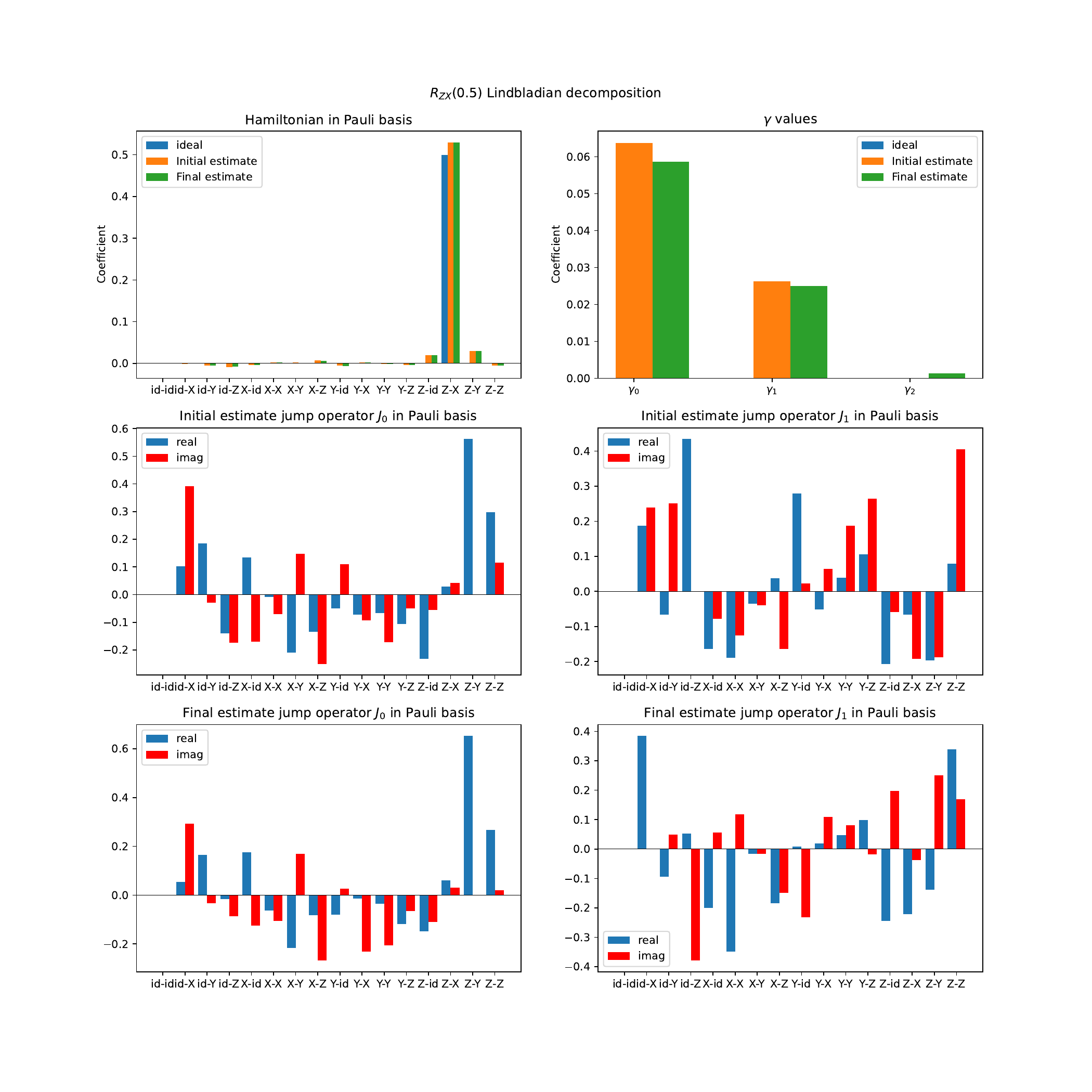}
\caption{Canonical decomposition (see~\Cref{subsec:Lindblad}) of the fitted Lindbladian for the $R_{ZX}(0.5)$ process applied to qubit pair $(3,5)$ on \textit{ibm\_perth}. We compare the ideal unitary process with the initial estimate prior to any gauge optimisation, and the final estimate after convergence of the Gate Set Flip-Flop algorithm. Pauli decompositions are given only for the two dominant jump operators.}\label{fig:RZX1}
\end{figure}
\begin{figure}
\centering
\includegraphics[width=\linewidth,trim={0 0 0 0},clip]{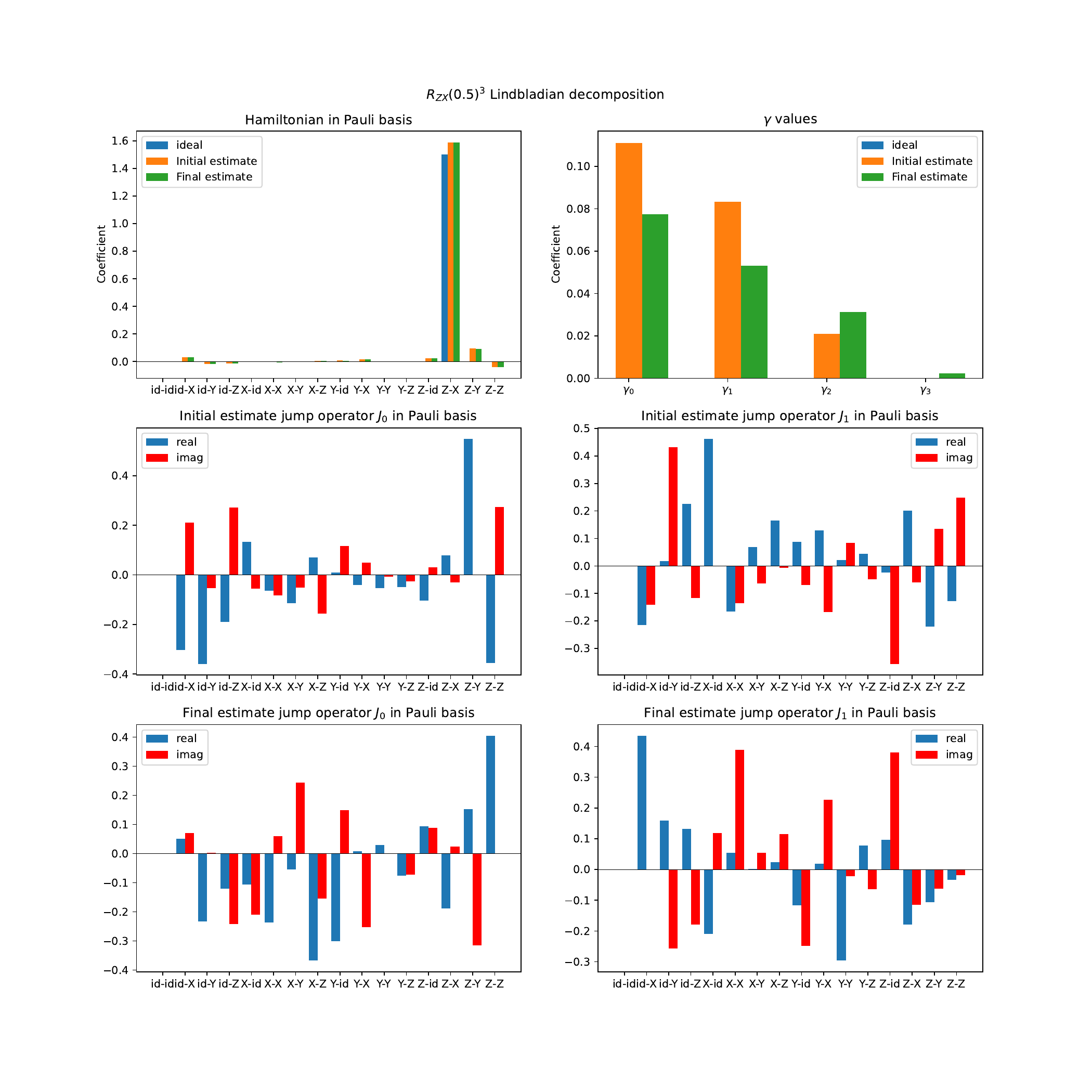}
\caption{Canonical decomposition (see~\Cref{subsec:Lindblad}) of the fitted Lindbladian for the $R_{ZX}(0.5)^3 = R_{ZX}(1.5)$ process applied to qubit pair $(3,5)$ on \textit{ibm\_perth}. The ideal process is compared with the initial estimate before gauge optimisation, and the final fitted Lindbladian after convergence of Gate Set Flip-Flop. Pauli decompositions are given only for the two dominant jump operators.}\label{fig:RZX3}
\end{figure}
\begin{figure}
\centering
\includegraphics[width=\linewidth,trim={0 0 0 0},clip]{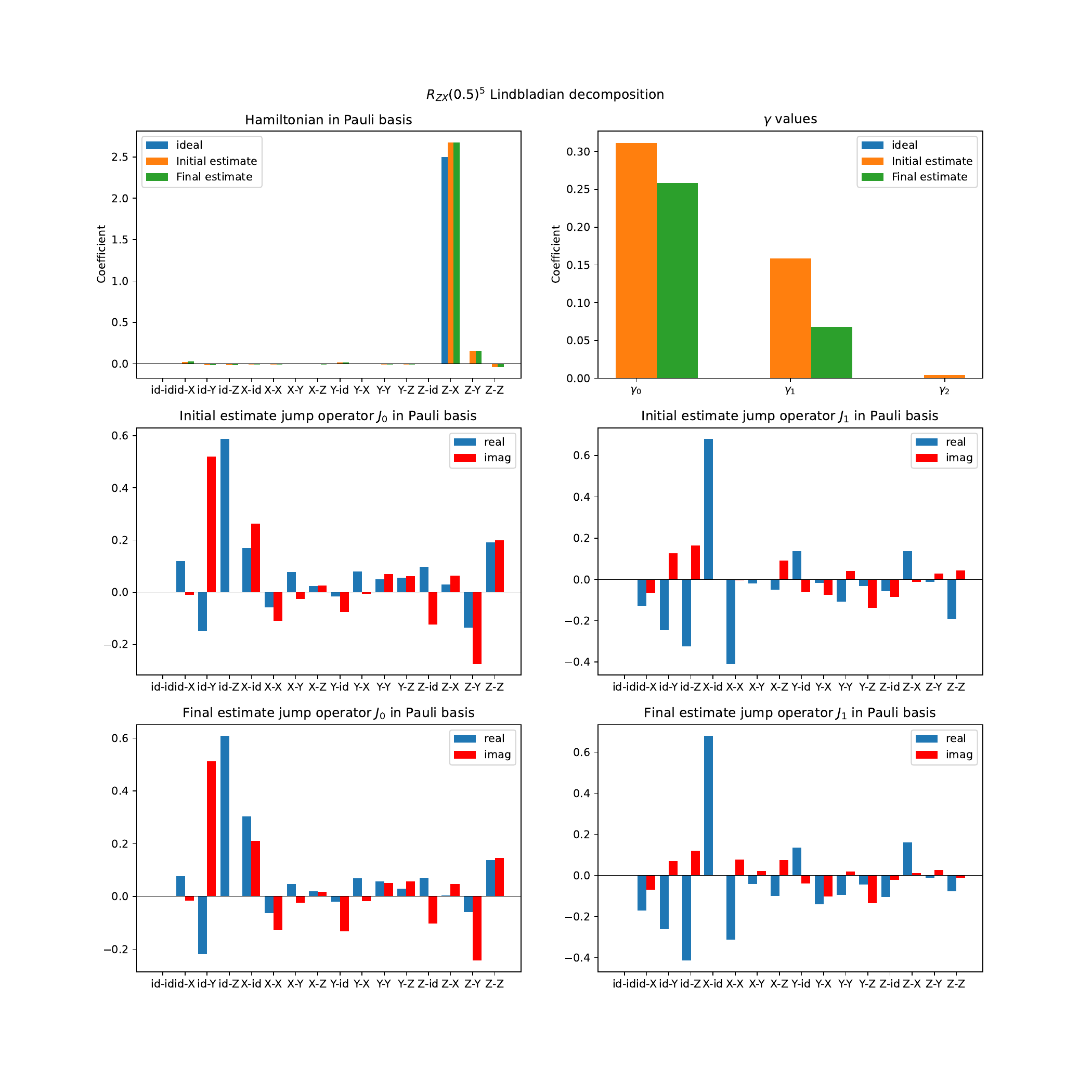}
\caption{Canonical decomposition (see~\Cref{subsec:Lindblad}) of the fitted Lindbladian for the $R_{ZX}(0.5)^5 = R_{ZX}(2.5)$ process applied to qubit pair $(3,5)$ on \textit{ibm\_perth}. We compare the ideal unitary process with fitted estimates before and after running the Gate Set Flip-Flop routine until converged. Pauli decompositions are given only for the two dominant jump operators.}\label{fig:RZX5}
\end{figure}
\FloatBarrier

\newpage

\subsection{Parallel Gate Cross-talk Experiment on \textit{ibm\_cairo}}
\newcommand{\CXsingle}{\text{CNOT}_{\mathrm{single}}}
\newcommand{\CXpar}{\text{CNOT}_{\mathrm{parallel}}}
\begin{figure}[hbtp]
\centering
\includegraphics[width=\linewidth,trim={2cm 5cm 4cm 4cm},clip]{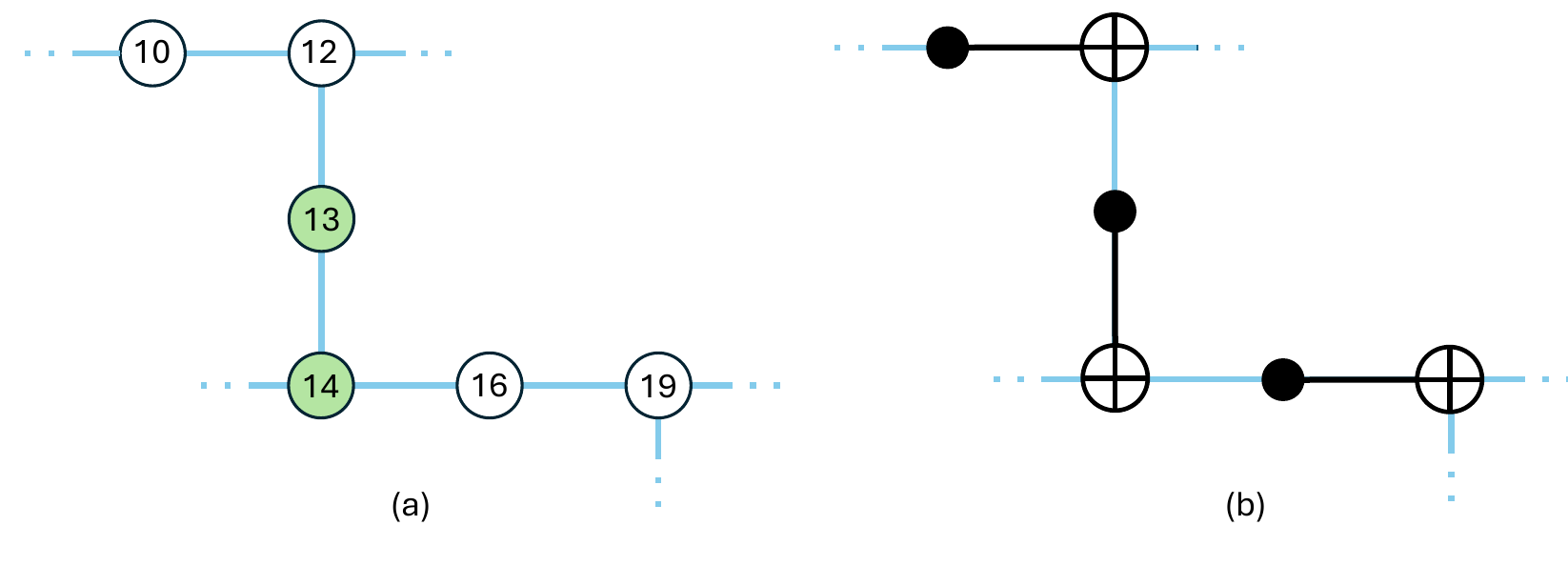}
\caption{(a) Partial qubit layout on \textit{ibm\_cairo}. Numbered vertices show qubit locations, edges show connections where CNOT gates are available.  Dashed lines show connections to qubits not used in the experiment. The qubit pair targeted for gate set tomography is shown shaded green. (b) Placement of CNOT gates during tomography of the $\CXpar$ process.}\label{fig:cairolayout}
\end{figure}
For this second set of data, we focus on the analysis of a particular gate, and in particular consider its susceptibility to cross-talk. We perform gate set tomography of the set
\begin{equation}
\mathcal{G} = \{ \CXsingle, \CXpar, T \otimes \1, \1 \otimes T, \sqrt{X} \otimes \1, \1 \otimes \sqrt{X} \}
\end{equation}
where $\CXsingle$ and $\CXpar$ are CNOT gates executed either in isolation or in parallel with CNOT gates executed on other qubit pairs. The data was collected from IBM's 27-qubit device \textit{ibm\_cairo}. Topology of the six-qubit patch used for the tomography is shown in~\Cref{fig:cairolayout}, along with the placement of two-qubit gates for the $\CXpar$ process. As in the previous demonstration, the single-qubit gates are provided to aid gauge optimisation, shots were taken at a rate of $10^4$ per setting, circuits were executed in random order, and readout mitigation was not applied to the data before running the algorithm. For this gate set, Convex Solve on the principal branch was again used to fit Lindbladians for all processes except the two CNOT processes, for which Alternating Projections was used. We illustrate the convergence of the solution in~\Cref{fig:gatesetfittingcairo}. Again, we observe good convergence within around 10 iterations. Comparing with~\Cref{fig:gatesetfittingperth} we see that we are able to obtain a closer fit to a set of time-independent Lindbladians, and the estimated Markovian processes generally have better fidelity. Final distances from the data are shown in~\Cref{tab:cairodistances}, where again we see only minimal improvement in fit for single-qubit gates, which are expected to already be close to ideal, with more significant improvement for the noisy two-qubit processes. We note that the difference between the tomographic input fitness for the converged solutions for single and parallel CNOTs is relatively small, suggesting that in this scenario cross-talk does not induce a dramatic increase in apparent non-Markovianity. On the other hand, we see that the algorithm suggests a reduced average gate fidelity for the parallel CNOT, and we observe qualitative differences in the estimated Lindbladian decompositions (\Cref{fig:singleparalleldecomps}). In particular, we see increased overrotation about the $Z \otimes \1$ axis for the parallel case, and an increase in the strength of dissipative noise.

For completeness, in~\Cref{fig:singleinitfinalCairo,fig:parallelinitfinalCairo} we compare the initial and final estimates for each implementation of the CNOT. We notice that in both cases, the Hamiltonian part of the Lindbladian remains relatively stable, but the final choice of gauge output by the flip-flop algorithm yields an apparent reduction in the strength of the dissipative part.

\begin{table}[htbp]
\centering
	\begin{tabular}{c|c|c}
         Process     & $\norm*{E - \Eideal}$ & $\norm{E-\exp(L)}$ \\
\hline
$\CXsingle$ & $0.23$ & $0.18$ \\
$T\otimes \1$ & $0.19$ & $0.18$ \\
$\1\otimes T$ & $0.20$ & $0.17$ \\
$\CXpar$ & $0.37$ & $0.17$ \\
$\sqrt{X}\otimes \1$ & $0.20$ & $0.17$ \\
$\1 \otimes \sqrt{X} $ & $0.19$ & $0.17$ \\
        \end{tabular}
\caption{Distances between the tomographic channel estimates and corresponding ideal transfer matrices $\Eideal$ and estimated Markovian channels $\exp(L)$, respectively, from gate cross-talk analysis on \textit{ibm\_cairo}, qubits. For each process $P^*$, the tomographically estimated transfer matrix estimate $E = B (g^*)^{-1} P^* B^{-1}$ is evaluated in the final gauge $B$ obtained from the flip-flop algorithm.}
\label{tab:cairodistances}
\end{table}

\begin{figure}[htbp]
\centering
\begin{subfigure}{0.92\linewidth}
    \includegraphics[width=0.62\textwidth,trim={0.2cm 0 0 0},clip]{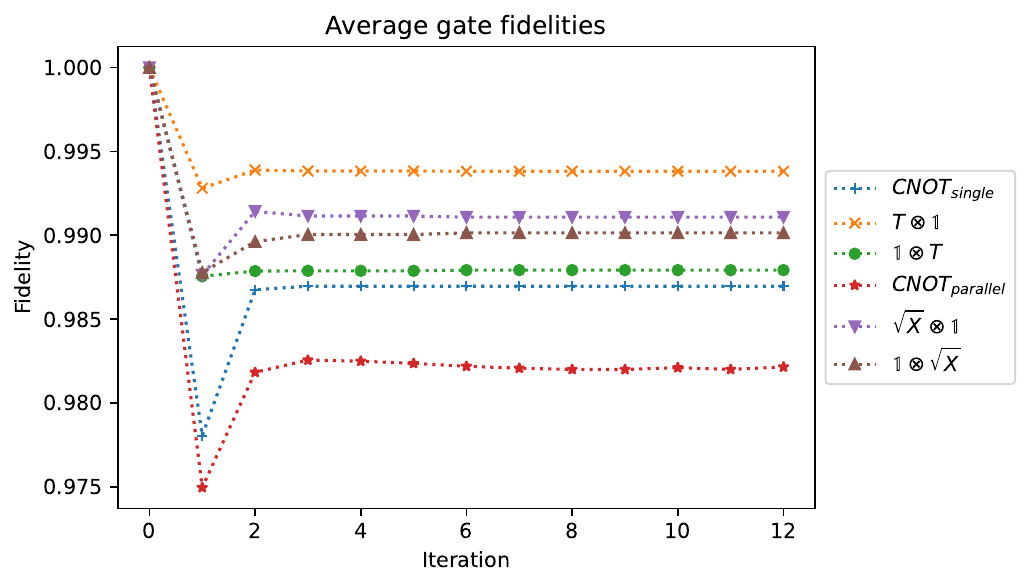}
\end{subfigure}
\begin{subfigure}{0.48\linewidth}
    \includegraphics[width=\linewidth,trim={0 0 5cm 0},clip]{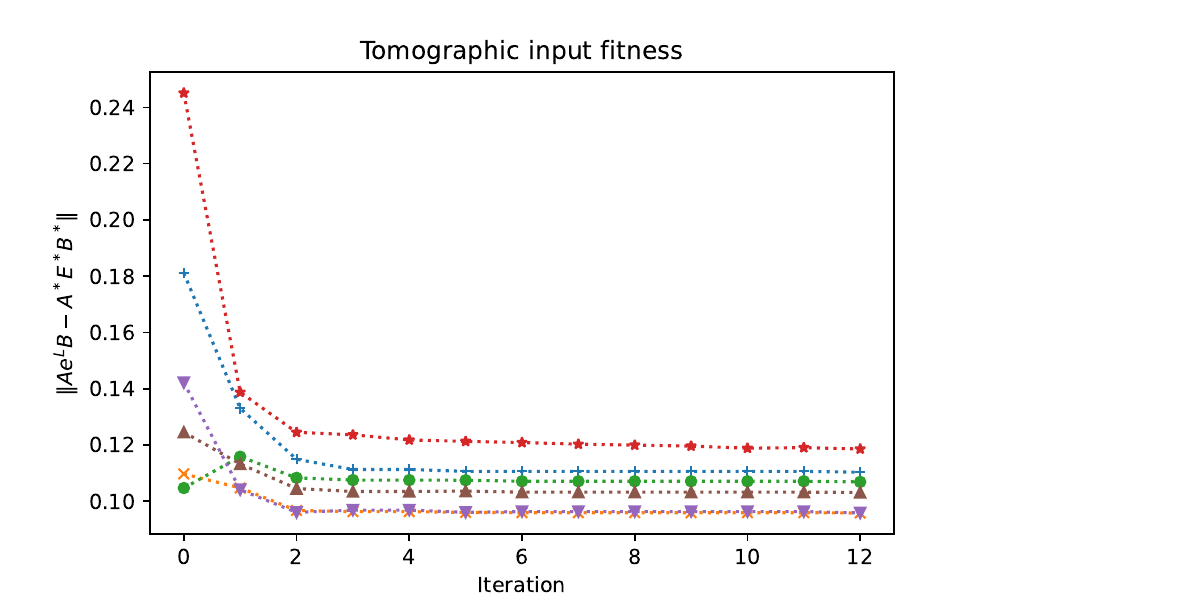}
\end{subfigure}
\begin{subfigure}{0.48\linewidth}
    \includegraphics[width=\linewidth]{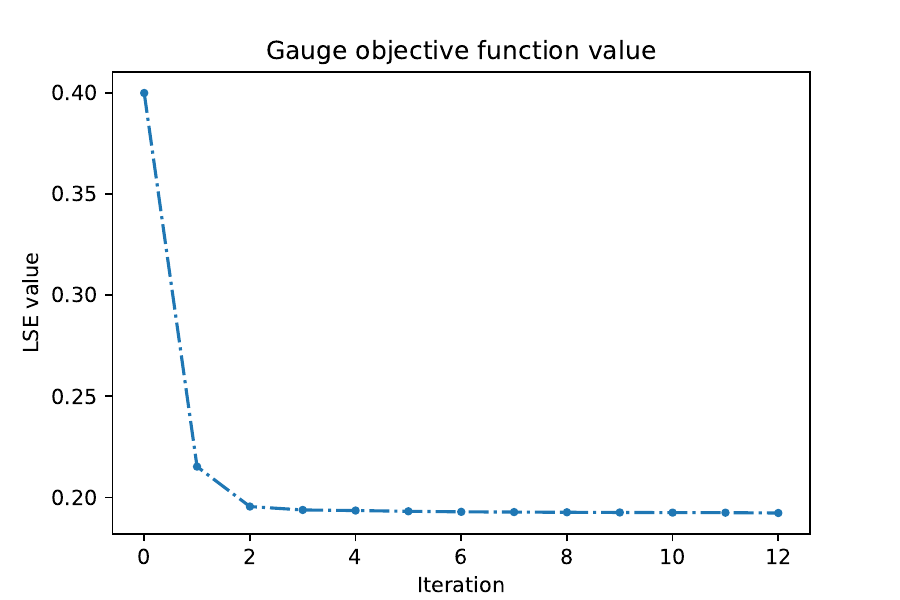}
\end{subfigure}
\caption{Gate Set Flip-Flop figures of merit for data collected from \textit{ibm\_cairo}, tomographing qubits 13 and 14. (Top panel) The average gate fidelity between the ideal gate and the estimated Markovian process at each iteration, computed using~\Cref{eq:expAvgFid}. (Bottom left panel) Distance between raw tomographic data and the equivalent data for the estimated Markovian process assuming the current gauge. (Bottom right panel) Gauge objective function as defined in ~\Cref{eq:h_def}.}\label{fig:gatesetfittingcairo}
\end{figure}
\begin{figure}[htbp]
\centering
\includegraphics[width=\linewidth,trim={0 0 0 0},clip]{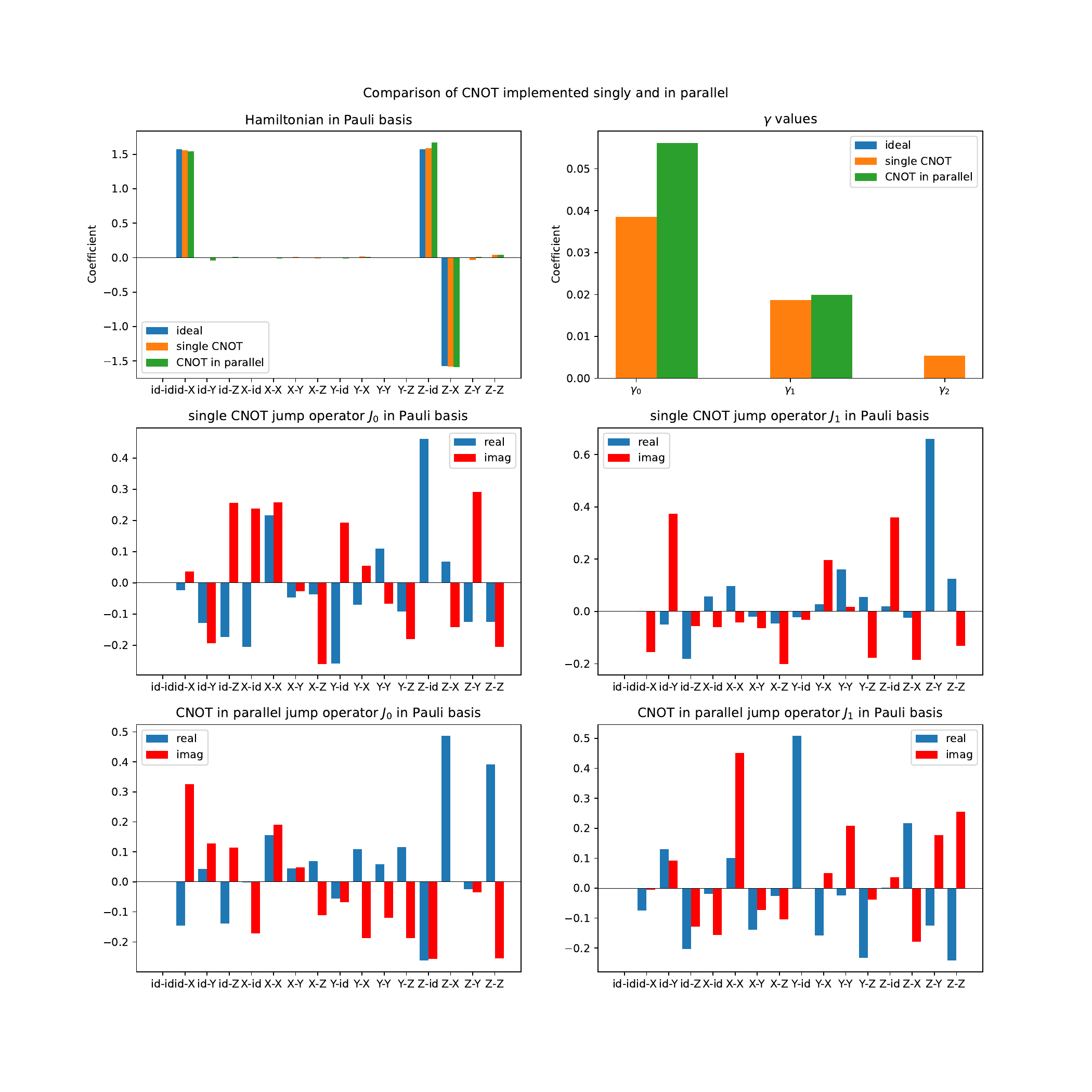}
\caption{Here we compare the canonical decompositions (see~\Cref{subsec:Lindblad}) of the the ideal unitary CNOT gate with the fitted Lindbladians after Gate Set Flip-Flop for the noisy CNOT gate implemented on qubit pair $(13,14)$ on \textit{ibm\_cairo}. Note that the estimated decomposition for the noisy CNOT differs depending on whether the CNOT was implemented in isolation while the rest of the device was left idle (orange) or in parallel with CNOT gates implemented elsewhere.}\label{fig:singleparalleldecomps}
\end{figure}
\begin{figure}[htbp]
\centering
\includegraphics[width=\linewidth,trim={0 0 0 0},clip]{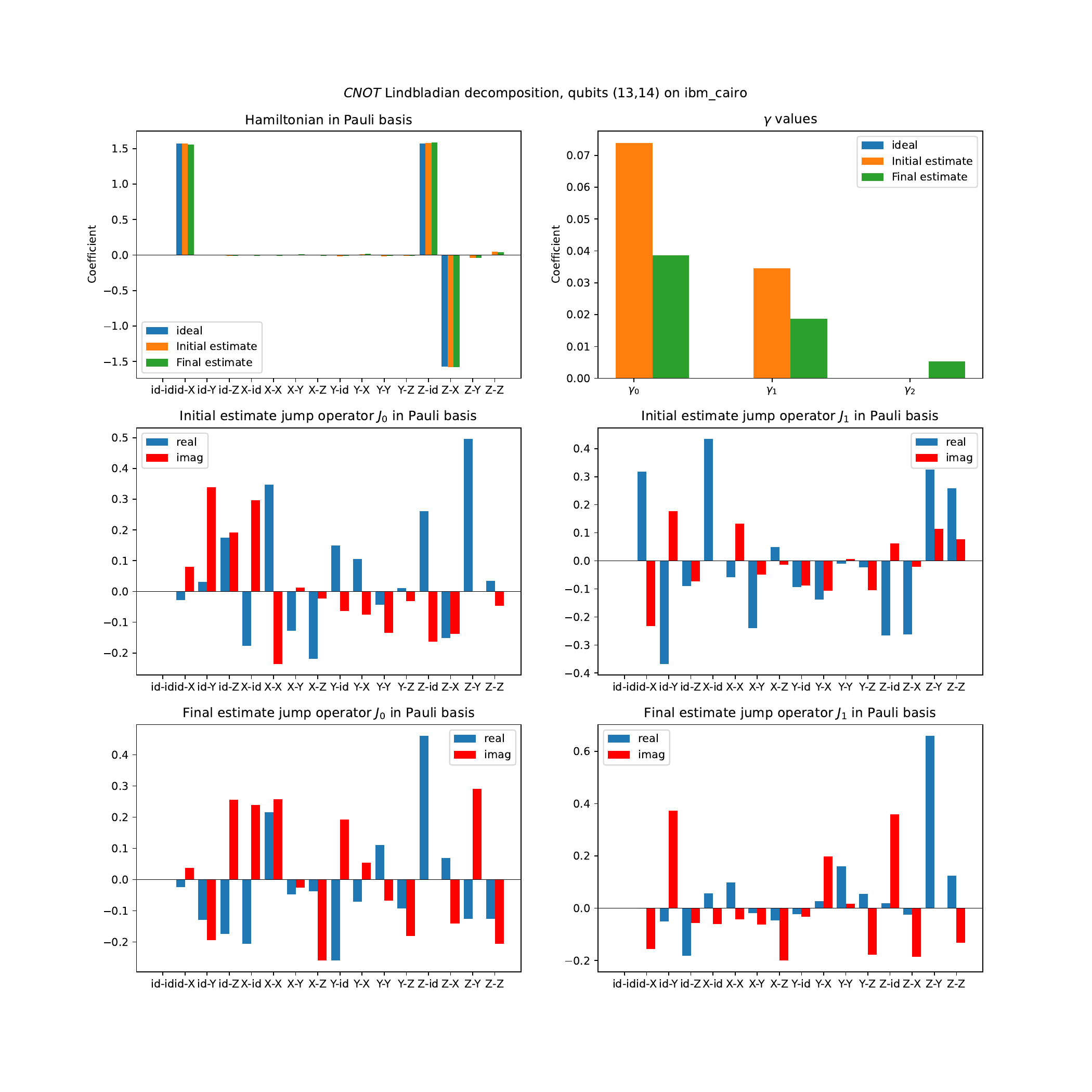}
\caption{The estimated canonical decomposition (\Cref{subsec:Lindblad}) of the estimated Lindbladian generators for the CNOT implemented on on qubit pair $(13,14)$ on \textit{ibm\_cairo}. Here we compare the ideal unitary process with experimental estimates before and after Gate Set Flip-Flop, for the case where the gate was implemented on $(13,14)$ alone while the rest of the device was left idle.}\label{fig:singleinitfinalCairo}
\end{figure}
\begin{figure}[htbp]
\centering
\includegraphics[width=\linewidth,trim={0 0 0 0},clip]{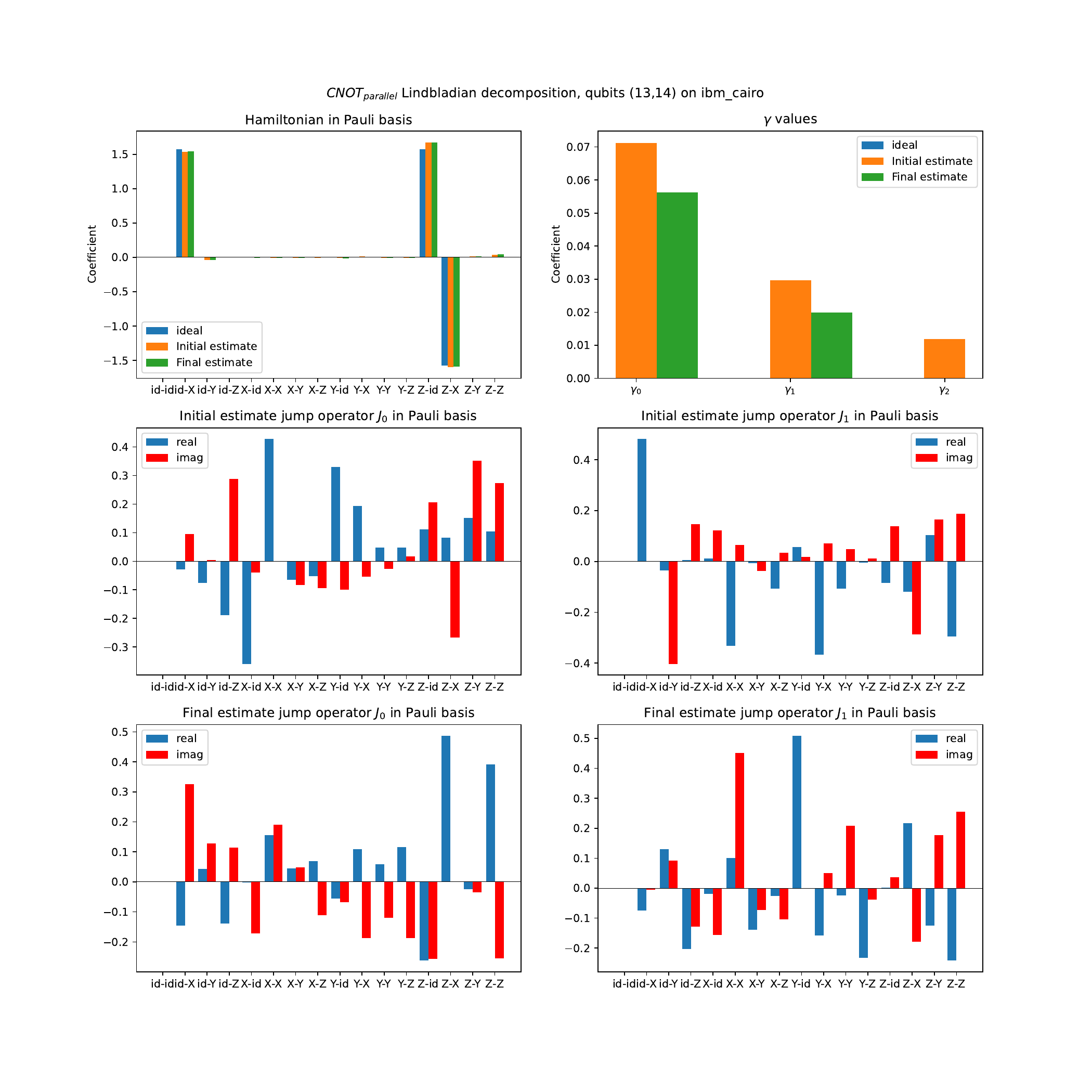}
\caption{Estimated canonical decomposition (\Cref{subsec:Lindblad}) of the estimated Lindbladian generators for the CNOT implemented on on qubit pair $(13,14)$ on \textit{ibm\_cairo}. Here we compare the ideal unitary process with experimental estimates before and after Gate Set Flip-Flop, for the case where the gate was implemented in parallel with CNOT gates being applied simultaneously to pairs $(10,12)$ and $(16,19)$.}\label{fig:parallelinitfinalCairo}
\end{figure}
\FloatBarrier

\newpage
\subsection{Idling Cross-talk Experiment on \textit{ibm\_perth}}
Here we investigate the noisy idling process on pairs of qubits. We take as the gate set,
\newcommand{\delay}[1]{\mathcal{D}_{#1}}
\begin{equation}
\mathcal{G} = \{\text{CNOT},  T \otimes \1, \1 \otimes T, \sqrt{X} \otimes \1, \1 \otimes \sqrt{X},
                         \delay{}, \delay{}^5,\delay{X},\delay{X}^5,\delay{\text{CNOT}},\delay{\text{CNOT}}^5 \}
\end{equation}
where $\delay{}$ is the two-qubit channel obtained by sending a delay instruction to the quantum device. We used a delay length of $380\,\mathrm{ns}$ to approximately match the time taken to implement a CNOT gate on the device. $\delay{X}$ are $\delay{\text{CNOT}}$ are obtained by sending the same length of delay to the target qubit pair, while either $X$ or CNOT gates, respectively, are applied to nearby qubits. In this section, we give a demonstration of our algorithm for two different pairs on the \textit{ibm\_perth} device, namely the pairs labelled $(0,1)$  and $(0,5)$. In~\Cref{fig:idlingPerthLayout} we show the device layout and highlight which qubits are targeted for tomography, and where the CNOT is applied in the case of the $\delay{\text{CNOT}}$ process. For the $\delay{X}$ processes, $X$ gates are applied to all qubits not targeted for tomography. For example, when characterising the $\delay{X}$ idling process on qubit pair $(0,1)$, the $X$ gate is applied to qubits $2$, $3$, $4$, $5$ and $6$. For the $(0,1)$ pair, the CNOT used to aid gauge optimisation is available on the device as an elementary gate. However for the gate set targeting $(0,5)$ , a direct connection is not available. We instead implement the CNOT for this pair using the circuit decomposition shown in~\Cref{fig:cnotdecomp}. Once again, we use the Alternating Projections method for fitting the CNOT, and Convex Solve on the principal branch for all others. We show the convergence for pairs $(0,1)$ and $(0,5)$ in~\Cref{fig:idling01merit,fig:idling05merit}. We remark that the approach of using Convex Solve on a single branch is still successful in achieving convergence to a good set of solutions, despite the fact that several of the idling processes turn out to be very far from the identity channel as evidenced in~\Cref{tab:perthidlingdistances}.  The distances of the $1900\,\mathrm{ns}$ idling processes from the ideal were observed to be much larger than any seen for the two-qubit gates analysed, whereas distances to the estimated Markovian channel were comparable to other cases. The Lindbladian decompositions are plotted in~\Cref{fig:01idlingLindblads,fig:05idlingLindblads}. We find that the noise differs significantly in character depending on the qubit choice and the gates applied in parallel. 

Remarkably, for both pairs, the noise is reduced dramatically when $X$ gates are applied to the other qubits, compared to the case where all qubits are left to idle. One possible explanation is that repeatedly applying $X$ pulses to neighbouring qubits might act as a rudimentary dynamical decoupling sequence \cite{dynamicaldecoupling2014}. For the $(0,1)$ pair, repeatedly applying a CNOT on qubits $3$ and $5$ appears to induce the build up of a coherent $Z$ error on qubit $1$, while both the $\delay{}$ and $\delay{\text{CNOT}}$ processes exhibit a $ZZ$ coupling between the target qubits. In contrast, the more separated $(0,5)$ pair shows no significant $ZZ$ coupling. Meanwhile application of CNOTs on the pair $(1,3)$, situated between the target qubits $(0,5)$, appears to reduce the effective coherent $Z$ error on qubit $5$ while increasing it on qubit $0$.
\begin{figure}[hbtp]
\centering
\includegraphics[width=\linewidth,trim={3cm 6cm 3cm 4cm},clip]{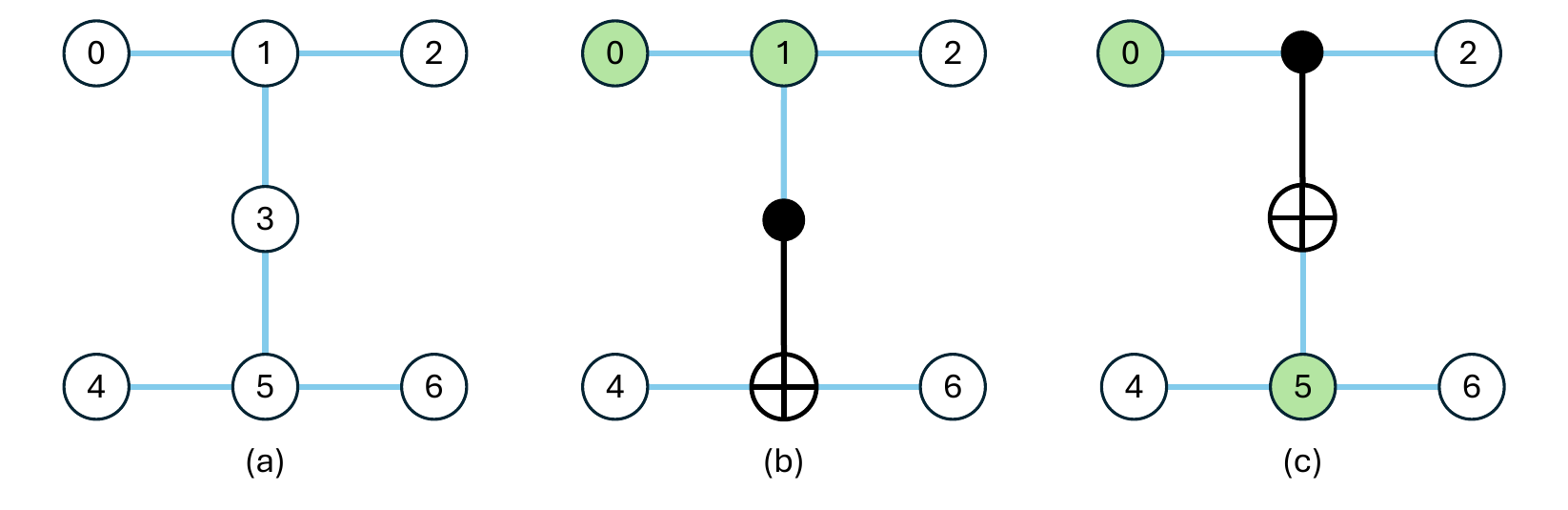}
\caption{(a) (a) Qubit layout on \textit{ibm\_perth}. Numbered vertices show qubit locations, edges show connections where CNOT gates are available.  (b) Shaded green qubits show the position of the pair targeted for process tomography in the $(0,1)$ run. For the $\delay{\text{CNOT}}$ process, the CNOT is applied on pair $(3,5)$ as shown. (c) During tomography of the $\delay{\text{CNOT}}$ process on qubits $(0,5)$, the CNOT is applied between qubits $1$ and $3$.}\label{fig:idlingPerthLayout}
\end{figure}
\begin{figure}[htbp]
\centering
\includegraphics[width=0.9\linewidth,trim={7cm 7cm 7cm 7cm},clip]{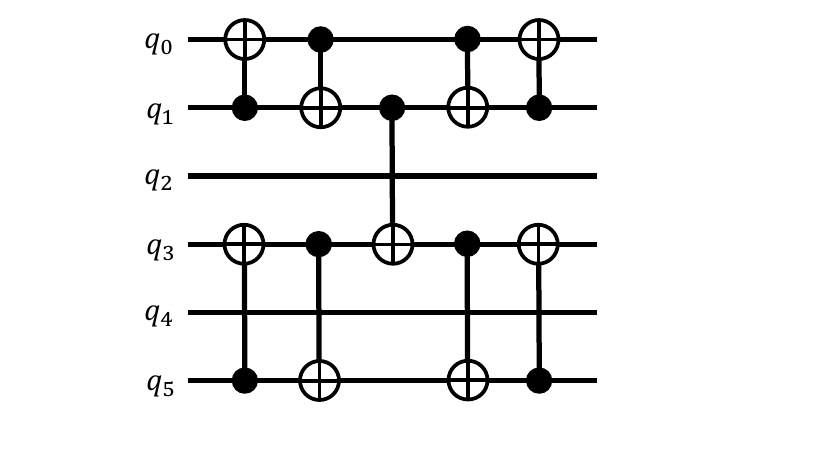}
\caption{Circuit decomposition for CNOT between qubits $0$ and $5$ on \textit{ibm\_perth}.}\label{fig:cnotdecomp}
\end{figure}
\begin{table}[htbp]
\centering
\subfloat[Qubit pair (0,1)]{
	\begin{tabular}{c|c|c}
         Process     & $\norm*{E - \Eideal}$ & $\norm{E-\exp(L)}$ \\
\hline
$\text{CNOT}$ & $0.28$ & $0.23$ \\
$T\otimes \1$ & $0.25$ & $0.22$ \\
$\1\otimes T$ & $0.25$ & $0.23$ \\
$\mathcal{D}$ & $0.57$ & $0.22$ \\
$\mathcal{D}^5$ & $2.08$ & $0.18$ \\
$\mathcal{D}_{\text{CNOT}}$ & $1.08$ & $0.21$ \\
$\mathcal{D}_{\text{CNOT}}^5$ & $4.07$ & $0.22$ \\
$\mathcal{D}_{X}$ & $0.33$ & $0.24$ \\
$\mathcal{D}_{X}^5$ & $0.75$ & $0.21$ \\
$\sqrt{X}\otimes \1$ & $0.27$ & $0.22$ 
\end{tabular}
}
\subfloat[Qubit pair (0,5)]{
	\begin{tabular}{c|c|c}
         Process     & $\norm*{E - \Eideal}$ & $\norm{E-\exp(L)}$ \\
\hline
$\text{CNOT}$ & $0.76$ & $0.24$ \\
$T\otimes \1$ & $0.30$ & $0.23$ \\
$\1\otimes T$ & $0.28$ & $0.25$ \\
$\mathcal{D}$ & $0.77$ & $0.21$ \\
$\mathcal{D}^5$ & $3.23$ & $0.23$ \\
$\mathcal{D}_{\text{CNOT}}$ & $0.66$ & $0.26$ \\
$\mathcal{D}_{\text{CNOT}}^5$ & $2.63$ & $0.18$ \\
$\mathcal{D}_{X}$ & $0.35$ & $0.24$ \\
$\mathcal{D}_{X}^5$ & $0.55$ & $0.20$ \\
$\sqrt{X}\otimes \1$ & $0.31$ & $0.24$ 
        \end{tabular}
}
\caption{Distances between the tomographic channel estimates and corresponding ideal transfer matrices $\Eideal$ and estimated Markovian channels $\exp(L)$, respectively, from idling analysis on \textit{ibm\_perth}, on qubit pairs (a) (0,1), and (b) (0,5). For each process $P^*$, the tomographically estimated transfer matrix estimate $E = B (g^*)^{-1} P^* B^{-1}$ is evaluated in the final gauge $B$ obtained from the flip-flop algorithm.}
\label{tab:perthidlingdistances}
\end{table}

\begin{figure}[htbp]
\centering
\begin{subfigure}{0.92\linewidth}
    \includegraphics[width=0.59\textwidth,trim={-0.2cm 0 0 0},clip]{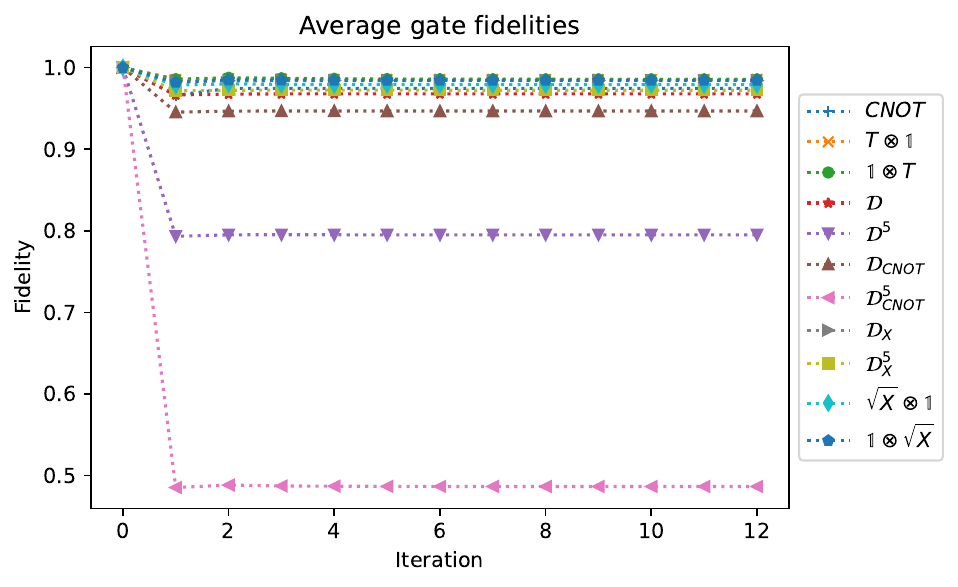}
\end{subfigure}
\begin{subfigure}{0.48\linewidth}
    \includegraphics[width=\linewidth,trim={0 0 5cm 0},clip]{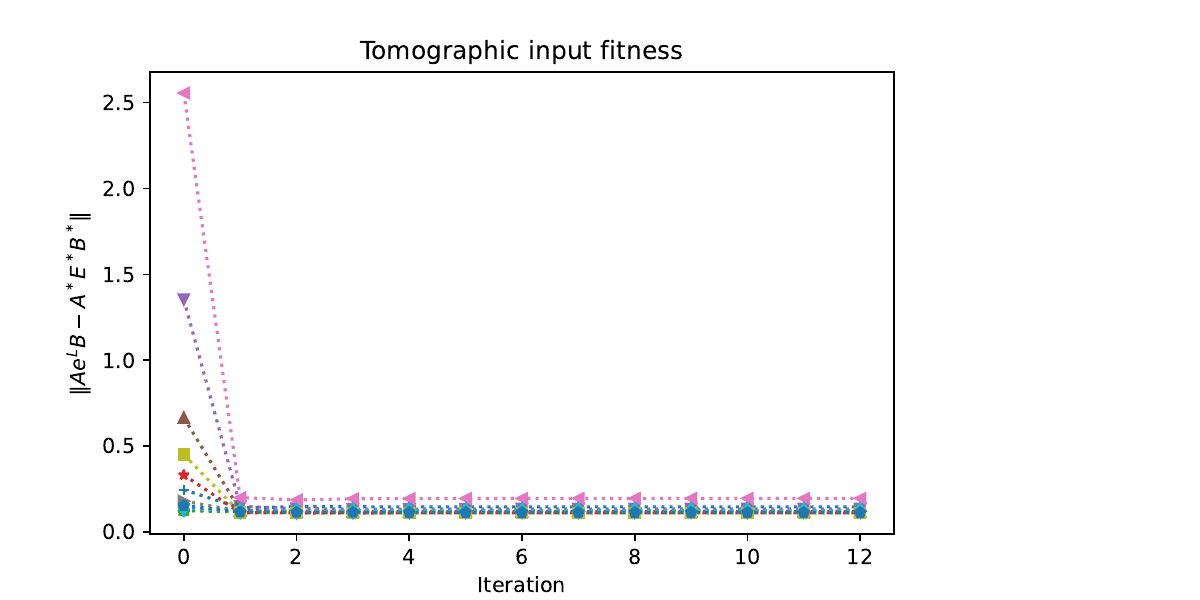}
\end{subfigure}
\begin{subfigure}{0.48\linewidth}
    \includegraphics[width=\linewidth]{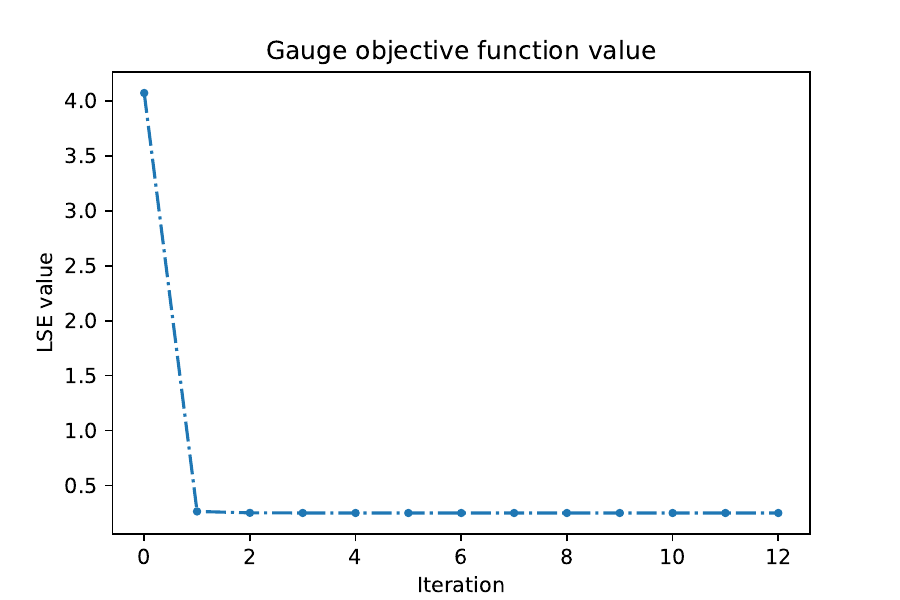}
\end{subfigure}
\caption{Gate Set Flip-Flop figures of merit for the idling experiment on qubits $(0,1)$ from \textit{ibm\_perth}. (Top panel) The average gate fidelity between the ideal gate and the estimated Markovian process at each iteration, computed using~\Cref{eq:expAvgFid}. (Bottom left panel) Distance between raw tomographic data and the equivalent data for the estimated Markovian process assuming the current gauge. (Bottom right panel) Gauge objective function as defined in ~\Cref{eq:h_def}.}\label{fig:idling01merit}
\end{figure}
\begin{figure}[htbp]
\centering
\begin{subfigure}{0.92\linewidth}
    \includegraphics[width=0.6\textwidth,trim={0 0 0 0},clip]{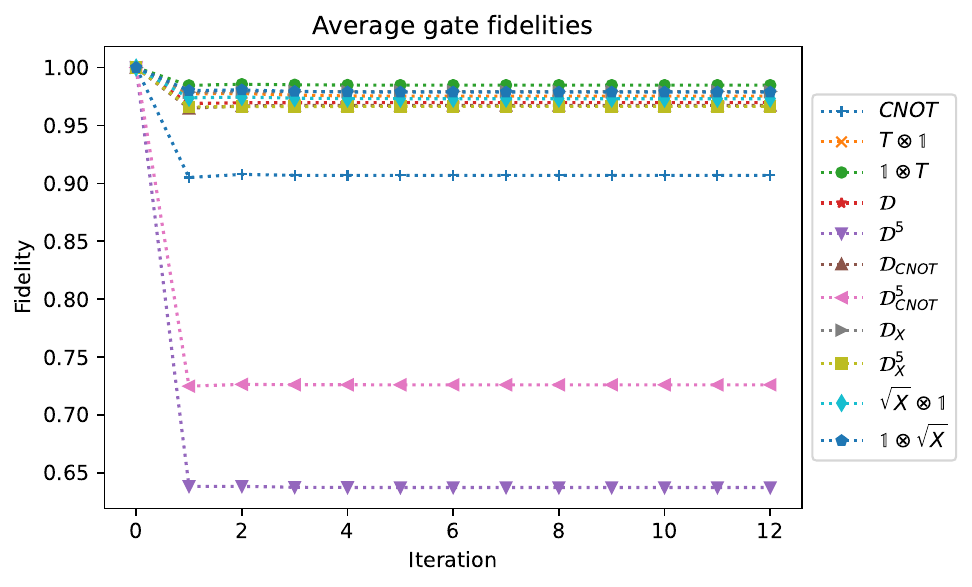}
\end{subfigure}
\begin{subfigure}{0.48\linewidth}
    \includegraphics[width=\linewidth,trim={0 0 5cm 0},clip]{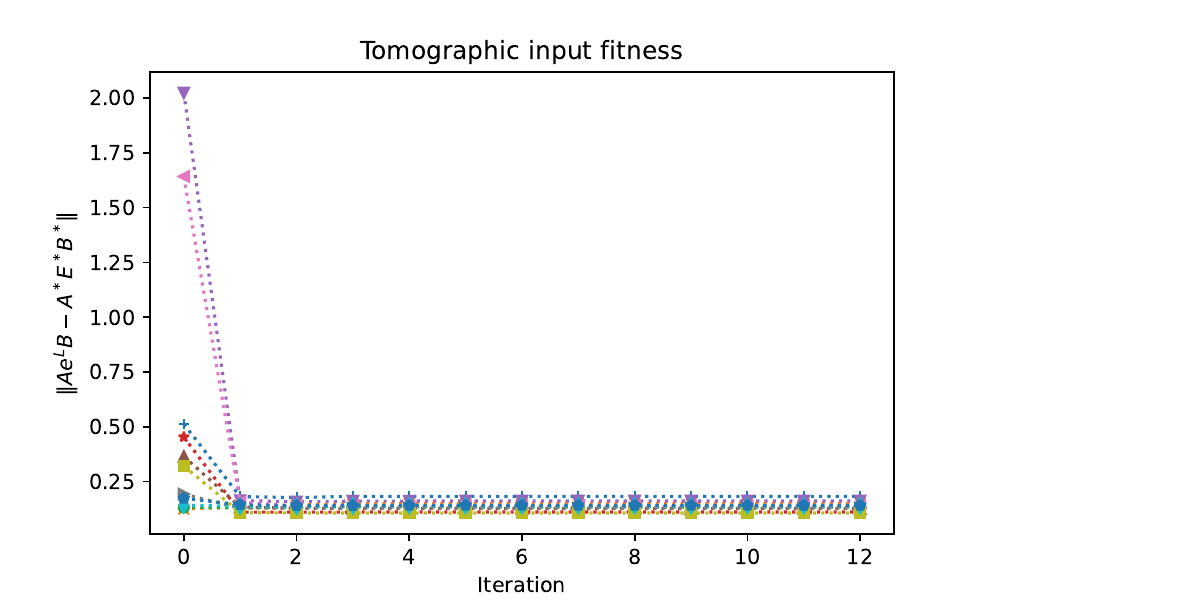}
\end{subfigure}
\begin{subfigure}{0.48\linewidth}
    \includegraphics[width=\linewidth]{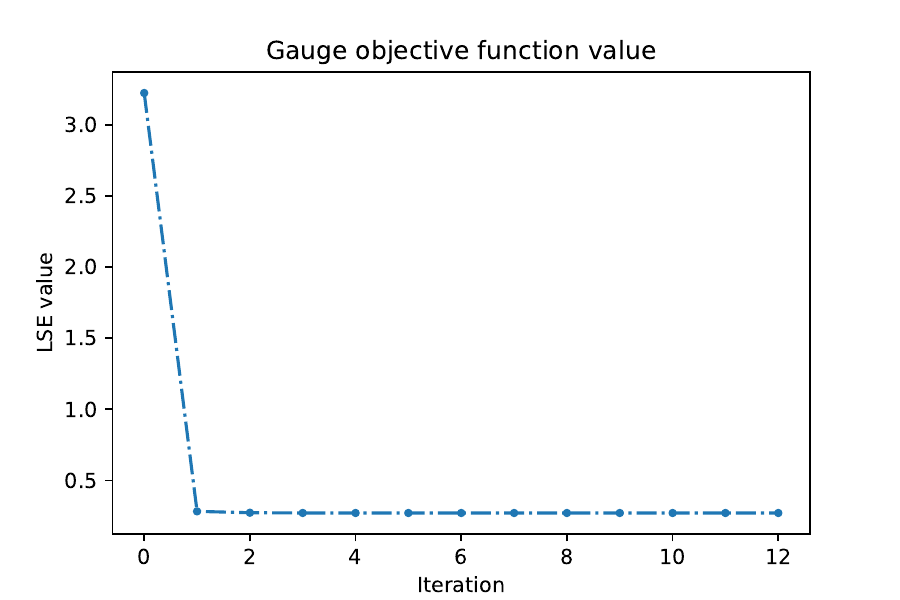}
\end{subfigure}
\caption{Gate Set Flip-Flop figures of merit for the idling experiment on qubits $(0,5)$ from \textit{ibm\_perth}. (Top panel) The average gate fidelity between the ideal gate and the estimated Markovian process at each iteration, computed using~\Cref{eq:expAvgFid}. (Bottom left panel) Distance between raw tomographic data and the equivalent data for the estimated Markovian process assuming the current gauge. (Bottom right panel) Gauge objective function as defined in ~\Cref{eq:h_def}.}\label{fig:idling05merit}
\end{figure}
\begin{figure}[htbp]
\centering
\includegraphics[width=\linewidth,trim={0 0 0 0},clip]{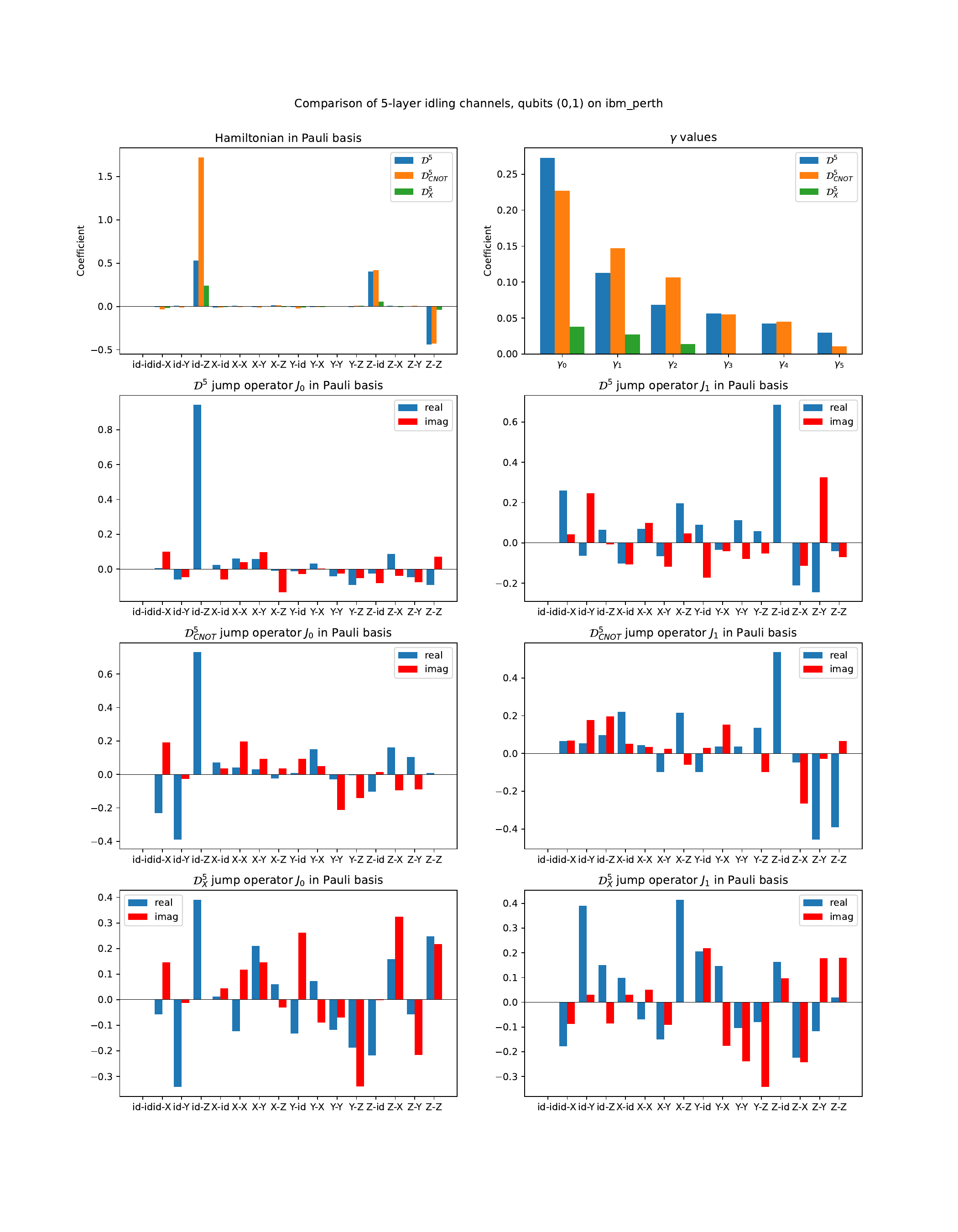}
\caption{Here we plot the canonical decompositions (see~\Cref{subsec:Lindblad}) of the estimated Lindbladian for three variants of the idling process on the adjacent qubit pair $(0,1)$ (see~\Cref{fig:idlingPerthLayout}) on \textit{ibm\_perth}. The ideal process (i.e. the generator of the identity channel) is not shown, since the Hamiltonian and jump operators would all be zero.The estimates are output from the Gate Set Flip-Flop algorithm once converged. (Blue) The whole device is left to idle for 5 repetitions of the $380\,\mathrm{ns}$ delay instruction; (Orange) Pair $(0,1)$ idles for the same length of time while a CNOT is applied in parallel to pair $(3,5)$ for 5 repetitions (see~\Cref{fig:idlingPerthLayout} (b)); (Green) Pair $(0,1)$ idles while the $X$ gate is applied 5 times to all other qubits.}\label{fig:01idlingLindblads}
\end{figure}
\begin{figure}[htbp]
\centering
\includegraphics[width=\linewidth,trim={0 0 0 0},clip]{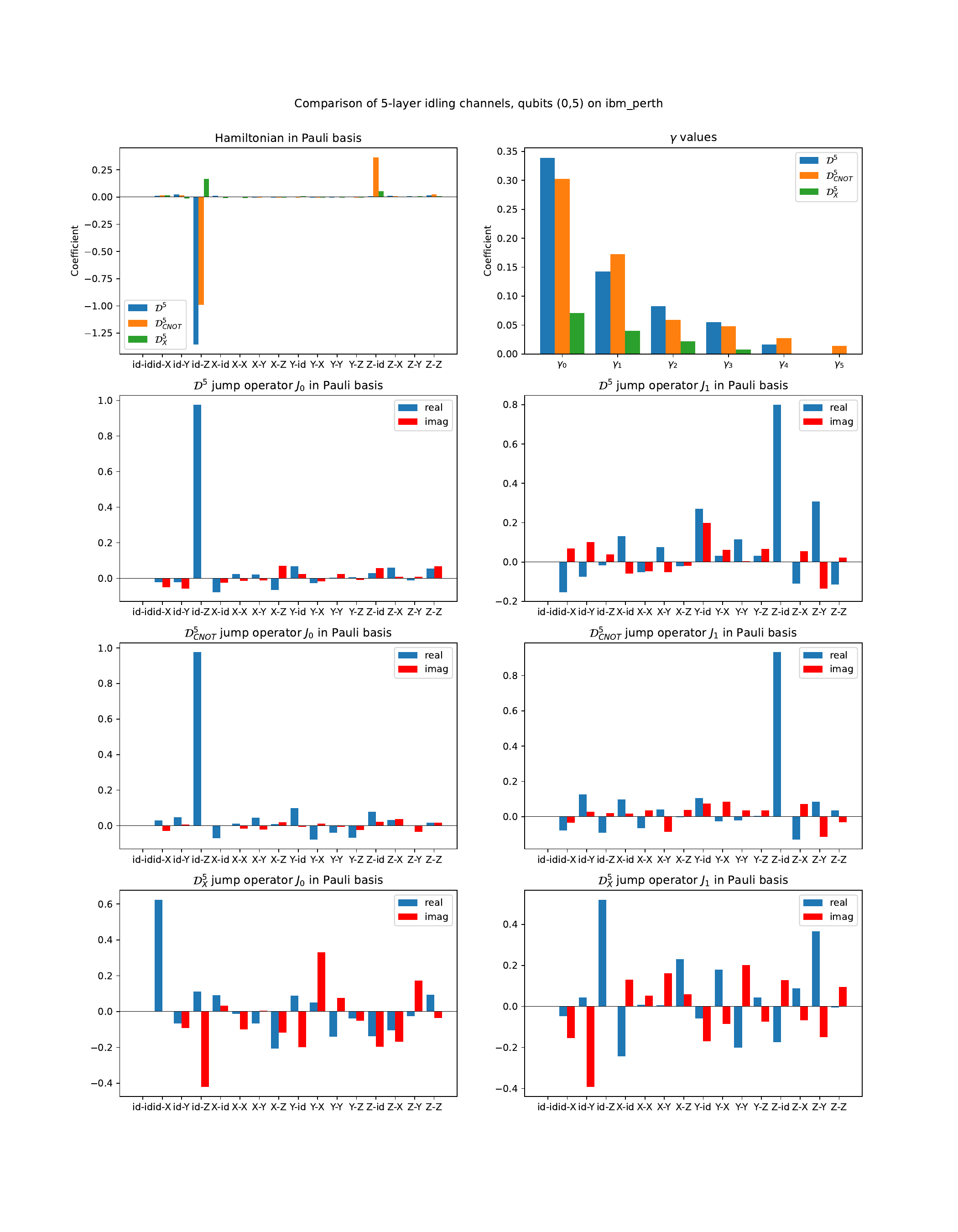}
\caption{Canonical decompositions (see~\Cref{subsec:Lindblad}) of idling processes for the separated qubit pair $(0,5)$ on \textit{ibm\_perth} (see~\Cref{fig:idlingPerthLayout}), estimated using Gate Set Flip-Flop. (Blue) The whole device is left to idle for 5 repetitions of the $380\,\mathrm{ns}$ delay instruction; (Orange) Pair $(0,5)$ idles for the same length of time while a CNOT is applied in parallel to pair $(2,3)$ for 5 repetitions (see~\Cref{fig:idlingPerthLayout} (c)); (Green) Pair $(0,5)$ idles while the $X$ gate is applied 5 times to all other qubits.}\label{fig:05idlingLindblads}
\end{figure}

\FloatBarrier

\section{Conclusions} \label{sec:conc}
In this work, we have presented two algorithms for the problem of fitting Lindbladian models to tomographic data: the Convex Solve method and the Alternating Projections method, which constitutes an improved variant of the logarithm search algorithms proposed in \cite{Onorati2023fittingquantumnoise}. By making use of an initial best guess for the model describing the target process, we avoid the global random basis search that was proposed in \cite{Onorati2023fittingquantumnoise} for the case where the ground truth channel has a degenerate spectrum, massively reducing the run-time for the logarithm search method in these difficult cases. 
On the other hand, we have argued that in many cases where gate noise is weak, we can sidestep the problems associated with degenerate spectrum altogether. In particular, we argue that problematic eigenspace splitting occurs only when some of the eigenvalues of the transfer matrix are real and negative. Otherwise, under modest assumptions that are often reasonable for quantum logic gates on current hardware, it is possible to find a good solution by Convex Solve on the principal branch alone (that is, with a single call to the SDP defined in~\cite{Onorati2023fittingquantumnoise}). We proved that this holds in the absence of statistical error and provided both numerical and experimental evidence that it remains an effective strategy in more realistic settings. Nevertheless, there exist important cases where eigenspace splitting does occur, such as CNOT, ISWAP, and tensor products of certain elementary gates. In these cases, we demonstrated numerically that the Alternating Projections method was effective for synthetic data with a variety of realistic gate noise models. We have therefore demonstrated that logarithm search methods for Lindbladian fitting can be made practical for the important case of dynamics induced by two-qubit quantum logic gates on noisy hardware. 

We also introduced the Gate Set Flip-Flop algorithm. This protocol augments the Lindbladian fitting procedure with techniques from gate set tomography to filter out state preparation and measurement errors. In our implementation, we use either Alternating Projections or Convex Solve as the Lindbladian fitting subroutine, depending on the type of process, but we note that any other Lindbladian fitting method such as direct numerical search~\cite{Samach_2022,dobrynin2024} could be substituted. We demonstrated that Flip-Flop paired with Alternation Projections/Convex Solve consistently improves the Lindbladian fit after a small number of iterations in the presence of realistic SPAM errors and statistical noise. Finally, we demonstrated that the algorithm can also be effective in practice for real experimental data, by applying it to tomographic data collected from IBM superconducting-qubit hardware~\cite{IBMQ}. 

Two avenues suggest themselves for further study. For simplicity, in this work we have considered individual tomographic snapshots in isolation, rather than analysing sequences of snapshots as a time-series. In particular, for gate sets that included repetitions of the same gate, or increased pulse duration of the same type of gate, our Gate Set Flip-Flop algorithm fits separate time-independent Lindbladians for each snapshot, whereas it would make sense to instead treat the snapshots as a time series generated by the same underlying Lindbladian. Indeed, the precursor work ~\cite{Onorati2023fittingquantumnoise} provided variants of the algorithm that could be applied to time series data for both time-independent and time-dependent Lindbladians. An obvious next step, then, would be to show how our Alternating Projections and Flip-Flop algorithms could be adapted to deal with time series data with possibly time-dependent Lindbladians. We leave this extension for future work.

Another limitation of our method as presented is that it is only practical for application to subsystems of small numbers of qubits. In this paper, we have focused on two-qubit gates. It is not necessarily completely infeasible to apply the method without modification to a three-qubit gate set, but it would certainly lead to much longer run-times, and going beyond this is likely not possible due to the exponential overhead in the number of parameters needed to fully characterise an $n$-qubit process. This limitation is shared by any fully general method for Lindbladian fitting, and indeed full QPT and GST. As noted in~\Cref{sec:related_works}, scalable methods for Lindbladian estimation do exist, but these must rely on strong assumptions about the noise model in order to reduce the number of parameters to be determined. A recent example is the local Pauli-Lindblad model employed in the protocols of~\cite{IBMPEC,IBMPEA}. An interesting direction for future work would be to investigate whether the techniques developed in the present work could be extended to larger systems by restricting the class of Lindbladians.

\section*{Acknowledgements}
This work was supported by Innovate UK [grant numbers 44167, 76963].

\section*{Ethics declaration}
The authors have filed a patent application EP25163986.0 on the novel methods for Lindbladian fitting introduced in this manuscript.
%========================================================================

\printbibliography

%========================================================================
%========================================================================
%========================================================================
\newpage

\phantomsection
\addcontentsline{toc}{section}{Appendix}
\appendix

\begin{center} {\LARGE \textbf{Appendix} } \end{center} 

\section{Prior work}\label{sec:prior work}
\subsection{Logarithm Search Methods}

Logarithm search methods \cite{Markdynamics08,CEW09,Onorati2023fittingquantumnoise} attempt to certify or estimate a Lindbladian model by taking the matrix logarithm of an input transfer matrix. These techniques must account for the fact that the complex matrix logarithm is not unique. For a transfer matrix $E$ with a non-degenerate spectrum there is a countably infinite family of matrices that exponentiate to the given $E$. In the degenerate case, there is a continuous degree of freedom in the choice of eigenbasis, which leads to an uncountable infinity of matrix logarithms. 
Therefore, assessing whether a channel is compatible with a time-independent Markovian channel means checking whether any branch of the matrix logarithm contains a matrix close to a Lindbladian. It was shown in \cite{Markdynamics08, CEW09} that for channels with non-degenerate spectrum, this problem can be cast as a mixed-integer semidefinite program. In principle, this can be solved efficiently for fixed Hilbert space dimension, though in practice searching over a small number of branches is often sufficient. These works did not deal with how to find the closest Markovian model for tomographic data where the noisy input channel is non-degenerate, but arises from a `true' channel with a degenerate spectrum. 

This case was tackled in \cite{Onorati2023fittingquantumnoise}, where matrix perturbation theory was used to develop explicit algorithms for Lindbladian estimation from noisy tomographic data.\footnote{\cite{Onorati2023fittingquantumnoise} also extended the method to deal with time series of tomographic snapshots and time-dependent Lindbladians.} It was shown that missing to consider that a non-degenerate channel can result from perturbation of a degenerate channel can lead to the algorithm of \cite{Markdynamics08, CEW09} failing. 
To counter this, \cite{Onorati2023fittingquantumnoise} developed a strategy to cluster eigenvalues according to some chosen precision to assess whether the true channel could be degenerate and construct a suitable eigenbasis accordingly. 
The algorithm then carries out convex optimisations to determine the closest Lindbladian over a finite number of branches of the matrix logarithm, constructed according to the chosen basis. However, in the degenerate case each choice of basis leads to a different set of matrix logarithms. Since the correct basis is not known a priori, the algorithm relies on repeated random sampling of bases. In principle, given enough random samples, one should eventually cover the parameter space with sufficient granularity to find a good approximation to the global minimum. In practice, the parameter space already becomes extremely large for the two-qubit case; for the example of a noisy ISWAP gate studied in \cite{Onorati2023fittingquantumnoise}, the run-time of the algorithm was around two weeks on a standard desktop computer. In fact, there exist certain important two-qubit cases where it is known that a good solution exists, but the optimisation landscape is particularly unforgiving, so that the algorithm is not able to find a close Markovian channel in any reasonable time, even allowing for heavy parallelisation. Another limitation of the method of \cite{Onorati2023fittingquantumnoise} is that its only mechanism to deal with SPAM noise is via an error tolerance parameter. Failure to properly account for SPAM can lead to apparent non-Markovianity or time-dependence in the input channel, preventing a good fit to a Lindbladian model and overestimating gate infidelity. Prior to the current work, logarithm search methods were untested on data collected from real quantum computing hardware.

\subsection{Direct Numerical Estimation}

A simple example of direct numerical estimation would be taking $f(x)=\norm{\exp(L(x))-E}$, where $L(x)$ is a Lindbladian parametrised by the optimisation variables $x$, $E$ is the input data, and a solver (e.g. MATLAB \texttt{fminsearch}) is called to minimise $f$ with respect to $x$. One key observation is that the set of all $d^2\times d^2$ Lindbladian matrices can be exactly parametrised by $d^2(d^2-1)$ real variables. Furthermore, gradients of the objective function can often be obtained mechanically via automatic differentiation, and in such cases, gradient-based solvers (e.g. MATLAB \texttt{fminunc}) can also be used. Finally, a good initial guess required by the solvers is often available based on the experimenter's prior knowledge about the input $E$. An early example of this approach estimated the Markovian evolution for a 2-qubit nuclear magnetic resonance (NMR) system~\cite{boulant2003robust} by minimising a time-series generalisation of the objective $f$ using MATLAB \texttt{fminsearch}, which implements the classic Nelder-Mead simplex search algorithm \cite{NelderMead1965}. More recently, state-of-the-art demonstrations of Lindbladian tomography have used an objective function based on maximum likelihood estimation (MLE)~\cite{Samach_2022,dobrynin2024}. This has been applied to a 2-qubit idling process on a superconducting-qubit device in \cite{Samach_2022}. Quantum process tomography is carried out for a series of different time delays, to build up a time series of tomographic snapshots. SPAM errors are filtered out by preparing a fiducial state and immediately measuring in the standard basis. The resulting SPAM-filtered channel estimates are then fit into a log-likelihood function that is maximised using MATLAB \texttt{fminsearch} to extract the noisy Lindblad generator of the idling channel. An evolution of this MLE technique was proposed in~\cite{dobrynin2024}. Under the assumption of small noise the log-likelihood function is linearised to simplify the task into a convex-optimisation problem. The procedure is also combined with the compressed sensing technique to reduce the number of required measurements. It is worth emphasising that despite being observed to be effective in practice in prior works, local direct numerical optimisation approaches (just like ours) are heuristic in nature and have little formal guarantee of returning optimal solutions.

\subsection{Ad-hoc Methods}

Another approach to Lindbladian estimation is to rely on some assumption about the system being investigated, so that the Lindbladian model is restricted in some way, and the number of parameters to estimate is reduced. Such techniques have been applied in NMR~\cite{NMRestimationChilds2001}, quantum optics~\cite{Quantum_optics_D_Ariano_2002} and superconducting-qubit~\cite{Postmarkovian_Zhang_2022} settings for small systems. We note that in the latter example the authors use Lindbladian fitting as one part of a protocol used to analyse non-Markovian dynamics via post-Markovian master equations~\cite{PMME_Shabani_2005}. Characterisation of such non-Markovian processes is beyond the scope of the present work. One restricted form of Lindbladian that has recently proved useful in the context of error mitigation on quantum computing devices with limited connectivity is the so-called sparse Pauli-Lindblad model~\cite{IBMPEC,IBMPEA}. The method assumes that a single layer of two-qubit gates has an associated Markovian noise process generated solely by local Pauli operators, where the locality is with respect to the connectivity of the device. This yields a Lindbladian that can be specified by a number of parameters that grows only linearly with the system size. The parameters can then be estimated efficiently by fitting the decay of Pauli observables after repeated applications of the gate layer. The scalability of this method makes it appealing for practical implementation, and indeed it has been demonstrated successfully on intermediate-sized superconducting-qubit devices as a step in probabilistic error cancellation~\cite{IBMPEC} and zero noise extrapolation~\cite{IBMPEA} error mitigation protocols. It should be noted, however, that the assumption of a Pauli-Lindblad model is a rather strong one, resting on the assumption that the noisy gate layer can be well-modelled by the ideal gate layer composed with a stochastic Pauli channel with local generators. The stochastic Pauli channel assumption was justified in~\cite{IBMPEC,IBMPEA} by the fact that gate layers are implemented under a particular randomised compiling protocol known as Pauli twirling~\cite{twirlingBennett_1996,twirlingknill2004faulttolerant,twirlingKern_2005,twirlingGeller_2013,twirlingWallman_2016}. This, as well as the noise characterisation step, in turn relied on the fact that the two-qubit gates being analysed were Clifford gates, and so preserve the Pauli group. These considerations mean that the method of~\cite{IBMPEC,IBMPEA} is only applicable to particular quantum computing devices with single- and two-qubit gates that are already relatively well calibrated, in which case it can provide a good effective noise model that could be used in certain error mitigation techniques. It is not aimed at estimating the underlying noise process for specific few-qubit processes.

\section{Lindbladian Canonical Decomposition}\label{app:lindblad_decomp}
Given a Lindbladian $L$ in the elementary basis representation, it is possible to solve for a canonical decomposition comprised of a Hamiltonian $H$, jump operators $\{J_\alpha\}$, and real positive scalars $\{\gamma_\alpha\}$ such that the elementary basis representation of the Lindbladian form (see \Cref{eq:Lindbladian_form}) defined by $H$, $\{J_\alpha\}$, and $\{\gamma_\alpha\}$ is $L$. Let $n=\log_2(d)$ and let $P_1,\ldots,P_{d^2-1}$ denote the non-identity $n$-qubit Pauli operators. Since $H$ is traceless, $H$ is determined by its expansion in the Pauli basis
$$H=\sum_{k=1}^{d^2-1}\frac{\tr(P_kH)}{d}P_k.$$
Using the formula \Cref{eq:Lindblad_matrix}, we get
$$L=i(I\otimes H^T-H\otimes I)+\sum_{\alpha}\gamma_\alpha\left[J_\alpha\otimes \overline{J_\alpha}-\frac{1}{2}(J_\alpha^\dagger J_\alpha\otimes I+I\otimes J_\alpha^T\overline{J_\alpha})\right].$$
By the fact that $H$ and $\{J_\alpha\}$ are traceless, for every $k\in\{1,\ldots,d^2-1\}$,
$$\tr(L(P_k\otimes I))=-i\tr(P_kH)d-\frac{1}{2}\sum_{\alpha}\gamma_\alpha\tr(P_kJ_\alpha^\dagger J_\alpha)d.$$
Thus,
\begin{align*}
&\tr(L(P_k\otimes I))-\overline{\tr(L(P_k\otimes I))}=-2i\tr(P_kH)d\\
\implies&\tr(P_kH)=\frac{i(\tr(L(P_k\otimes I))-\overline{\tr(L(P_k\otimes I))})}{2d}.
\end{align*}
Similarly, for every $j,k\in\{1,\ldots,d^2-1\}$,
$$\tr(L(P_j\otimes \overline{Q_k}))=\sum_{\alpha}\gamma_\alpha\tr(P_jJ_\alpha)\tr(Q_kJ_\alpha^\dagger)$$
and we can organise these values into a $(d^2-1)\times(d^2-1)$ matrix
$$C=\begin{bmatrix}\tr(L(P_1\otimes \overline{Q_1})) & \cdots & \tr(L(P_1\otimes \overline{Q_{d^2-1}}))\\
\vdots & \ddots & \vdots\\
\tr(L(P_{d^2-1}\otimes \overline{Q_1})) & \cdots & \tr(L(P_{d^2-1}\otimes \overline{Q_{d^2-1}}))\end{bmatrix}.$$
Since the jump operators are orthonormal, 
$$C=\sum_{\alpha}\gamma_\alpha\begin{bmatrix}\tr(P_1 J_\alpha) \\ \vdots \\ \tr(P_{d^2-1} J_\alpha)\end{bmatrix}\begin{bmatrix}\tr(P_1 J_\alpha^\dagger) & \cdots & \tr(P_{d^2-1} J_\alpha^\dagger)\end{bmatrix}$$
is a spectral decomposition of $C$. Thus, the positive scalars $\{\gamma_\alpha\}$ and the jump operators $\{J_\alpha\}$ can be retrieved by computing a spectral decomposition of $C$.

\section{Analysing the Convex Solve Method for the Lindbladian Fitting Problem Without Statistical Error}\label{app:trivial}

Recall that the input to the Lindbladian fitting problem is a $d^2\times d^2$ transfer matrix $E$, and the goal is to search for an $L\in\mathcal{L}$ that minimises $\lVert e^L-E\rVert$ where $\mathcal{L}$ is the set of all Lindbladians. In \Cref{subsec:Lindblad}, we have restricted ourselves to special cases of the problem satisfying two additional assumptions. Namely, we assume there exists an unknown Lindbladian $L^*$ such that for $E^*=e^{L^*}$, $\lVert E-E^*\rVert\leq c_1$ for some small constant $c_1$, and there exists a known Lindbladian $L^{\text{ideal}}$ such that $\lVert L^*-L^{\text{ideal}}\rVert\leq c_2$ for some small constant $c_2$. In this appendix, we further assume $E=E^*$ (i.e. $c_1=0$), so the goal becomes outputting an $L\in\mathcal{L}\cap\log(E)$ which need not necessarily be $L^*$ and where $\log(E)=\{A:e^A=E\}$.

We first establish a key lemma. For every diagonalisable matrix $A$, let $\rho(A)=\max\{|\lambda|:\text{$\lambda$ is an eigenvalue of $A$}\}$ denote its spectral radius. For every invertible matrix $B$, let $\kappa(B)=\lVert B\rVert\lVert B^{-1}\rVert$ denote its condition number.

\begin{lemma}
Write $L^{\text{ideal}}=VDV^{-1}$ for some diagonal matrix $D$ and invertible matrix $V$. If $\lVert L^*-L^{\text{ideal}}\rVert<\frac{\pi-\rho(L^{\text{ideal}})}{\kappa(V)}$, then for every eigenvalue $\mu$ of $L^*$, $|\Im(\mu)|<\pi$.
\label{lemma:trivial}
\end{lemma}
\begin{proof}
Let $\mu$ be an eigenvalue of $L^*$. By the Bauer-Fike theorem, there exists an eigenvalue $\lambda$ of $L^{\text{ideal}}$ such that
\begin{equation}
|\mu|-|\lambda|\leq |\mu-\lambda|\leq \kappa(V)\lVert L^*-L^{\text{ideal}}\rVert<\pi -\rho(L^{\text{ideal}}).
\end{equation}
Thus,
\begin{equation}
|\Im(\mu)|\leq |\mu|< \pi-\rho(L^{\text{ideal}})+|\lambda|\leq \pi.
\end{equation}
\end{proof}

Note that for $L^\text{ideal}$ generating a unitary gate, all the eigenvalues of the ideal transfer matrix $E^\text{ideal}=e^{L^{\text{ideal}}}$ have modulus $1$, which implies $\rho(L^\text{ideal})\leq \pi$ for $L^\text{ideal}$ in the principal branch. It is worth emphasising that \Cref{lemma:trivial} is not asymptotic in nature, and its assumption, which only depends on $L^\text{ideal}$, holds for realistic values of $c_2$ for many practically relevant $L^{\text{ideal}}$. For example, for the identity gate with $L^\text{ideal}=0$, we have $\rho(L^{\text{ideal}})=0$ and $\kappa(V)=1$, so $c_2<\pi$ suffices. As another example, for the $T\otimes I$ gate, we have $\rho(L^{\text{ideal}})=\frac{\pi}{4}$ and $\kappa(V)=1$, so we can choose $c_2<\frac{3}{4}\pi$. Unfortunately, for both CNOT and ISWAP, we have $\rho(L^{\text{ideal}})=\pi$, so \Cref{lemma:trivial} does not apply. 

When $E$ has degenerate eigenvalues, a degenerate eigenspace of $E$ could split into different eigenspace of $L^*$ due to the non-uniqueness of complex logarithms, which induces continuous degrees of freedom in specifying an $A\in\log(E)$. We call this phenomenon eigenspace splitting. It should be clear that conditioned on eigenspace splitting not occurring for $L^*$, the Convex Solve method will eventually locate $L^*$ in $\log(E)$ since in this case, the Convex Solve method correctly enumerates over the remaining discrete degrees of freedom. Our main observation states that if the assumption of \Cref{lemma:trivial} holds, then eigenspace splitting cannot occur for $L^*$. We can now prove \Cref{thm:trivial} from the main text, which stated that provided $L^*$ and $L^{\text{ideal}}$ satisfy $\lVert L^*-L^{\text{ideal}}\rVert<\frac{\pi-\rho(L^{\text{ideal}})}{\kappa(V)}$, then given the exact transfer matrix $E=\exp{L^*}$, the Convex Solve method will output a solution $L\in\mathcal{L}\cap\log(E)$.

\begin{proof}[Proof of \Cref{thm:trivial}]
First consider the case where $E$ has $d^2$ distinct eigenvalues. Then \Cref{lemma:trivial} directly implies that $L^*$ is the principal branch of $\log(E)$. 

Now suppose $E=E^*$ contains degenerate eigenvalues. Suppose there exist $a,b\in\RR$ and integer $k$ with $|k|\geq 1$ such that $L^*$ has two eigenvalues $a+bi$ and $a+bi+2\pi k i$ which both exponentiate to the eigenvalue $e^{a+bi}$ for $E$. Then either $|b|\geq\pi$ or $|b+2\pi k|\geq\pi$, which contradicts \Cref{lemma:trivial}. Therefore, $L^*$ and $E$ have the same set of eigenspace projectors and $L^*$ is the principal branch of $\log(E)$.
\end{proof}

Although \Cref{thm:trivial} does not apply if $E^{\text{ideal}}$ has real negative eigenvalues, it remains true that if eigenspace splitting does not occur for $L^*$, then the Convex Solve method succeeds on $E$. Next, we describe two mechanisms by which eigenspace splitting could occur when $E^{\text{ideal}}$ contains real negative eigenvalues for arbitrarily small $\lVert L^*-L^{\text{ideal}}\rVert$. 

Consider a situation where $-1$ is a degenerate eigenvalue of $E^{\text{ideal}}$ with degeneracy at least $4$ (for example, the ideal transfer matrix of CNOT has a $6$-dimensional degenerate $-1$ eigenspace). Then $\pi i$ and $-\pi i$ are eigenvalues of $L^{\text{ideal}}$ each with degeneracy at least $2$. Then for every $\varepsilon>0$, it is possible for $L^*$ to have eigenvalues $(\pi+\varepsilon) i$, $(-\pi-\varepsilon)i$, $(\pi-\varepsilon)i$, and $(-\pi+\varepsilon) i$. Now, $(\pi+\varepsilon)i$ and $(-\pi+\varepsilon)i$ exponentiate to the same eigenvalue for $E^*$ and similarly for $(\pi-\varepsilon)i$ and $(-\pi-\varepsilon)i$. Thus, a degenerate eigenspace of $E^*$ corresponding to the eigenvalue $e^{(\pi+\varepsilon)i}$ gets split into different eigenspaces for $L^*$ corresponding to the eigenvalues $(\pi+\varepsilon)i$ and $(-\pi+\varepsilon)i$. Note that since $\varepsilon$ is arbitrary, this can happen for arbitrary $\lVert L^*-L^{\text{ideal}}\rVert$, so no condition on $\lVert L^*-L^{\text{ideal}}\rVert$ alone can guarantee eigenspace splitting does not occur for noisy CNOT gates. We note that the behaviour described in this paragraph has never been observed in practice during our numerical testing. 

Another failure mode for the Convex Solve method occurs when $E$ persists to have real negative eigenvalues. This can happen when precisely one out of the two qubits is noisy. We view such scenarios to be highly non-generic. 

In the real world, the input $E$ obtained from quantum process tomography suffers from finite statistical error controlled by the number of sampling shots. In particular, $E\neq E^*$ and $E$ will almost never correspond to a Markovian channel. At best, we can assume $E$ is the transfer matrix of a CPTP channel since it is cheap to project $E$ to be CPTP. Hence, the actual input $E$ to our problem is always the transfer matrix of a non-Markovian channel, and the goal is to find a Lindbladian $L$ that best fits $E$ with the promise that $E$ is close to some unknown Markovian channel $E^*$. The presence of statistical noise in the input $E$ is necessary for the Lindbladian fitting problem to be non-trivial for the family of Markovian channels we consider as statistical noise causes the input to be non-Markovian.

\section{Detailed Description of Alternating Projections}\label{app:TTEAPtech}
\SetKwFunction{FnAP}{APCore}
In this appendix, we give a more complete description of our Alternating Projections method for solving the Lindbladian fitting problem. Here we report the specific implementation used to obtain the numerical results given in the main text, but we note that our approach is heuristic, and it is not necessary to be prescriptive about technical details such as the random perturbation and logarithm search steps. For clarity we separate out the outer loop (\cref{alg:APouter}) that handles this preprocessing of the input tomographic data, and the core subroutine (\cref{alg:APcore}) that alternates between projections of the Lindbladian eigenvectors onto the approximate eigenspaces of the tomographic data, and projection of the resultant approximate matrix logarithm $\bar{A}$ onto the set of Lindbladians. We assume the input transfer matrix $E$ and initial guess Lindbladian $L_0$ are both diagonalisable. In the presence of statistical error, the eigenvalues of $E$ will be distinct, but this is not a requirement of the algorithm.

\begin{algorithm}\small
 \SetKwInOut{Input}{input}
 \SetKwInOut{Output}{output}
 \SetKw{KwMin}{minimise}
 \SetKw{KwSubjTo}{subject to}
 \SetKw{KwBreak}{break}
\SetKwFunction{FnAP}{APCore}
 \Input{Estimated (diagonalisable) transfer matrix $E$; initial guess Lindbladian $L_0$; precision $\beta\geq0$; maximum search depth $T$; number of random starts $N$}
 \Output{Estimated Lindbladian $L^{\text{est}}$ such that $\exp L^{\text{est}} \approx E $}
Compute set of eigenvalue-eigenvector pairs $\{(\mu_j\, ,\Ket{r_j})\}_{j=1}^{d^2}$ for $E$, and compute set of corresponding left eigenvectors $\{\Bra{l_j}\}_j$ such that $\BraKet{l_i}{r_j} = \delta_{i,j}$\\
 \For{$ j = 1$ \KwTo $d^2$}{$\lambda_j \gets$ logarithm of $\mu_j$ such that $|\Im(\lambda_j)|\leq\pi$}
\label{step:cluster}Perform eigenvalue clustering: $\mu_i, \mu_j$ are taken to be in the same cluster if $|\mu_i-\mu_j|\leq\beta$. Suppose $n$ clusters are detected. The $k$-th cluster is defined by an index set $C_k = \{i_{k,1},\ldots,i_{k,w_k}\}$, where $w_k =|C_k|$ corresponds to the rank of the associated subspace so that for each pair
$(\mu_j\, ,\Ket{r_j})$, the index $j$ is in exactly one of the index sets\\
For each $k$, let $\Pi_k = \sum_{j \in C_k} \KetBra{r_j}{l_j}$\\
\label{step:constructM}Enumerate hermiticity-compatible branches of the matrix logarithm $M$\\ \tcp{See main text for construction of $M$} 
 \For{$m \in M$}{
	$\hat{\lambda}_{i_{k,a}} \gets \lambda_{i_{k,a}} + 2 \pi m_{i_{k,a}} i$\label{step:eigshift}\\
	\For{$p = 1$ \KwTo $N$}{
	Generate random perturbation $\tilde{L}_0$ of $L_0$ \label{step:perturb}\\ \tcp{in our implementation, we alternate between $\tilde{L}_0=L_0+D$ and $\tilde{L}_0=L_0+H^{\otimes d}DH^{\otimes d}$ for random diagonal matrix $D$ with small norm.}
	\label{step:APcore}$L_{m,p} \gets $ \FnAP{$\tilde{L}_0$, $T$, $\{ \hat{\lambda}_j \}_{j=1}^{d^2}$, $\{\Pi_k\}_{k=1}^{n}$, $\{C_k \}_{k=1}^{n}$}\\
	}
}
$L_{\text{est}} \gets \text{argmin}_{L_{m,p}}\norm*{e^{L_{m,p}}-E}$\\
 \KwRet $L_{\text{est}}$
 \caption{Alternating Projections outer loop}\label{alg:APouter}
\end{algorithm}

\begin{algorithm}\small
\SetKwFunction{FnAP}{APCore}
 \SetKwInOut{Input}{input}
 \SetKwInOut{Output}{output}
 \SetKw{KwMin}{minimise}
 \SetKw{KwSubjTo}{subject to}
 \SetKw{KwBreak}{break}
\SetKwProg{Fn}{Function}{:}{}
 \Input{Approximate Lindbladian matrix $\tilde{L}_0$; maximum iteration depth $T$; list of log-eigenvalues $\hat{\lambda}_j$ of tomographic data $E$; set of $n$ orthogonal projectors $\Pi_k$ satisfying $\sum_{k=1}^{n} \mathrm{rank(\Pi_k)} = d^2$; set of index sets $\{C_k \}_{k=1}^n$ corresponding to eigenvalue clustering}
 \Output{A candidate Lindbladian $\tilde{L}_T$}
\Fn{\FnAP{$\tilde{L}_0$, $T$, $\{ \lambda_j \}_{j=1}^{d^2}$, $\{\Pi_k\}_{k=1}^{n}$, $\{C_k \}_{k=1}^{n}$}}{
	$\hat{D}\gets \mathrm{Diag}(\hat{\lambda}_1,\ldots,\hat{\lambda}_{d^2})$\\
	\For{$t=0$ \KwTo $T-1$}{
		Compute eigenvalue-eigenvector pairs $(\sigma_1,\Ket{\tilde{v}_1}),\ldots,(\sigma_{d^2},\Ket{\tilde{v}_{d^2}})$ for $\tilde{L}_t$\\
		Assign each $(\sigma_j,\Ket{\tilde{v}_j})$ to exactly one subspace projector $\Pi_k$ by solving minimum-cost maximum-flow problem to minimise $\lVert \Ket{\tilde{v}_j}-\Pi_k\Ket{\tilde{v}_j}\rVert$ for each pairing. For each $k$, let this assignment define an index set $C'_k = \{i'_{k,a} \}_{a=1}^{w_k}$, so that $j$ is in $C'_k$ if $\Ket{\tilde{v}_j}$ is assigned to the $k$-th subspace\label{step:mincostmaxflow}\;
		Initialise $d^2 \times d^2$ matrix $K$\\
		\For{$k=1$ \KwTo $n$}{
			Find minimum-cost perfect matching $W_k$ between data log-eigenvalues $\hat{\lambda}_{i_{k,a}}$ where ${i_{k,a}} \in C_k$  and Lindbladian eigenvalues $\sigma_{i'_{k,b}}$ where $i'_{k,b} \in C'_k$ according to cost function $\abs*{\hat{\lambda}_{i_{k,a}} - \sigma_{i'_{k,b}}}$\\
			For each $\{\hat{\lambda}_{i_{k,a}}, \sigma_{i'_{k,b}}\} \in W_k$, set the $i_{k,a}$-th column of $K$ to $\Pi_k\Ket{\tilde{v}_{i'_{k,b}}}$\\
		}
		$\bar{A} \gets K\hat{D}K^{-1}$\\
		Let $L_{\text{opt}}$ be the solution to the convex optimisation problem:\\
 			\quad \quad \quad \KwMin $\norm*{L - \bar{A}}$  \\
 			\quad \quad \quad \KwSubjTo $L$ Lindbladian \\
 		\uIf{
$\norm*{e^{L_{\text{opt}}}- E}<\norm*{e^{\tilde{L}_t} - E}$
 			}{
 				$\tilde{L}_{t+1} \gets L_{\text{opt}}$\\
 			}
 		\uElse{
 		$\tilde{L}_T \gets \tilde{L}_t$\\
           \KwBreak}
	}
\KwRet $\tilde{L}_T$
}
 \caption{Alternating Projections core subroutine}\label{alg:APcore}
\end{algorithm}
The enumeration of relevant branches of the logarithm in~\cref{step:constructM} of~\cref{alg:APouter} goes as follows.  We assume for simplicity that $|\Im(\lambda)|\leq\pi$ for every eigenvalue $\lambda$ of $L^{\text{ideal}}$ i.e. $L^{\text{ideal}}$ is in the principal branch of $\log(E^{\text{ideal}})$. We then start with an initial set of branches
$$M'=\{(m_1,\ldots,m_{d^2})\in\mathbb{Z}^{d^2}:m_j\in\{-1,0,1\}\text{ for every } j\in\{1,\ldots,d^2\}\},$$ where each entry corresponds to a shift in the $j$-th eigenvalue by $2 \pi m_j i$, as per~\cref{step:eigshift} of~\cref{alg:APouter}.
Recall that for $\hat{\lambda}_1,\ldots,\hat{\lambda}_{d^2}$ to be the set of eigenvalues of a $d^2\times d^2$ Lindbladian matrix, they need to come in complex conjugate pairs \cite{Onorati2023fittingquantumnoise}. $M$ is defined to be the subset of elements of $M'$ satisfying this condition. This constraint alone greatly reduces the number of branches of $M'$ to try, and the search over $m\in M$ is embarrassingly parallelisable.

The minimum-cost maximum-flow problem in~\cref{step:mincostmaxflow} of~\cref{alg:APcore} is formulated as follows. We construct a directed graph with $d^2+n+2$ vertices, where $n$ is the number of approximate eigenspaces determined by the eigenvalue-clustering procedure in~\cref{step:cluster} of~\cref{alg:APouter}. There is a vertex corresponding to each eigenvector $\Ket{\tilde{v}_j}$, $j\in\{1,\ldots,d^2\}$ and a vertex corresponding to each eigenspace projector $\Pi_k$, $k\in\{1,\ldots,n\}$. For convenience, we shall use $\tilde{v}_j$ and $\Pi_k$ to denote these vertices. In addition, there is a special source vertex $s$ and a special sink vertex $t$. Every directed edge in the graph has an integer capacity and an integer weight. For every $j\in\{1,\ldots,d^2\}$, there is an edge from $s$ to $\tilde{v}_j$ with capacity $1$ and weight $0$. For every $j\in\{1,\ldots,d^2\}$ and $k\in\{1,\ldots,n\}$, there is an edge from $\tilde{v}_j$ to $\Pi_k$ with capacity $1$ and weight\footnote{The purpose of the constant $10^4$ here is to amplify differences between floating point values so that any pair of edges for which $\lVert \tilde{v}_j-\Pi_k\tilde{v}_j\rVert$ differs will have distinct integer weights. This is to avoid numerical problems that can afflict flow problem solvers when edges have floating point weights. The precise value of the large constant is unimportant.} given by $\lfloor 10000\lVert \tilde{v}_j-\Pi_k\tilde{v}_j\rVert\rfloor$.  For every $k\in\{1,\ldots,n\}$, there is an edge from $\Pi_k$ to $t$ with capacity the rank of $\Pi_k$ and weight $0$. The capacity constraints on the last set of edges encode the requirement that an eigenspace projector $\Pi_k$ with rank $l$ needs to accept exactly $l$ eigenvectors. Then, an assignment of the eigenvectors to the eigenspace projectors that minimises the sum of the weights can be read off from a minimum-cost maximum $(s,t)$-flow of the graph.

By diagonalising $L^{\text{ideal}}$ as $L^{\text{ideal}}=VDV^{-1}$, assuming $\lVert L^*-L^{\text{ideal}}\rVert<\frac{2\pi}{\kappa(V)}$, and together with the simplifying assumption that $L^\text{ideal}$ is in the principal branch of $\log(E^{\text{ideal}})$, the same argument in the proof of \Cref{lemma:trivial} implies it suffices to consider $|m_j|\leq 1$ for every $j\in\{1,\ldots,d^2\}$. The effect of the randomness introduced in~\cref{step:perturb} of~\cref{alg:APouter} is to induce a slightly different eigendecomposition of the initial guess Lindbladian in each call to \FnAP in~\cref{step:APcore}. This randomness in the starting point $\tilde{L}_0$ can lead to a different candidate Lindbladian model $\tilde{L}_T$ output for each call. Therefore we repeat this procedure $N$ times for each branch and accept the candidate model that best fits the tomographic data.

\end{document}